\documentclass{sig-alt-mod} 

\newcommand{\mysize}{\fontsize{10pt}{10pt}\selectfont}

\usepackage[boxruled,vlined]{algorithm2e}
\usepackage{graphicx,amssymb,wrapfig,makeidx,amsfonts,amsmath}

\renewcommand{\P}{\mbox{${\cal P}$}}

\newtheorem{ex}{EXAMPLE}[section]
\newenvironment{example}{\begin{ex} \nopagebreak
  \begin{rm}}{{\hfill$\Box$}\end{rm}\end{ex}} 

\newtheorem{defin}{Definition}[section]
\newenvironment{definition}[1]{\begin{defin}\begin{rm}({\bf #1})}{{\hfill$\Box$}\end{rm}\end{defin}}

\newtheorem{lemm}{Lemma}[section]
\newenvironment{lemma}{\begin{lemm}}{{\hfill$\Box$}\end{lemm}}

\newtheorem{thm}{Theorem}[section]
\newenvironment{theorem}{\begin{thm} \nopagebreak}{{\hfill$\Box$}\end{thm}}

\newtheorem{prop}{Proposition}[section]
\newenvironment{proposition}{\begin{prop}}{{\hfill$\Box$}\end{prop}}

\newtheorem{corol}{Corollary}[section]

\newtheorem{conjec}{Conjecture}[section]

\newsavebox{\savepar}

\newcommand{\x}[1]{\mbox{\ensuremath{#1}}}
\newcommand{\squishlist}{
  \begin{list}{$\bullet$}
   {
     \setlength{\itemsep}{0pt}
     \setlength{\parsep}{0pt}
     \setlength{\topsep}{0pt}
     \setlength{\partopsep}{0pt}
     \setlength{\leftmargin}{1.5em}
     \setlength{\labelwidth}{1em}
     \setlength{\labelsep}{0.5em} } }
\newcommand{\squishend}{
   \end{list}  }

\newcommand{\nop}[1]{}                       % To "hide" large text segments

\SetAlgoSkip{4cm}

\begin{document}

\title
{Obtaining Information about Queries \\ Behind Views and Dependencies 
}

\numberofauthors{2}

\author{
\alignauthor Rada Chirkova
\\
	\affaddr{Department of Computer Science} \\
	\affaddr{NC State University, Raleigh, NC 27695, USA} \\
	\email{chirkova@csc.ncsu.edu}
\alignauthor Ting Yu\\
	\affaddr{Department of Computer Science} \\
	\affaddr{NC State University, Raleigh, NC 27695, USA} \\
	\email{yu@csc.ncsu.edu}
}

\maketitle

{\mysize % mysize ghu

\begin{abstract} 
{\mysize % mysize bubul 
We consider the problems of  finding and determining certain query answers and  of determining containment between queries; each problem is formulated in presence of materialized views and dependencies under the closed-world assumption. We show a tight relationship between the problems in this setting. Further, we introduce algorithms for solving each problem for those inputs where all the queries and views are conjunctive, and the dependencies are embedded weakly acyclic \cite{FaginKMP05}. We also determine the complexity of each problem under the security-relevant complexity measure introduced in \cite{ZhangM05}. The problems studied in this paper are fundamental in ensuring correct specification of database access-control policies, in particular in case of fine-grained access control. Our approaches can also be applied in the areas of inference control, secure data publishing, and database auditing. 
} % end \mysize bubul
\end{abstract} 

\vspace{-0.2cm} 

\section{Introduction} 

In this paper, we consider the problems of finding and determining certain answers to relational queries, and of containment between relational queries. %(``Determining a certain query answer'' is the problem of determining whether a given tuple is a certain answer to a given query w.r.t. the given materialized views and, possibly, dependencies.) 
For the former two problems, we build on the setting of \cite{AbiteboulD98}, and for the latter -- on the setting of \cite{ZhangM05}; we point out and exploit a tight relationship between the settings. To begin with, in all these settings the set of databases (a.k.a. instances) of interest is not given directly, and is instead specified via a set of ``materialized views.'' That is, we are given definitions of one or more named queries {\em (definitions of views).} We are also given a set of answer tuples  for each view, that is, each view is {\em materialized} into a relation. Intuitively, each set $MV$ of materialized views specifies a set of ``base instances'' $I$, such that each relation in $MV$ can be obtained as an answer, on the instance $I$, to the respective view definition. In addition, for a given set of integrity constraints {\em (dependencies)} on the instances of interest, we deem relevant only those instances that satisfy all the dependencies. %(Example~\ref{intro-main-three-ex} provides an extended illustration.)  
In summary, we consider %both 
the problems of finding and determining certain query answers and the problem of query containment, each with respect to the sets of base instances % $I$, 
specified by a given set of materialized views %$MV$ 
{\em and} a given set of dependencies. 

The following motivating example draws on the area of database security called ``database-access control'' \cite{BertinoGK11}.  

\vspace{-0.1cm} 

\begin{example}
\label{intro-main-three-ex} 
Suppose a relation {\tt Emp} stores information about employees of a company, using attributes {\tt Name}, {\tt Dept} (department), and {\tt Salary}. Two other relations of interest are {\tt HQDept(Dept)} and {\tt OfficeInHQ} {\tt (Name,Office).} The relation {\tt HQDept} stores the names of the departments that are located in the company headquarters; {\tt OfficeInHQ} associates employees working in the headquarters with their office addresses. 

We now describe the integrity constraint {\em (dependency)} that holds on the database schema {\bf P} $=$ $\{$ {\tt Emp,} {\tt HQDept,} {\tt OfficeInHQ} $\}$. %We assume that no ``equality-generating dependencies'' \cite{AbiteboulHV95} hold on {\bf P}. (In particular, the only primary key of each relation in {\bf P} is the set of all its attributes.) At the same time, s
Suppose that for  all the departments located in the company headquarters, all their employees %of the departments 
have their offices in the headquarters. This can be expressed as a ``tuple-generating dependency'' \cite{AbiteboulHV95}, which we call $\sigma$. (Please see Example~\ref{intro-formalized-ex} for a formalization.) 

%\begin{tabbing} 
%$\sigma$: {\tt Emp(X,Y,Z)} $\wedge$ {\tt HQDept(Y)} $\rightarrow \exists${\tt S} \ {\tt OfficeInHQ(X,S)}. 
%\end{tabbing} 

Let a ``secret query'' \cite{MiklauS07} 
{\tt Q} ask for the names and salaries of all the employees who work in the company headquarters. We can 
formulate {\tt Q} in the standard relational query language SQL, as follows: 

\vspace{-0.1cm} 

{\small 
\begin{verbatim} 
(Q):SELECT DISTINCT Emp.Name, Salary FROM Emp, OfficeInHQ  
    WHERE Emp.Name = OfficeInHQ.Name; 
\end{verbatim} 
} % end \small 
 
 \vspace{-0.1cm} 

Consider three views, {\tt U,} {\tt V,} and {\tt W}, that are defined for
some class(es) of users, in SQL on the schema {\bf P}. The view {\tt U} returns the relation {\tt HQDept}, 
 the
view {\tt V} returns the department names for each employee, and {\tt W} returns the salaries in each department: 

\vspace{-0.1cm}

{\small
\begin{verbatim}
(U): DEFINE VIEW U(Dept) AS SELECT * FROM HQDept; 
(V): DEFINE VIEW V(Name,Dept) AS
     SELECT DISTINCT Name, Dept FROM Emp;
(W): DEFINE VIEW W(Dept,Salary) AS
     SELECT DISTINCT Dept, Salary FROM Emp;
\end{verbatim}
} % end \small

\vspace{-0.1cm}

Consider database users who are authorized to see only the answers to the views {\tt U}, 
{\tt V} and {\tt W}. (In particular, these users are not authorized to see any answers to the query {\tt Q}.) That is, {\tt U}, {\tt V}, and {\tt W} are {\em access-control views} for these users on the database with the schema {\bf P}. Suppose that at some point in time, these users can see
the following set $MV$ of answers to the views: 

\vspace{-0.2cm}

\begin{tabbing}
$MV$$=$$\{${\tt U(sales)},{\tt V(johnDoe,sales)},{\tt W(sales,50000)}$\}.$ 
\end{tabbing}

\vspace{-0.2cm}

%Consider a tuple $\bar t$ $=$ $(${\tt JohnDoe}, {\tt \$50000}$)$ of domain values in $MV$. 

A basic security question in this setting is as follows: Can these users find out which tuples must be in the answer 
to the query {\tt Q} on {\em all} the relevant ``back-end'' databases? If the answer to this question is positive, then, intuitively, there is a security breach, in the form of leakage   of some answers to {\tt Q} to unauthorized users. 

Let the back-end databases of interest be those instances of schema {\bf P} (i.e., ``base instances'') that satisfy the dependency $\sigma$ and that ``generate exactly $MV$ as answers to {\tt U}, {\tt V}, and {\tt W}.'' (The latter requirement is the ``closed-world assumption,'' to be discussed in detail shortly.)  
%Then it can be shown \cite{AbiteboulD98} that the tuple $\bar t$ is present in the answer to the secret query {\tt Q} on all such back-end databases. (All tuples with this property are called {\em certain answers} to the query {\tt Q} in presence of {\tt V,} {\tt W,} and $MV$.) 
%
Using an algorithm introduced in this paper, we can show that the tuple 
%
%\begin{tabbing} 
$\bar t$ $=$ $(${\tt johnDoe}, {\tt 50000}$)$ 
%\end{tabbing} 
%
\noindent 
is in the answer to the secret query {\tt Q} on all such base instances.  %(All tuples with this property are called {\em certain answers} to the query {\tt Q} in presence of {\tt V,} {\tt W,} and $MV$.) 
Thus, we can answer the above security question in the positive for the secret query {\tt Q} in the ``materi- alized-view setting'' {\em (}{\bf P}, $\{$$\sigma$$\}$, $\{${\tt U}, {\tt V}, {\tt W}$ \}$, $MV$). 
\end{example} 

\vspace{-0.1cm}

A tuple that is in the answer to the query of interest on all the relevant base instances is called a {\em certain answer} to the query. ``Determining a certain query answer'' is the problem of determining if a given tuple is a certain answer to a given query w.r.t. the given materialized views and, possibly, dependencies. (E.g., the tuple $\bar t$ in Example~\ref{intro-main-three-ex} is a certain answer to the query {\tt Q} in the %given 
setting ({\bf P}, $\{$$\sigma$$\}$, $\{${\tt U}, {\tt V}, {\tt W}$\}$, $MV$).)  The problems of finding and   determining  certain query answers on the instances defined by the given materialized views have been considered 
%The ``materialized-view'' setting, as described above, was considered in depth in the classic paper \cite{AbiteboulD98}, for the problem of finding certain query answers. (Example~\ref{intro-main-three-ex} provides the intuition for the problem). The study of \cite{AbiteboulD98} was performed 
both under the ``open-world assumption'' {\em (OWA)} and under the ``closed-world assumption'' {\em (CWA).}  %Informally, %\footnote{Please see Section~\ref{} for the formal definitions.} 
That is, for a base instance $I$, consider the answer tuples generated by the given view definitions on $I$. Then, informally, %the instance 
$I$ is relevant to the given instance $MV$ of materialized views under CWA  iff these answer tuples together comprise exactly $MV$. In contrast, OWA permits $MV$ to not contain all such tuples. %(It follows that for a given instance $MV$, each base instance that is relevant to $MV$ under CWA is also relevant to $MV$ under OWA, but the converse does not necessarily hold.)  
%OWA permits the materialized views in the given instance $MV$ to not contain all the tuples that are generated by the relevant view definitions on the relevant base instances. CWA is defined in a complementary way to OWA. 

The classic paper \cite{AbiteboulD98} by Abiteboul and Duschka addressed the complexity of % the problem of 
determining certain query answers under both OWA and CWA, for a range of query and view languages and in the absence of dependencies. %The paper 
\cite{AbiteboulD98} also provided algorithms for finding certain query answers under both OWA and CWA, for queries defined in datalog and for views in nonrecursive datalog,  with disequalities ($\neq$) permitted in both, again  in the absence of dependencies. The algorithms of \cite{AbiteboulD98}  are based on the ``conditional tables'' of \cite{ImielinskiL84}. The formulation of the latter problem was extended, in the context of database security, to account for database dependencies; the extended problem was solved  in \cite{BrodskyFJ00,StoffelS05}, under OWA for more restricted (than in \cite{AbiteboulD98}) languages of queries and views 
and for ``embedded dependencies'' \cite{AbiteboulHV95}.  

%In this materialized-view setting, we revisit the problem of finding certain query answers \cite{AbiteboulD98} and the problem of query containment \cite{ZhangM05}, by adding into consideration integrity constraints {\em (dependencies)} that might hold on the base instances of interest. 

Our original motivation for this current work comes from the fact that finding certain query answers is a basic problem in database security, as illustrated by Example~\ref{intro-main-three-ex}. Moreover, in its security form, this problem makes the most sense under CWA, rather than under OWA (see, e.g., \cite{RizviMSR04}). Intuitively, for those database attackers who are seeking unauthorized answers to a secret query, {\tt Q,} in presence of a set of view answers $MV$, the only relevant base instances are those that ``generate exactly $MV$,'' that is only the CWA-relevant instances. Suppose the owners of the back-end database run an {\em OWA-}based algorithm for finding certain answers to {\tt Q} w.r.t. $MV$.  They could then arrive at the empty set of answers (and thus conclude that their database is secure), even though under CWA, the set of certain answers to {\tt Q} would not be empty for the same $MV$ and dependencies. Indeed, we can use the results of \cite{BrodskyFJ00,StoffelS05} to show that in %the setting of 
Example~\ref{intro-main-three-ex}, the set of certain answers to the query {\tt Q} is the empty set under OWA.\footnote{The algorithms of \cite{AbiteboulD98} could not be applied to the problem instance of Example~\ref{intro-main-three-ex}, as that instance contains a dependency that is not a ``full''  \cite{AbiteboulHV95} (a.k.a. ``total'' \cite{BeeriV81}) dependency.} %At the same time, this basic problem has been addressed in the database-security literature, including in the presence of dependencies, only under OWA, please see \cite{BrodskyFJ00,StoffelS05}. (\cite{AbiteboulD98} provides algorithms for finding certain query answers for more expressive view and query languages than CQ queries and in the absence of dependencies, for both the OWA ans CWA cases. They also remark that the approach could be extended to the case of ``full dependencies''\cite{BeeriV81}. Please see Section~\ref{related-work-section} for a discussion of the relevant algorithms of \cite{AbiteboulD98}.) At the same time, 

To address this challenge, we have developed a CWA-based approach for finding certain answers to conjunctive queries {\em (CQ queries),} in presence of CQ views and of weakly acyclic embedded dependencies \cite{FaginKMP05}. (Similarly to the dependency-free case %for determining certain query answers 
\cite{AbiteboulD98}, the CWA version of this problem is harder %markedly more complex 
than the OWA version considered in \cite{BrodskyFJ00,StoffelS05}.) We then realized that our techniques can be connected to the solution of \cite{ZhangM05}, by Zhang and Mendelzon, to the (CWA-based) problem of determining containment between queries in presence of %answers to 
materialized views. %(\cite{ZhangM05} makes the CWA assumption about the base instances relevant to the given materialized views.) 
The latter problem arises, for instance, in determining whether a user query formulated on the base database relations has an equivalent rewriting in terms of the access-control views for this user. (If the answer to the question is positive, then the user query can be answered safely, %please 
see \cite{RizviMSR04,ZhangM05} for the details.) A natural and practically important generalization of this problem is its extension to the consideration of dependencies holding on the %schema of the 
base instances. 
%
%It turns out that for a given set $MV$ of materialized views, the problem of determining certain query answers w.r.t. $MV$ under CWA can be regarded as a special case of the problem of \cite{ZhangM05} of determining containment between queries w.r.t. $MV$, both in presence and in the absence of dependencies. (\cite{AbiteboulD98} makes the same connection for the dependency-free case under OWA.) Hence, w
We have been able  %, we were able 
to extend to the case of dependencies the algorithm of \cite{ZhangM05} for their query-containment problem w.r.t. materialized views $MV$, by building on our %proposed 
approach to the problem of finding certain query answers w.r.t. $MV$. 

\vspace{-0.1cm} 

\paragraph{Our contributions} 

%\vspace{-0.1cm} 

Our specific contributions are as follows: 
\begin{itemize} 
\vspace{-0.1cm} 
	\item We formalize the problem of determining containment between two queries under CWA and in presence of materialized views and dependencies, by building on the formalization of \cite{ZhangM05} that does not consider dependencies. 
	
	\vspace{-0.1cm} 

	\item We develop an algorithm for solving this problem, in the case where the input queries and all the view definitions are CQ queries, and the input dependencies are embedded weakly acyclic. 
	
	\vspace{-0.1cm}

	\item We show that the problem of determining certain answers to a query, under CWA and in presence of materialized CQ views and dependencies, is a special case of the above containment problem. It follows that the algorithm that we introduce for the containment problem also solves correctly this ``certain-query-answer'' problem, for all inputs where the queries and views are CQ queries and the dependencies  are embedded weakly acyclic. 
	
	\vspace{-0.1cm} 

	\item For the problem of {\em finding} all certain query answers under CWA and in presence of materialized views and dependencies, we develop two algorithms that are sound and complete for all inputs  where the queries and views are CQ queries and the dependencies are embedded weakly acyclic. The first algorithm uses as a subroutine our algorithm for the  ``certain-query-answer'' problem. The second algorithm both builds on the standard approach to answering queries in relational data exchange, and uses a simpler version of the technique that we use to solve the above {\em containment}  problem (w.r.t. materialized views and dependencies). 
	
	\vspace{-0.1cm} 

	\item We determine the complexity of each of the three problems under the security-relevant complexity measure of \cite{ZhangM05}. In this measure, it is assumed that everything is fixed except for the materialized views and queries (but not for view definitions). 
	
\end{itemize} 

%{\bf [[[ Where the results of this paper can be used: ]]]} 
%As a summary, in this paper we address the problems of  finding and determining certain query answers and  of determining containment between two queries. We consider each problem in presence of materialized views and dependencies under the closed-world assumption, and show a tight relationship between the problems in this setting. We introduce sound and complete algorithms for solving each problem in case where all the queries and views are conjunctive, and the dependencies are embedded weakly acyclic \cite{FaginKMP05}. We also determine the complexity of each problem under the security-relevant complexity measure introduced in \cite{ZhangM05}. 

\vspace{-0.1cm}

The problems that we study in the paper can be used to model and analyze a wide range of database-security problems, including database-policy analysis, secure data publishing, inference control, and auditing. For instance, database-security policies are often implemented through views. It is important to ensure that security views are defined correctly, so that no sensitive information can be learned by unauthorized parties from granted view access~\cite{BrodskyFJ00, YipL98}. Clearly, an information disclosure happens if an attacker can learn certain answers to a secret query. The same modeling can be applied to capture the secure data-publishing problem~\cite{MiklauS07}. Similarly, in database query auditing, answers to user-issued queries can be modeled as materialized views~\cite{MotwaniNT08, NabarKMM08}. A potential inference attack happens if those answers combined together can be used to derive secret information as defined by a query.

%The problems that we study in this paper are fundamental in ensuring correct specification of database access-control policies, in particular in case of fine-grained access control. Our approaches can also be applied in the areas of inference control, secure data publishing, and database auditing. 

%{\bf [[[ Do I have space here to say how the remaining sections of the paper are organized? ]]]} 

The remainder of the paper is organized as follows. After discussing related work in Section~\ref{related-work-section}, we review the background definitions and results in Section~\ref{prelim-sec}, and then define our three problems of interest in Section~\ref{problem-statement-sec}. In Sections~\ref{query-cont-algor-sec}--\ref{main-vv-dexchg-sec} we introduce our approaches to solving the three problems in the CQ weakly acyclic setting. 
% the problems of ``${\cal M}\Sigma$-conditional'' query containment and of determining certain answers to a query; in Section~\ref{main-vv-dexchg-sec} we present our solution to the problem of finding all certain query answers. 
Finally, in Section~\ref{complexity-sec} we address the complexity of the three problems in the CQ weakly acyclic setting.

%\vspace{2cm} 

\vspace{-0.25cm} 

\section{Related Work} 
\label{related-work-section} 

The seminal paper \cite{AbiteboulD98} by Abiteboul and Duschka addressed the complexity of the problem of determining whether a tuple is a certain query answer in presence of materialized views (the ``certain-query-answer problem'') under both OWA and CWA, for a range of query and view languages and in the absence of dependencies. \cite{AbiteboulD98} also solved the problem of finding all certain query answers under both OWA and CWA, for datalog queries and for views in nonrecursive datalog,  with disequalities ($\neq$) permitted in both,  %again 
in the absence of dependencies. The algorithms of \cite{AbiteboulD98}  are based on the ``conditional tables'' of \cite{ImielinskiL84}. (See \cite{AbiteboulHV95,GrecoMS12} for detailed overviews of incomplete databases and of their representations, including conditional tables \cite{ImielinskiL84}.) It is remarked  in \cite{AbiteboulD98} that their algorithms %of \cite{AbiteboulD98} 
for finding certain answers could be extended to the case of ``full'' (or ``total'') dependencies \cite{BeeriV81,AbiteboulHV95}. In this current paper, we provide sound and complete algorithms for finding certain answers %this problem 
and for the certain-query-answer problem, under CWA for CQ queries and views and for ``weakly acyclic''  \cite{FaginKMP05} embedded dependencies, of which the class of full dependencies %\cite{AbiteboulHV95} 
is a proper subclass. We also address the complexity of both problems in this CQ weakly acyclic setting. 

The paper \cite{BrodskyFJ00} by Brodsky and colleagues 
introduced in the security context the problem of finding certain query answers under OWA and in presence of embedded dependencies, and proposed a sound and complete algorithm for the case where the queries and views are CQ queries expressible without joins. Then, Stoffel and colleagues in \cite{StoffelS05} made a connection between this problem and the techniques introduced in data exchange \cite{FaginKMP05,Barcelo09,LibkinDataExchange}, by developing (\cite{StoffelS05}) a data-exchange-based approach for finding certain query answers, under OWA for CQ views, UCQ queries (i.e., unions of CQ queries), and embedded dependencies. 

In this paper we extend the data-exchange approach of \cite{FaginKMP05}, also used in \cite{StoffelS05}, to solve the problem of finding certain query answers under CWA, 
%solve the CWA version of the problem, 
 for CQ queries and views in presence of weakly acyclic embedded dependencies. The approaches of this current paper do not use ``target-to-source dependencies'' introduced in the data-exchange context in \cite{FuxmanKMT06}. The dependencies in our approaches do use constants (as was suggested back in \cite{FaginKMP05}), and thus are related to ``conditional dependencies'' %The dependencies (using constants) that we use in this paper are not ``conditional dependencies'' as discussed in, e.g., 
\cite{FanG12}. Conditional dependencies are intuitively understood as enforcing a (perhaps constant-involving) pattern onto (typically constant-determined) subsets %of the tuples 
of the given relations. As the dependencies that we use do not have constants in their antecedents, they intuitively behave in the ways expected of standard (constant-free) embedded dependencies. %, as in their extension to constant that was suggested in \cite{FaginKMP05}.  

While the term ``data exchange'' is mentioned in the paper \cite{MiklauS07} by Miklau and Suciu, data-exchange methods are not used in the technical development in \cite{MiklauS07}. Rather, the term %``data exchange'' 
is used in \cite{MiklauS07} informally as a reference to today's  universal sharing of data (as in, e.g., on the Web). %Moreover, 
\cite{MiklauS07} addresses the problem of ``data publishing,'' in which the goal is to determine, for a given set of view definitions and for a ``secret query'' $Q$, whether any materializations of the given views would disclose information about any answers to $Q$. (In contrast, in the three problems considered in this current paper, we assume that a specific set of view materializations is provided in the problem input.) Further, the notion of disclosure in \cite{MiklauS07}, inspired by Shannon's notion of perfect secrecy \cite{Shannon49}, is as follows: There is no disclosure of  query $Q$ via views $\cal V$ if and only if the {\em probability} of an adversary guessing the answer to $Q$ is the same (or, in another scenario, is almost the same) whether the adversary knows the answers to $\cal V$ or not. In this current paper, we %address a different problem: We 
use a deterministic, rather than probabilistic, notion of disclosure of a query answer, in presence %(rather than in the absence) 
of a specific set of view materializations; this leads to different security decisions than those following from \cite{MiklauS07}. 

%, as follows. Suppose the approach of \cite{MiklauS07} determines, for a given $Q$ and set $\cal V$, that there is probabilistic ``information leakage'' of $Q$ via $\cal V$, intuitively due to some materialization $MV$ of $\cal V$. It is still possible that another materialization, $MV'$, of $\cal V$ does not disclose deterministically (due to the algorithms introduced in this current paper) any answers to the ``secret query'' $Q$. In those cases where $MV'$ is the actual set of view answers over the current enterprise database, then $MV'$ can still be disclosed to users, even though the approach of \cite{MiklauS07} rules out the posssibility. 

The work \cite{ZhangM05} by Zhang and Mendelzon introduced and solved the problem of ``conditional containment'' between two CQ queries in presence of materialized CQ views, under CWA and in the absence of dependencies. \cite{ZhangM05} also introduced a  security-relevant complexity metric, under which their problem is $\Pi^p_2$ complete. %ness of %analyze the complexity of 
%their problem. 
(\cite{ZhangM05} also provides an excellent overview of the connections of the query-containment problem of \cite{ZhangM05} to database-theory literature.) %, including {\bf [[[ See \cite{ZhangM05} for a couple of pointers. ]]]}.  
In our work, we add dependencies to the formulation of the problem of \cite{ZhangM05}, and extend the approach of \cite{ZhangM05}, both to solve the resulting problem in the CQ weakly acyclic setting and to analyze the complexity of the problem. We also uncover a tight relationship of the problem with the problems of finding and determining certain query answers, under CWA in presence of view materializations and of dependencies.

% end \nop aj-kakaj 

\vspace{-0.25cm}

\section{Preliminaries} 
\label{prelim-sec} 

\vspace{-0.2cm}

\subsection{Instances and Queries} 
\label{inst-query-sec} 

%In this subsection and in Section~\ref{dexchg-prelim-sec}, we generally use the definitions and formal results as reported in \cite{Barcelo09}, which is an excellent overview of relational data exchange. 

{\bf Schemas and instances.} A {\em schema} {\bf P} is a finite sequence $<$ $P_1$, $\ldots,$ $P_m$ $>$ of relation symbols, with each $P_i$ having a fixed arity $k_i$ $\geq$ $0$. An {\em instance} $I$ {\em of} {\bf P} assigns to each $P_i$ $\in$ {\bf P} a finite $k_i$-ary relation $I[P_i]$, which is a set of tuples. For tuple membership in relation $I[P_i]$, we use the notation $\bar t$ $\in$ $I[P_i]$. Each element of each tuple in an instance belongs to one of two disjoint infinite sets of values, {\sc Const} and {\sc Var}. We call elements of {\sc Const} {\em constants,} and denote them by lowercase letters $a$, $b$, $c$, $\ldots \ $; the elements of {\sc Var} are called {\em (labeled) nulls,} denoted by symbols $\perp$, $\perp_1$, $\perp_2$, $\ldots \ $. %We will use the symbol $\epsilon$ to denote empty instances. 
%Such instances are known as ``V-tables'' \cite{ImielinskiL84}, the intuition being that all nulls in an instance are variables ranging over {\sc Const}. 

%It is often convenient to define instances by simply listing the tuples attached to the corresponding relation symbols. 
Sometimes we use the notation $P_i({\bar t})$ $\in$ $I$ instead of $\bar t$ $\in$ $I[P_i]$, and call $P_i({\bar t})$ a {\em fact} of $I$. When all the values in ${\bar t}$ are constants, we say that $P_i({\bar t})$ is a {\em ground fact,} and $\bar t$ is a {\em ground tuple.} The {\em active domain} of instance $I$, denoted $adom(I)$, is the set of all the elements of {\sc Const} $\cup$ {\sc Var} that occur in any facts in $I$. When each fact in $I$ is a ground fact, we call $I$ a {\em ground instance.} %End users of relational data are interested only in ground instances; non-ground instances are typically used as formal technical tools, as in, e.g., data exchange \cite{Barcelo09}. 

{\bf Queries.} We consider the class of queries called ``unions of conjunctive queries with disequalities,'' {\em $UCQ^{\neq}$ queries.} In the definitions for $UCQ^{\neq}$ queries, we will use  the following notions of {\em relational atom} and of {\em (dis)equality atom.} Let {\sc Qvar} be an infinite set of values disjoint from {\sc Const} $\cup$ {\sc Var}; we call {\sc Qvar} {\em the set of (query) variables.} We will denote variables by uppercase letters $X$, $Y$, $\ldots \ $ . Then $P({\bar t})$, with $P$ a $k$-ary relation symbol and $\bar t$ a $k$-vector of values, is a {\em relational atom} whenever each value in $\bar t$ is an element of {\sc Const} $\cup$ {\sc Qvar}. Further, an {\em equality (resp. disequality) atom} is a built-in predicate of the form $S$ $\theta$ $T$, where $\theta$ is $=$ (resp. $\neq$), and each of $S$ and $T$ is an element of {\sc Const} $\cup$ {\sc Qvar}. 

A $CQ^{\neq}$-{\em rule} over schema {\bf P}, with $k$-ary ($k$ $\geq$ $0$) {\em output relation symbol} $Q$ $\notin$ {\bf P}, 
%with arity $k$ $\geq$ $0$ of $Q$, 
is an expression of the form  %Denote by {\sc Qvar} an infinite set of values that is disjoint from {\sc Const} $\cup$ {\sc Var}, we call elements of {\sc Qvar} {\em (query) variables.} 

\vspace{-0.1cm} 

\begin{tabbing} 
$Q({\bar X}) \ \leftarrow \ P_1({\bar U}_1) \wedge \ldots \wedge P_n({\bar U}_n) \wedge C .$  
\end{tabbing} 

\vspace{-0.1cm} 

\noindent 
Here, $n$ $\geq$ $1$; the vector $\bar X$ has $k$ elements; for each $i$ $\in$ $[1,$ $n]$, $P_i$ $\in$ %is a relation symbol in 
{\bf P}; %, of arity $k_i \geq 0$; %each ${\bar U}_i$ is a vector of $k_i$ elements, 
%the atoms 
each of $Q({\bar X})$, $P_1({\bar U}_1),$ $\ldots,$ $P_n({\bar U}_n)$ is a relational atom; and  $C$ is a (possibly empty) finite conjunction of disequality atoms. %
%Each element of $\bar X$ and of each of ${\bar U}_1$, $\ldots$, ${\bar U}_n$ is either a constant (i.e., an element of {\sc Const}), or belongs to a {\em set} {\sc Qvar} {\em of (query) variables,} where {\sc Qvar} is an infinite set of values that is disjoint from {\sc Const} $\cup$ {\sc Var}.  We denote variables by uppercase letters $X$, $Y$, $\ldots \ $ . 
%
We consider only {\em safe} rules: That is, each variable in $\bar X$, as well as each variable occurring in $C$,  
also occurs in at least one of ${\bar U}_1$, $\ldots$, ${\bar U}_n$. All the variables of the rule that do not appear in $\bar X$ (i.e., the {\em nonhead variables} of the rule) are assumed to be existentially quantified. We call the atom $Q({\bar X})$ the {\em head} of the rule, call $\bar X$ the {\em head vector} of the rule, and call the conjunction of its remaining atoms the {\em body} of the rule. Each atom in the body of a rule is called a {\em subgoal} of the rule. The conjunction in the body %of a rule 
is usually written using commas, as $P_1({\bar U}_1), \ldots, P_n({\bar U}_n), C .$ 

% is given in a {\em normalized form,} that is $C$ does not include any explicit equality atoms. (Every $CQ^{\neq}$-rule can be converted in polynomial time into an equivalent normalized form, which is unique up to isomorphism, by using the same var) 

A {\em conjunctive query with disequalities (a $CQ^{\neq}$ query)} is a query defined by a single $CQ^{\neq}$-rule; a {\em conjunctive query (a CQ query)} is a $CQ^{\neq}$ query with an empty $C$. We will be referring to a $CQ^{\neq}$ query %defined by a $CQ^{\neq}$-rule 
with head $Q({\bar X})$ %a $k$-ary ($k$ $\geq$ $0$) output relation symbol $Q$ and with head vector $\bar X$. Wherever clear from the context, 
%we will be referring to the query 
as just $Q({\bar X})$, or even $Q$, whenever clear from the context. We will be using %$head_{(Q)}$ and 
$body_{(Q)}$ as a concise name %for the head and 
for the body of the (rule for) $Q$. 

Finally, for a $k$-ary relation symbol $Q$, with $k$ $\geq$ $0$, let ${\cal S}(Q)$ $=$ $\{$ $Q^{(1)}$, $\ldots,$ $Q^{(l)}$ $\}$ be a finite set of %nonisomorphic 
$CQ^{\neq}$-rules over schema {\bf P}, such that $Q$ is the output relation symbol in each rule. Then we say that the set ${\cal S}(Q)$ {\em defines a $UCQ^{\neq}$ query} $Q$ {\em over} {\bf P}, and that each element of ${\cal S}(Q)$ defines {\em a $CQ^{\neq}$ component of} $Q$. In the special case where ${\cal S}(Q)$ $=$ $\emptyset$, we say that the corresponding $UCQ^{\neq}$ query $Q$ is a {\em trivial query.} 

{\bf Semantics of $UCQ^{\neq}$ queries.} We now define the semantics of a $UCQ^{\neq}$ query $Q$. In the definition, we  will need the notions of {\em homomorphism} and of {\em valuation.} Consider two conjunctions, $\varphi({\bar Y})$ and $\psi({\bar Z})$, of relational atoms. % such that each variable occurring in $C$ also occurs in $\bar Z$. 
Then a mapping $h$ from the set of elements of $\bar Y$ to the set of elements of $\bar Z$ is called a {\em homomorphism} from $\varphi({\bar Y})$ to $\psi({\bar Z})$ whenever (i) $h(c)$ $=$ $c$ for each constant $c$ in $\bar Y$, and (ii) for each conjunct of the form $p({\bar U})$ in $\varphi({\bar Y})$, the relational atom $p(h({\bar U}))$ is a conjunct in $\psi({\bar Z})$. (For a vector $\bar S$ $=$ $[s_1 s_2 \ldots s_l]$, for some $l$ $\geq$ $0$, we define $h({\bar S})$ as the vector $[h(s_1) h(s_2) \ldots h(s_l)]$. %) We will denote homomorphisms by lowercase letters $g$, $h$, $\ldots \ $, possibly with subscripts. 
By convention, a homomorphism is an identity mapping when applied to empty vectors and to empty tuples.) % applied to the empty vector $[]$ is $[]$, and $h$ applied to the empty tuple $()$ is $()$.)  

We define homomorphisms in the same way for the case where either one of $\varphi({\bar Y})$ and $\psi({\bar Z})$ (or both) is a conjunction of facts. %, i.e., each element  of either $\bar Y$ or $\bar Z$ is in {\sc Const} $\cup$ {\sc Vars}, rather than in {\sc Const} $\cup$ {\sc Qvars} as in the above case of relational atoms. 
%Further, for a conjunction $C$ of (dis)equality atoms and two conjunctions  $\varphi({\bar Y})$ and $\psi({\bar Z})$ of  relational atoms or of facts, we say that each homomorphism, $h$, from $\varphi({\bar Y})$ to $\psi({\bar Z})$ is also a homomorphism from $\varphi({\bar Y})$ to $\psi({\bar Z}) \wedge C$. 
Further, for a conjunction $C$ of disequality atoms and two conjunctions  $\varphi({\bar Y})$ and $\psi({\bar Z})$ of  relational atoms or of facts, we say that every homomorphism, $h$, from $\varphi({\bar Y})$ to $\psi({\bar Z})$ is also a homomorphism from $\varphi({\bar Y})$ to $\psi({\bar Z})  \wedge C$. 
We will denote homomorphisms by lowercase letters $g$, $h$, $\ldots \ $, possibly with subscripts. 

Now suppose we are given a conjunction $\varphi({\bar Y})$ of relational atoms, a conjunction $\psi({\bar Z})$ of facts, and a conjunction $C$ of disequalities over variables in $\bar Y$ and constants in {\sc Const}. %(Note that $C$ does not contain equality atoms.) 
Suppose there is a homomorphism, $h$, from $\varphi({\bar Y})$ to $\psi({\bar Z})$, such that for each atom of the form $S$ $\neq$ $T$ in $C$, the values $h(S)$ and $h(T)$ are distinct elements of {\sc Const} $\cup$ {\sc Var}. % whenever $\theta$ is $\neq$, and are the same element of {\sc Const} $\cup$ {\sc Var} whenever $\theta$ is $=$. 
Then we say that $h$ is 
 a {\em valuation} from $\varphi({\bar Y}) \wedge C$ to $\psi({\bar Z})$. % is a homomorphism $h$ from $\varphi({\bar Y})$ to $\psi({\bar Z})$. % such that for each disequality atom of the form $Y$ $\theta$ $Z$ in $C$, the atom . 
We will use Greek letters $\mu$, $\nu$, $\ldots \ $, possibly with subscripts, for valuations. 

Given a $k$-ary $CQ^{\neq}$ query $Q({\bar X})$ and given an instance $I$, which we interpret as a conjunction of all the facts in $I$. Then {\em the answer to $Q$ on} $I$, denoted $Q(I)$, is 

\vspace{-0.1cm}

\begin{tabbing} 
$Q(I) = \{ \ {\nu({\bar X})} \ | \ \nu$ is a valuation from $body_{(Q)}$ to $I \ \}$ . 
\end{tabbing} 

\vspace{-0.1cm} 

\noindent 
(When $k$ $=$ $0$, i.e., $Q$ is a {\em Boolean query,} $\nu({\bar X})$ is the empty tuple.) %Note that $Q(I)$ is a relation defining an instance of schema $< Q >$. Thus, for each query mapping $M_Q$ that is associated with a CQ query, $M_Q$ maps each instance $I$ of the input schema to the instance of schema $< Q >$, with relation $Q(I)$. 
Further, for  a $UCQ^{\neq}$ query $Q$ defined by $l$ $\geq$ $1$ %$CQ^{\neq}$ 
rules $\{$ $Q^{(1)}$, $\ldots,$ $Q^{(l)}$ $\}$, and for an 
instance $I$, {\em the  answer to} $Q$ {\em on} $I$ is the union $\cup_{i=1}^l Q^{(i)}(I)$. 
%The set of ground facts in the relation $Q(I)$ will be of special interest to us, we will denote it by $ground(Q(I))$. 
By convention, for every trivial $UCQ^{\neq}$ query $Q$ and for every instance $I$, we have $Q(I)$ $=$ $\emptyset$. 

{\bf Query containment.} A query $Q_1$ {\it is contained in query $Q_2$,} denoted $Q_1 \sqsubseteq Q_2,$ if $Q_1(I) \subseteq Q_2(I)$ for every instance $I$. A classic result in~\cite{ChandraM77} by Chandra and Merlin states that a necessary and sufficient condition for the containment  $Q_1 \sqsubseteq Q_2,$ for CQ queries $Q_1$ and $Q_2$ of the same arity, is the existence of a containment mapping from $Q_2$ to $Q_1.$ Here, a {\it containment mapping} \cite{ChandraM77} from CQ query $Q_2({\bar X}_2)$ to CQ query $Q_1({\bar X}_1)$ is a homomorphism $h$ from $body_{(Q_2)}$ to $body_{(Q_1)}$ such that $h(\bar{X}_2) = \bar{X}_1$. By the results in %{\bf [[[ Mush check in ZMendelzon whether \cite{LevyMSS95} is the right pointer! ]]]} 
\cite{LevyMSS95}, this containment test of \cite{ChandraM77} remains true when $Q_1$ has built-in predicates. Thus, the same test holds in particular when $Q_1$ is a $CQ^{\neq}$ query. It follows that, for $Q_1$ a $UCQ^{\neq}$ query and for $Q_2$ a CQ query, determining whether $Q_1 \sqsubseteq Q_2$ is decidable. Indeed, the containment holds iff %either (i) $Q_1$ is trivial, or (ii) 
for each $CQ^{\neq}$ rule $Q^{(i)}_1$ $\in$ ${\cal S}(Q_1)$, $i$ $\in$ $[1,$ $l]$, we have $Q^{(i)}_1$ $\sqsubseteq$ $Q_2$.

{\bf Canonical database.} Every $CQ^{\neq}$ query $Q$ can be regarded as a  symbolic ground instance $I^{(Q)}$. $I^{(Q)}$ is defined as the result of % ``freezing'' the body of $Q$, which 
turning each relational atom $P_i(\ldots)$ in $body_{(Q)}$ into a tuple in the relation $I^{(Q)}[P_i]$. % that corresponds to predicate $P_i$. 
The procedure is to keep each constant in the body of $Q$, and to replace consistently each variable in the body of $Q$ by a distinct constant different from all the constants in  $Q$. The tuples that correspond to the resulting ground facts are the only tuples in the {\it  canonical database} $I^{(Q)}$ for $Q$, which is unique up to isomorphism.  

{\em Remark.} We have defined $CQ^{\neq}$-rules as not having explicit equality atoms in their bodies. As a result and by definition of canonical database, we are restricting our consideration to the set of all and only {\em satisfiable} $CQ^{\neq}$-rules/queries. (A $CQ^{\neq}$-rule/query $Q$ is {\em satisfiable} iff there exists an instance $I$ such that $Q(I)$ $\neq$ $\emptyset$.) 

\vspace{-0.1cm} 

%%\vspace{-0.1cm} 
\subsection{Dependencies and Chase} 
\label{dep-chase-prelims-sec} 

{\bf Embedded dependencies.} We consider dependencies $\sigma$ of the form 
%
%\vspace{-0.1cm}
%
\begin{equation} 
\label{basic-emb-dep-eq} 
\sigma: \phi(\bar{U},\bar{V}) \rightarrow \exists \bar{W} \ \psi(\bar{U},\bar{W})  
\end{equation} 

%\vspace{-0.1cm} 

\noindent 
with $\phi$ and $\psi$ conjunctions of relational atoms, possibly with equations added. (All the variables in $\bar U$, $\bar V$ are understood to be universally quantified.) Such dependencies, called {\it embedded dependencies,} are expressive enough to specify all usual integrity constraints, such as keys, foreign keys, and inclusion dependencies~\cite{AbiteboulHV95}. %We call $\phi$ the {\it premise} and $\psi$ the {\it conclusion.} 
If $\psi$ is a single equation, then $\sigma$ is an {\it equality-generating dependency (egd).} %An egd with $\psi$ comprising a single atom is a  {\it functional dependency (fd).} 
If $\psi$ consists only of relational atoms, then $\sigma$ is a {\it tuple-generating 
dependency (tgd).} We follow \cite{FaginKMP05} in allowing constants in egds and tgds. Each set of embedded dependencies without constants is equivalent to a set %, $\Sigma$, 
of tgds and egds~\cite{AbiteboulHV95}. We write $I \models \Sigma$ if instance $I$ satisfies all elements of set $\Sigma$ of dependencies. All the sets $\Sigma$ that we refer to are finite. 

{\bf Query containment under dependencies.} We say that {\em query $Q$ is contained in query} $P$ {\it under set of dependencies} $\Sigma$, denoted $Q \sqsubseteq_{\Sigma} P,$ if for every instance $I \models \Sigma$ we have $Q(I)$ $\subseteq$ $P(I)$. Queries $Q$ and $P$ {\em are equivalent under} $\Sigma$, denoted $Q \equiv_{\Sigma} P,$ if both $Q \sqsubseteq_{\Sigma} P$ and $P \sqsubseteq_{\Sigma} Q$ hold. $Q$ and $P$ {\em are equivalent (in the absence of dependencies),} denoted $Q$ $\equiv$ $P$, if  $Q \equiv_{\emptyset} P$. 
%of set containment in the absence of dependencies, see Sections~\ref{basics-sec}. 
%The definitions of bag equivalence and bag-set equivalence under dependencies (denoted by $\equiv_{\Sigma,B}$ and $\equiv_{\Sigma,BS}$ \ , respectively) are analogous extensions of the respective definitions of equivalence for the case $\Sigma = \emptyset$, see Section~\ref{bag-equiv-defs}. 

{\bf Chase for CQ queries.} Given a CQ query $Q(\bar{X}) \  \leftarrow \ \xi(\bar{X},\bar{Y})$ and a tgd $\sigma$  as in Eq. (\ref{basic-emb-dep-eq}); 
%of the form $\phi(\bar{U},\bar{V}) \rightarrow \exists \bar{W} \ \psi(\bar{U},\bar{W})$; 
%
assume w.l.o.g. 
that $Q$ has none of the variables in $\bar{W}$. The (standard \cite{GrecoMS12}) {\it chase of $Q$ with $\sigma$ is applicable} if there is a homomorphism $h$ from $\phi$ to $\xi$, %, where $\psi '$ is a (perhaps empty) proper subset of the set of atoms in $\psi$,\footnote{We use the standard interpretation of conjunctions as sets of conjuncts.} 
such that $h$ cannot be extended to a homomorphism from  $\phi \wedge \psi$ to $\xi$.  Then, a (standard)  {\it chase step} on $Q$ with $\sigma$ and $h$ is a rewrite of $Q$ into a CQ query $Q^*(\bar{X}) \  \leftarrow \ \xi(\bar{X},\bar{Y}) \wedge \psi (h(\bar{U}),\bar{W})$. %, where $\psi ''$ is $\psi$ $-$ $\psi '$ (as set difference on sets of atoms). 
It can be shown that $Q^*$ $\equiv_{\{\sigma \}}$ $Q$ and that $Q^*$ $\sqsubseteq$ $Q$. 

We now define a (standard \cite{GrecoMS12}) chase step with an egd. Assume a CQ query $Q$, as before, and an egd $\sigma$ of the form $\phi(\bar{U}) \rightarrow U_1 = U_2.$ The {\it chase of $Q$ with $\sigma$ is applicable} if there is a homomorphism $h$ from $\phi$ to $\xi$ such that $h(U_1) \neq h(U_2)$. Suppose at least one of $h(U_1)$ and $h(U_2)$ is a variable; let w.l.o.g. $h(U_1)$ be a variable. Then a {\it chase step}  
on $Q$ with $\sigma$ and $h$ is a rewrite of $Q$ into a CQ query, $Q^*$, that results from replacing all occurrences of $h(U_1)$ in $Q$ by $h(U_2)$. Again, we have $Q^*$ $\equiv_{\{\sigma \}}$ $Q$ and $Q^*$ $\sqsubseteq$ $Q$. If, for an %homomorphism 
$h$ as above, $h(U_1)$ and $h(U_2)$ are distinct constants, then we say that {\em chase with $\sigma$ fails on} $Q$. In this case, $Q(I)$ $=$ $\emptyset$ on all $I$ $\models$ $\{ \sigma \}$.  

A {\it $\Sigma$-chase sequence} ${\cal C}$ (or just {\it chase sequence,} if $\Sigma$ is clear from the context) {\em for CQ query} $Q_0$  %under semantics $X$ for query evaluation 
is a sequence of CQ queries $Q_0, Q_1, \ldots$ such that each query $Q_{i+1}$ ($i \geq 0$) in ${\cal C}$ is obtained from $Q_i$ by a chase step $Q_i \Rightarrow^{\sigma} Q_{i+1}$ using a dependency $\sigma \in \Sigma$. %Here, $X$ is one of $B$, $BS$, and $S$, referring to bag, bag-set, and set semantics, respectively. 
A chase sequence $Q = Q_0, Q_1, \ldots, Q_n$ is {\it terminating} if $I^{(Q_n)} \models \Sigma$, 
where  $I^{(Q_n)}$ is the canonical database for $Q_n$. In 
this case we denote $Q_n$ by $(Q)^{\Sigma}$ and say that $(Q)^{\Sigma}$ is the (terminal) {\it result} of the chase. All chase results for a given CQ query are equivalent in the absence of dependencies~\cite{DeutschPods08}. %, that is [[[??? isomorphic ??? under weakly acyclic ??? ]]] . 

%Chase of CQ queries is known to terminate in finite time for a class  of embedded dependencies that we will refer to as {\em weakly acyclic dependencies,} see~\cite{FaginKMP05} and references therein. Each set of weakly acyclic dependencies is a union of a set of egds with a set of tgds called {\it weakly acyclic tgds.} Due to the space limit, we provide the formal definitions in Appendix~\ref{weakly-acyclic-app}. 

{\bf Weakly acyclic dependencies \cite{FaginKMP05}.} 
Let $\Sigma$ be a set of tgds over schema {\bf T}. We construct the {\em dependency graph} of $\Sigma$, as follows. The nodes (positions) of the graph are all pairs $(T,$ $A)$, for $T$ $\in$ {\bf T} and $A$ an attribute of $T$. We now add edges:  For each tgd $\varphi({\bar X})$ $\rightarrow$ $\exists$ $\bar Y$ $\psi({\bar X}, {\bar Y})$ in $\Sigma$, and for each $X$ $\in$ $\bar X$ that occurs in $\varphi$ in position $(T$, $A)$ and that occurs in $\psi$, do the following.  

\begin{itemize} 
%\vspace{-0.2cm} 
	\item For each occurrence of $X$ in $\psi$ in position $(S,$ $B)$, add a {\em regular edge} from $(T,$ $A)$ to $(S$, $B)$; and % (if does not already exist); and 
\vspace{-0.2cm} 
	\item For each existentially quantified variable $Y$ $\in$ $\bar Y$ and for each occurrence of $Y$ in $\psi$ in position $(R$, $C)$, add a {\em special edge} % labeled $\star$ 
	from $(T$, $A)$ to $(R$, $C)$. % (if does not already exist). 
%\vspace{-0.2cm} 
\end{itemize} 

For a set $\Sigma$ of tgds and egds, with $\Sigma^{t}$ the set of all tgds in $\Sigma$, we say that $\Sigma$ is {\em weakly acyclic} if the dependency graph of $\Sigma^{t}$ does not have a cycle going through a special edge. % labeled $\star$. 
Chase of CQ queries terminates in finite time under sets of weakly acyclic dependencies \cite{FaginKMP05}.

%\reminder{Revisit~\cite{DeutschDiss} on which bag-semantics chase steps are sound there}

The following result is immediate %{\bf [[[ NB! Need to check in both \cite{JohnsonK84} and \cite{AbiteboulHV95} whether in Theorem~\ref{chase-theorem}, there is simply $Q_2$ instead of (what I have in that theorem now:) $(Q_2)^{\Sigma}$.  ]]]} 
from~\cite{AbiteboulHV95,DeutschDiss,DeutschPods08,JohnsonK84}.  %for sets of functional dependencies; we have extended it to sets of embedded dependencies  (see also~\cite{DeutschDiss}). %; we present here our own proof for the types of constraints we consider in the paper.)
%%\vspace{-0.1cm}

%\vspace{-0.1cm} 

\vspace{-0.1cm} 

%%\vspace{-0.2cm}
\begin{theorem}
\label{chase-theorem}
Given CQ queries $Q_1$, $Q_2$ and a set $\Sigma$ of embedded dependencies. Then 
%%\vspace{-0.1cm}
%\begin{enumerate}
%	\item \reminder{Do I need this item for the purpose of this paper???} $Q_1 \sqsubseteq_{\Sigma,S} Q_2$ iff $(Q_1)_{\Sigma,S} \ \sqsubseteq_S (Q_2)_{\Sigma,S}$ in the absence of dependencies.
%%\vspace{-0.1cm}
%	\item 
	$Q_1 \sqsubseteq_{\Sigma} Q_2$ iff $(Q_1)^{\Sigma} \sqsubseteq Q_2$ in the absence of dependencies.
%\end{enumerate}
%%\vspace{-0.6cm}
\end{theorem}

%\vspace{-0.1cm} 

\vspace{-0.1cm} 

{\bf Chase of instance.} Let $I$ be an instance of schema {\bf P}, and $\Sigma$ a set  of egds and tgds; we interpret $I$ as a conjunction of its facts. We follow \cite{DeutschPods08} in defining chase of $I$ with $\Sigma$ in the same way as chase of a CQ query with $\Sigma$. That is, in the chase steps we treat each distinct null in $I$ as a distinct variable (in the chase for CQ queries). Further, each chase step with a tgd that has existential variables introduces, in the result of the chase step, a distinct new null for each existential variable of the tgd. Chase sequences and chase termination are also defined in the same way as for CQ queries; the result $I'$ of the chase of $I$ with $\Sigma$ always satisfies $\Sigma$, that is, $I'$ $\models$ $\Sigma$. 

%\newpage 

\vspace{-0.2cm} 

\section{The Problem Statements} 
\label{problem-statement-sec}

In this section we formalize the problems of finding and determining certain query answers and of query containment, under CWA and in presence of dependencies. We then establish a direct %straightforward 
relationship between the latter two problems in the case of CQ view definitions. 

\vspace{-0.1cm} 

%\subsection{Certain query answers and query containment w.r.t. views and dependencies} 
\subsection{Certain Query Answers and Query Containment w.r.t. Views and Dependencies} 
\label{probl-stmt-defs-sec} % problem of finding the set of  certain query answers w.r.t. a setting

We begin by introducing the notion of ``materialized-view setting'' (``setting'' for short). Suppose that we are given a schema {\bf P} and a set of dependencies $\Sigma$ on {\bf P}.  Let $\cal V$ be a finite set of relation symbols not in {\bf P}, with each symbol  {\em (view name)}  $V$ $\in$ $\cal V$ of some 
arity $k_V$ $\geq$ $0$. Each $V$ $\in$ $\cal V$ is associated with a $k_V$-ary query  on the schema {\bf P}. We call $\cal V$ a {\em set of views on} {\bf P}, and call the query for each $V$ $\in$ $\cal V$ the {\em definition of the view} $V$, or {\em the query for} $V$. We assume that  the query for each $V$ $\in$ $\cal V$ is associated with ($V$ in) the set $\cal V$. We call a ground instance $MV$ of schema $\cal V$ a {\em set of view answers for} $\cal V$. 
%We say that $MV$  is {\em nonempty,} $MV$ $\neq$ $\emptyset$, whenever there exists at least one view $V$ $\in$ $\cal V$ such that the relation $MV[V]$ is not the empty set.  

Let $I$ be a ground instance of schema {\bf P}. We say that $I$ {\em is a $\Sigma$-valid base instance for} $\cal V$ {\em and} $MV$, denoted by $\cal V$ $\Rightarrow_{I,{\Sigma}}$ $MV$, whenever (a) $I$ $\models$ $\Sigma$, and (b) the answer $V(I)$ to the query for $V$ on the instance $I$ is identical to the relation $MV[V]$ in $MV$, for each $V$ $\in$ $\cal V$. (This is the {\em closed-world assumption (CWA),} as defined in, e.g., \cite{AbiteboulD98}, with an added requirement that $I$ $\models$ $\Sigma$.) Further, we say that $MV$ {\em is a $\Sigma$-valid set of view answers for} $\cal V$, denoted by $\cal V$ $\Rightarrow_{*,{\Sigma}}$ $MV$, whenever there exists a $\Sigma$-valid base instance for $\cal V$ and $MV$. 

\vspace{-0.1cm} 

\begin{definition}{Materialized-view setting ${\cal M}\Sigma$} % and ${{\cal M}\Sigma}_{\emptyset}$} %\linebreak 
\label{mat-setting-def} 
Given a schema {\bf P}, a set $\Sigma$ of dependencies without constants on {\bf P}, a set $\cal V$ of views on {\bf P}, and a ($\Sigma$-valid) set $MV$ %$\neq$ $\emptyset$ 
of view answers for $\cal V$:  
We call ${{\cal M}\Sigma}$ $=$ ({\bf P}, $\Sigma$, $\cal V$, $MV$) the {\em (valid) materialized-view setting for} {\bf P}, $\Sigma$, $\cal V$, and $MV$. %Further, for a given setting ${\cal M}\Sigma$, the setting ${{\cal M}\Sigma}_{\emptyset}$ is the result of replacing $\Sigma$ in ${\cal M}\Sigma$ with the empty set.   
\end{definition} 

\vspace{-0.1cm} 

%For a setting ${\cal M}\Sigma$ as in Definition~\ref{mat-setting-def}, we denote by ${{\cal M}\Sigma}_{\emptyset}$ the result of replacing the set of dependencies $\Sigma$ in ${\cal M}\Sigma$ with the empty set. 

%\newpage 

Let ${\cal M}\Sigma$ be a materialized-view setting $(${\bf P}, $\Sigma$, $\cal V$, $MV)$, 
and let $Q$ be a query over {\bf P}. We define %$certain_{{\cal M}\Sigma}(Q)$, 
the {\em set of certain answers of} $Q$ {\em w.r.t. the setting} ${\cal M}\Sigma$ as 

%\vspace{-0.1cm} 

% $\cal M$ $=$ ({\bf S}, {\bf T}, $\Sigma_{st}$ $\cup$ $\Sigma_t$)
\begin{tabbing} 
$certain_{{\cal M}\Sigma}(Q)=\bigcap \ \{ Q(I)$ | $I$ s.t. $\cal V$ $\Rightarrow_{I,{\Sigma}}$ $MV$ in ${{\cal M}\Sigma}\}.$
\end{tabbing} 

\noindent 
That is, the set of certain answers of a query w.r.t. a setting is understood, as usual, as the set of all tuples that are in the answer to the query on all the instances relevant to the setting. (Cf. \cite{AbiteboulD98} for the case $\Sigma$ $=$ $\emptyset$.)

%\vspace{-0.1cm} 

\vspace{-0.1cm} 

\begin{definition}{Certain-query-answer problem in a materialized-view setting} 
\label{certain-query-answer-def} 
Given a setting ${{\cal M}\Sigma}$ $=$ $(${\bf P}, $\Sigma$, $\cal V$, $MV)$, a $k$-ary ($k$ $\geq$ $0$) query $Q$ over the schema {\bf P} in ${{\cal M}\Sigma}$, and a ground $k$-tuple $\bar t$. %of constants from $adom(MV)$. 
Then the {\em certain-query-answer  problem for $Q$ and $\bar t$ in} ${\cal M}\Sigma$ is to determine whether $\bar t$ $\in$ $certain_{{\cal M}\Sigma}(Q)$. 
 %we say that the tuple ${\cal D}_s$ $=$ $(${\bf P}, $\Sigma$, $\cal V$, $Q$, $MV)$ is a {\em (valid) specific instance of the information-leak problem.} 
\end{definition} 

\vspace{-0.1cm} 

%We now return to the problem of Definition~\ref{certain-query-answer-def}. It is easy to show that a tuple $\bar t$ can be a certain answer to a query in a setting ${\cal M}\Sigma$ only if all values in $\bar t$ are in $adom(MV)$, for the set $MV$ of view answers in ${\cal M}\Sigma$. 

%\begin{proposition} 
%Given a setting ${\cal M}\Sigma$ with set $MV$ of materialized views, a $k$-ary query $Q$, and a ground $k$-tuple $\bar t$. Then $\bar t$ $\in$ $certain_{{\cal M}\Sigma}(Q)$ only if for each element $t$ of $\bar t$, we have $t$ $\in$ $adom(MV)$. 
%\end{proposition} 

%\noindent 
%Thus, in Definition~\ref{certain-query-answer-def} we can restrict our consideration to only those tuples $\bar t$ whose all values are in $adom(MV)$, for $MV$ in the given setting ${\cal M}\Sigma$. 

It is easy to show that a tuple $\bar t$ can be a certain answer to a query $Q$ in a setting ${\cal M}\Sigma$ only if all the values in $\bar t$ are in $consts({\cal M}\Sigma)$, which denotes %ere $consts({\cal M}\Sigma)$ is 
the set of constants occurring in ${\cal M}\Sigma$. (For a given materialized-view setting ${\cal M}\Sigma$ $=$ $(${\bf P}, $\Sigma$, $\cal V$, $MV)$, we define $consts({\cal M}\Sigma)$ as the union of $adom(MV)$ with the set of all the  constants used in the definitions of the views $\cal V$.) 
%$adom(MV)$, for the set $MV$ of view answers in ${\cal M}\Sigma$. 
By this observation, in Definition~\ref{certain-query-answer-def} we can restrict our consideration to the tuples $\bar t$ with this property. % w.r.t. ${\cal M}\Sigma$. % whose all values are in $adom(MV)$, for $MV$ in the given setting ${\cal M}\Sigma$. 

The problem as in Definition~\ref{certain-query-answer-def}, the {\em problem of determining certain query answers,} will be featured in our characteristic of the relationship between the extensions of the problems of \cite{AbiteboulD98} and of \cite{ZhangM05} to the case of dependencies under CWA. %, and to characterize the complexity of both problems. 
We will also consider the {\em problem of finding the set of  certain query answers w.r.t. a setting:} Given a setting ${\cal M}\Sigma$ and a query $Q$, % defined on the schema {\bf P} in ${{\cal M}\Sigma}$, 
find the set of certain answers of $Q$ w.r.t. ${\cal M}\Sigma$. %Specifically, in Sections~\ref{} and~\ref{} 
In Sections~\ref{query-cont-algor-sec}--\ref{main-vv-dexchg-sec}, we will introduce algorithms for solving the ``CQ weakly acyclic case'' of this problem and of the problem of Definition~\ref{certain-query-answer-def}. %(See Sections~\ref{query-cont-algor-sec}--\ref{main-vv-dexchg-sec}.) %of finding the set of  certain query answers w.r.t. a setting,} 
The {\em CQ weakly acyclic case} of each problem is the case where: (i) each %setting 
${\cal M}\Sigma$ is {\em conjunctive} (i.e., all the views in ${\cal M}\Sigma$ are defined as CQ queries) and {\em weakly acyclic} (i.e., $\Sigma$ in ${\cal M}\Sigma$ is a set of weakly acyclic embedded dependencies), and (ii) each $Q$ is a CQ query. 

We now turn our attention to the problem of query containment w.r.t. a setting ${\cal M}\Sigma$. Our Definition~\ref{main-mendelz-def} extends the formalization of this problem due to \cite{ZhangM05}, to the case of dependencies on the relevant base instances. 

\vspace{-0.1cm} 

\begin{definition}{${\cal M}\Sigma$-conditional query containment}  
\label{main-mendelz-def} 
Given a materialized-view setting ${\cal M}\Sigma$ and queries $Q_1$ and $Q_2$ over the schema {\bf P} in ${\cal M}\Sigma$. Then we say that $Q_1$ {\em  is ${\cal M}\Sigma$-conditionally contained in} $Q_2$, denoted $Q_1$ $\sqsubseteq_{{\cal M}\Sigma}$ $Q_2$, iff for each instance $I$ s.t. $\cal V$ $\Rightarrow_{I,{\Sigma}}$ $MV$ in ${\cal M}\Sigma$, we have $Q_1(I)$ $\subseteq$ $Q_2(I)$. Further, %given ${\cal M}\Sigma$, $Q_1$, and $Q_2$ as above,
 {\em the problem of ${\cal M}\Sigma$-conditional containment for} $Q_1$ and $Q_2$ %w.r.t. the given ${\cal M}\Sigma$ 
 is to determine whether $Q_1$ $\sqsubseteq_{{\cal M}\Sigma}$ $Q_2$. 
\end{definition} 

\vspace{-0.1cm} 

%Please see Appendix~\ref{probl-stmt-illustr-sec} for an example illustrating the three problems formulated in this subsection. 
%\vspace{-0.1cm} 

%\section{An illustration of the three \\ problems considered in this \\ paper} 
\subsection{An Illustration} 
\label{probl-stmt-illustr-sec} 

In this subsection we recast Example~\ref{intro-main-three-ex} into the formal terms of Section~\ref{probl-stmt-defs-sec}. The results of this paper permit us to obtain correct solutions to all the three problems formulated at the end of Example~\ref{intro-formalized-ex}.

\begin{example} 
\label{intro-formalized-ex} 
The %materialized-view 
setting ${\cal M}\Sigma$ outlined in Example~\ref{intro-main-three-ex} uses the schema\footnote{We abbreviate the relation names of Example~\ref{intro-main-three-ex} using the first letter of each name.}  {\bf P} $=$ $\{$$E$, $H$, $O$$\}$ and a weakly acyclic set $\{ \sigma \}$ of dependencies, with $\sigma$ as follows:   

\vspace{-0.1cm} 

\begin{tabbing} 
$\sigma$: $E(X,Y,Z) \wedge H(Y) \rightarrow \exists S \ \ O(X,S)$. 
\end{tabbing} 

\vspace{-0.1cm} 

Further, $\cal V$ $=$ $\{$$U$, $V$, $W$$\}$ is the set of CQ views in ${\cal M}\Sigma$, with the view definitions as follows: 

\vspace{-0.1cm} 

\begin{tabbing} 
Hop me b \= boo \kill
$U(X)$ \> $\leftarrow H(X).$ \\ 
$V(X,Y)$ \> $\leftarrow E(X,Y,Z).$ \\ 
$W(Y,Z)$ \> $\leftarrow E(X,Y,Z).$ 
\end{tabbing} 

\vspace{-0.1cm} 

Finally, for brevity we encode the constants of Example~\ref{intro-main-three-ex} as $c$ for {\tt johnDoe,} $d$ for {\tt sales,} and $f$ for {\tt 50000.} Then the set of view answers $MV$ of Example~\ref{intro-main-three-ex} can be recast for ${\cal M}\Sigma$ as $MV$ $=$ $\{ \ U(d), V(c,d), W(d,f) \ \}.$  % follows: 

%\begin{tabbing}
%$MV$$=$$\{ U(d), V(c,d), W(d,f) \}.$ 
%\end{tabbing}

Now that we have specified a %complete specification of the
CQ weakly acyclic setting ${\cal M}\Sigma$, consider the CQ query $Q$ of Example~\ref{intro-main-three-ex}: 

\vspace{-0.1cm} 

\begin{tabbing} 
$Q(X,Z)$ $\leftarrow E(X,Y,Z), O(X,S).$ 
\end{tabbing} 

\vspace{-0.1cm} 

Consider another CQ query, $Q_1$, defined as follows: 
\begin{tabbing} 
$Q_1(c,f)$ $\leftarrow H(d), E(c,d,X), E(Y,d,f).$ 
\end{tabbing} 

\vspace{-0.1cm} 

%Define a tuple $\bar t$ as $\bar t$ $=$ $(c,f)$. Then, 
For the ${\cal M}\Sigma$, $Q$, and $Q_1$ as above and for a tuple $\bar t$ $=$ $(c,f)$, we have the following problems as in Section~\ref{probl-stmt-defs-sec}: 

\begin{enumerate} 
	\item The certain-query-answer problem for $Q$ and $\bar t$ in ${\cal M}\Sigma$ is ``Is $\bar t$ a certain answer of $Q$ w.r.t. ${\cal M}\Sigma$?'' 
\vspace{-0.1cm} 
	\item The problem of finding the set of certain answers to $Q$ w.r.t.  ${\cal M}\Sigma$ is ``Return the set $certain_{{\cal M}\Sigma}(Q)$ for $Q$ and ${\cal M}\Sigma$''; and, finally, 
\vspace{-0.1cm} 
	\item The problem of ${\cal M}\Sigma$-conditional containment for $Q_1$ and $Q$ %w.r.t ${\cal M}\Sigma$ 
	is ``Does $Q_1$ $\sqsubseteq_{{\cal M}\Sigma}$ $Q$ hold?'' 
	%``Is %the query 	$Q'$ ${\cal M}\Sigma$-conditionally contained in %the query 	$Q$?''  
\end{enumerate} 
\vspace{-0.6cm} 
\end{example}

\vspace{-0.1cm}

\subsection{Relationship between the Problems} 
\label{probl-stmt-relationship-sec} 

We now establish a direct relationship between the certain-query-answer problem for a given $Q$, $\bar t$, and ${\cal M}\Sigma$, %on the one hand, 
and the problem of ${\cal M}\Sigma$-conditional containment for $Q_1$ and $Q$, % w.r.t. ${\cal M}\Sigma$, 
for the same $Q$ and ${\cal M}\Sigma$. (We prove the relationship for the case where all the views are defined as CQ queries.) %definitions are CQ.) 
Here, the query $Q_1$ is constructed from the given $Q$, $\bar t$, and ${\cal M}\Sigma$. A similar relationship was observed in \cite{AbiteboulD98} %for finding certain query answers 
between  the certain-query-answer problem, for a range of query and view languages in the dependency-free case under OWA, and {\em unconditional} ($Q_1$ $\sqsubseteq$ $Q$) query containment. % and in the absence of dependencies. 
In contrast, our result holds under CWA, in presence of dependencies, and involves ${\cal M}\Sigma${\em -conditional} query containment. Due to this result, the algorithm that we introduce in Section~\ref{query-cont-algor-sec} for checking %the problem of 
${\cal M}\Sigma$-conditional containment, can also be used to solve the certain-query-answer problem, in the CQ weakly acyclic case of each problem. (The {\em CQ weakly acyclic case} of the containment problem covers CQ weakly acyclic settings and CQ input queries.) 

We formulate the main result of this section, Theorem~\ref{problem-relationship-thm}, using the following notation. For a set $\cal V$ $=$ $\{ $$V_1$, $\ldots$, $V_m$$\}$ of $m$ $\geq$ $1$ CQ views and for a set $MV$ of view answers for $\cal V$,  %let a set  of view answers be 
consider the conjunction 

\vspace{-0.1cm} 

\begin{tabbing} 
huhu mu \= bubu \kill 
${\cal C}_{MV}$ $=$ $\bigwedge_{i=1}^m \bigwedge_{j=1}^{l_i} V_i({\bar t}_{ij})$ 
%
%$\{$ $V_1({\bar t}_{1,1})$, $V_1({\bar t}_{1,2})$, $\ldots$, $V_1({\bar t}_{1,l_1})$, $V_2({\bar t}_{2,1})$, $\ldots$, \\ 
%\> $V_2({\bar t}_{2,l_2})$, $\ldots$, $V_m({\bar t}_{m,1})$, $\ldots$, $V_m({\bar t}_{m,l_m})$ $\}$,  
\end{tabbing} 

\vspace{-0.1cm} 

%\begin{tabbing} 
%huhu mu \= bubu \kill 
%$MV$ $=$ $\{$ $V_1({\bar t}_{1,1})$, $V_1({\bar t}_{1,2})$, $\ldots$, $V_1({\bar t}_{1,l_1})$, $V_2({\bar t}_{2,1})$, $\ldots$, \\ 
%\> $V_2({\bar t}_{2,l_2})$, $\ldots$, $V_m({\bar t}_{m,1})$, $\ldots$, $V_m({\bar t}_{m,l_m})$ $\}$,  
%\end{tabbing} 

\noindent 
The conjunction is over all the ground facts $V_i({\bar t}_{ij})$ in the set $MV$. (For each $i$ $\in$ $[1,$ $m]$, the relation $MV[V_i]$ in $MV$ is of cardinality $l_i$ $\geq$ $0$.) That is, we treat each ground fact in $MV$ as a relational atom, and ${\cal C}_{MV}$ is the conjunction of all these relational atoms. (For each $i$ such that $MV[V_i]$ $=$ $\emptyset$, we define $\bigwedge_{j=1}^{l_i} V_i({\bar t}_{ij})$ $:=$ $true$.)  
%(In the special case where $MV$ $=$ $\emptyset$, we define ${\cal C}_{MV}$ $:=$ $true$.) 

Observe that ${\cal C}_{MV}$ can be treated as the body of a CQ query over the schema $\cal V$. Thus, we can use the view definitions in $\cal V$ to do the standard {\em expansion} (as in a rewriting \cite{LevyMSS95}) of ${\cal C}_{MV}$ into a conjunction of atoms, ${\cal C}^{exp}_{MV}$, over the schema {\bf P}. We call ${\cal C}^{exp}_{MV}$ {\em the expansion of $MV$ over} {\bf P}. As an illustration, in the setting of Example~\ref{intro-formalized-ex}, ${\cal C}_{MV}$ is $U(d) \wedge V(c,d) \wedge W(d,f)$, and %its expansion 
${\cal C}^{exp}_{MV}$ is the body of the query $Q_1$ in the example. 

We now formulate Theorem~\ref{problem-relationship-thm}. (Due to the page limit, the straightforward proof and other details can be found in Appendix~\ref{app-rewriting-sec}.) This result says that for a valid  CQ materialized-view setting ${\cal M}\Sigma$ and for an arbitrary query $Q$ and an arbitrary ground tuple $\bar t$, there exists a (constructible) CQ query $Q_1$ such that the certain-query-answer problem for $Q$ and $\bar t$ in ${\cal M}\Sigma$ is the problem of ${\cal M}\Sigma$-conditional containment for $Q_1$ and $Q$. % w.r.t. ${\cal M}\Sigma$. 

%{\bf [[[ Must proofread Appendix~\ref{app-rewriting-sec}, for reformulating it into the language of this Section~\ref{probl-stmt-relationship-sec} and then for verifying all the ensuing results. ]]]} 

\vspace{-0.1cm} 

\begin{theorem} 
\label{problem-relationship-thm} 
Given a valid CQ materialized-view setting ${\cal M}\Sigma$ $=$ $(${\bf P}, $\Sigma$, $\cal V$, $MV)$, a $k$-ary ($k$ $\geq$ $0$) query $Q$ defined in an arbitrary query language over {\bf P}, and a $k$-tuple $\bar t$ of values in $consts({\cal M}\Sigma)$. Consider the CQ query $Q_1(\bar t) \leftarrow {\cal C}^{exp}_{MV}$. Then $\bar t$ $\in$ $certain_{{\cal M}\Sigma}(Q)$ if and only if $Q_1$ is  ${\cal M}\Sigma$-conditionally contained in $Q$. 
\end{theorem} 

\vspace{-0.1cm} 

Whenever determining validity of a setting ${\cal M}\Sigma$ is decidable (as is the case for, e.g., CQ weakly acyclic settings, via our view-verified data-exchange approach of Section~\ref{main-vv-dexchg-sec}, see Appendix~\ref{vv-correctness-sec}), ${\cal M}\Sigma$ not being valid implies that $certain_{{\cal M}\Sigma}(Q)$ $=$ $\emptyset$ for every query $Q$. 

%{\bf [[[ Throughout the paper: Must give intuition for each proof ]]]} 

%Theorem~\ref{problem-relationship-thm} considers only conjunctive settings ${\cal M}\Sigma$, which is sufficient for the algorithms introduced in this paper. At the same time, we note that it appears straightforward to extend Theorem~\ref{problem-relationship-thm}  to (at least) the view language UCQ, by using the techniques of \cite{AbiteboulD98}.  

\vspace{-0.2cm} 

\section{The Query-Containment Problem} 
\label{query-cont-algor-sec} 

In this section we outline our approach to solving the problem of ${\cal M}\Sigma$-conditional query containment. (See Definition~\ref{main-mendelz-def}.) We show that this approach  is a correct algorithm for the %following {\em } 
CQ weakly acyclic case of the problem. %: The setting ${\cal M}\Sigma$ is conjunctive and weakly acyclic,  and each of $Q_1$ and $Q_2$ is a CQ query. 
Thus, our algorithm extends to the case of weakly acyclic dependencies the solution of \cite{ZhangM05}  for their problem of conditional containment between CQ queries in presence of materialized CQ views.\footnote{A full version of \cite{ZhangM05}, including proofs of its results, has never been published.} %, we had to develop from scratch all the proofs for our extended version of the problem.} 
We show that our extension of the method of \cite{ZhangM05} is not trivial. By Theorem~\ref{problem-relationship-thm}, %the results of Section~\ref{probl-stmt-relationship-sec}, 
the approach reported in this section is also a correct algorithm for the CQ weakly acyclic cases of the certain-query-answer problem.  

\vspace{-0.1cm} 

\subsection{Intuition and Discussion} 
\label{zmendelz-intuition-sec} 

We begin by sketching our containment-checking approach via an extended example. The example illustrates, in particular, how disequalities and disjunction may arise in the %result of 
chase of a CQ query in this approach. 

\begin{example} 
\label{disj-neq-ex} 
%that ${\cal M}\Sigma$-conditional query contained can be incorrectly inferred if we do not apply negds in the chase of one of the input queries. 
Consider CQ queries $Q_1$ and $Q_2$: %, as follows. 

\vspace{-0.1cm} 

\begin{tabbing} 
$Q_1(X) \leftarrow P(X,Y).$ \\ 
$Q_2(X) \leftarrow P(X,Y), R(Z).$  
\end{tabbing} 

\vspace{-0.1cm} 

Consider a dependency (full tgd) $\sigma$ on the schema {\bf P} $=$ $\{ P, R \}$, a view $V$, and an instance $MV$, as follows. 

\vspace{-0.1cm} 

\begin{tabbing} 
$\sigma: P(X,X) \rightarrow R(X)$ \\ 
$V(X) \leftarrow P(X,X).$ \\ 
$MV$ $=$ $\{$ $V(c)$ $\}$.  
\end{tabbing} 

\vspace{-0.1cm} 

Let us specify a setting ${\cal M}\Sigma$ as $(${\bf P}, $\{ \sigma \},$ $\{ V \}, MV)$. The setting ${\cal M}\Sigma$ is CQ weakly acyclic by definition. 

%\noindent 
By the results reviewed in Section~\ref{prelim-sec}, the query $Q_1$ is not unconditionally contained in $Q_2$, either in the absence of dependencies or in presence of  $\sigma$. At the same time, by our results in this section, $Q_1$ {\em is} ${\cal M}\Sigma$-contained in $Q_2$. Our approach to proving it is by chasing the query $Q_1$ using a ``$\neq${\em -transformation,}'' $\sigma_{(\neq)}$, of the given tgd $\sigma$ on the schema {\bf P},  %in ${\cal M}\Sigma$, 
as well as ``$MV$-induced dependencies.'' (We introduce both kinds of dependencies in Section~\ref{containment-dependencies-sec}.) %this paper.) 
The first step of the approach is to conjoin the body of $Q_1$ with ${\cal C}^{exp}_{MV}$ $=$ $P(c,c)$ (see Section~\ref{probl-stmt-relationship-sec} for the definition of ${\cal C}^{exp}_{MV}$): 

\vspace{-0.1cm} 

\begin{tabbing} 
$Q'_1(X) \leftarrow P(X,Y), P(c,c).$  
\end{tabbing} 

\vspace{-0.1cm} 

Now the only $MV$-induced dependency, $\tau_V$, is 

\vspace{-0.1cm} 

\begin{tabbing} 
$\tau_V: P(X,Y) \rightarrow (X = c \wedge Y = c) \vee (X \neq Y)$. 
\end{tabbing} 

\vspace{-0.1cm} 

\noindent 
It says that, for each subgoal of the form $P(X,Y)$ that could arise in the chase of $Q'_1$ with the dependencies $\tau_V$ and $\sigma_{(\neq)}$: Either (i) the subgoal must become $P(c,c)$, which would (correctly) give rise to  $V(c )$ in $MV$, or (ii)  %the subgoal 
$P(X,Y)$ must be accompanied by the disequality $X \neq Y$, to prevent atoms of the form $V(d)$, where $d$ is a constant not equal to $c$, from arising in $MV$. (These requirements must be satisfied for our approach to be correct, see Proposition~\ref{zmendelz-valuation-prop} in Section~\ref{containment-procedure-sec}.) 

The chase of $Q'_1$ with $\tau_V$ produces a $UCQ^{\neq}$ query: 

\vspace{-0.1cm} 

\begin{tabbing} 
$Q^{(a)}_1(c ) \leftarrow P(c,c), P(c,c).$ (We then drop the duplicate.) \\   
$Q^{(b)}_1(X) \leftarrow P(X,Y), P(c,c), X \neq Y.$ 
\end{tabbing} 

\vspace{-0.1cm} 

Now the dependency $\sigma_{(\neq)}$, which we obtain from the tgd $\sigma$, is $\sigma_{(\neq)}: P(X,Y) \rightarrow R(X) \vee X \neq Y.$ Applying $\sigma_{(\neq)}$ to the above $UCQ^{\neq}$ query yields the $UCQ^{\neq}$ result $(Q_1)^{{\cal M}\Sigma}$ $=$ $\{ Q^{(1)}_1, Q^{(2)}_1, Q^{(3)}_1 \}$ of chasing the query $Q'_1$ with the dependencies $\tau_V$ and $\sigma_{(\neq)}$: 

\vspace{-0.1cm} 

\begin{tabbing} 
$Q^{(1)}_1(c ) \leftarrow P(c,c), R( c).$  \\   
$Q^{(2)}_1(X) \leftarrow P(X,Y), R(X), P(c,c), R(c ).$ \\ 
$Q^{(3)}_1(X) \leftarrow P(X,Y), X \neq Y, P(c,c), R(c ).$ 
\end{tabbing} 

\vspace{-0.1cm} 

Now the results of \cite{LevyMSS95} can be used to ascertain the unconditional containment of $(Q_1)^{{\cal M}\Sigma}$ in the query $Q_2$. We conclude that the query $Q_1$ is ${\cal M}\Sigma$-contained in $Q_2$. 

Finally, suppose that we change the query $Q_2$ slightly, by replacing its subgoal $R(Z)$ with $R(X)$. Then the same procedure as above can be used to show that the resulting query would {\em not}  ${\cal M}\Sigma$-contain the query $Q_1$. 
\end{example} 

In some particularly simple cases, queries $(Q_1)^{{\cal M}\Sigma}$ can be CQ queries; see Appendix~\ref{cannot-chase-mendelzon-sec}. In general in our approach, queries $(Q_1)^{{\cal M}\Sigma}$  are $UCQ^{\neq}$ queries. 

In our proposed approach for checking ${\cal M}\Sigma$-containm- ent of CQ queries, the intuition is the same as in cheching query containment in presence of dependencies \cite{AbiteboulHV95,DeutschDiss,DeutschPods08,JohnsonK84} (see Section~\ref{prelim-sec}). That is, %in our approach, 
to determine if a query $Q_1$ is contained in query $Q_2$ on a set of instances that are ``relevant'' to a set of view answers $MV$, %Example~\ref{disj-neq-ex}  shows how 
		we chase %the query 
		$Q_1$ to transform it into a query, $(Q_1)^{{\cal M}\Sigma}$, which is equivalent, by construction, to $Q_1$ on all the relevant instances. (The ``relevant instances'' are the $\Sigma$-valid base instances for the given $\cal V$ and $MV$.) In addition to this  property, the query $(Q_1)^{{\cal M}\Sigma}$, by its construction, ``exhibits the flavor of the relevant instances,'' in a very precise sense (see Proposition  \ref{zmendelz-valuation-prop} in Section~\ref{containment-procedure-sec}). %:  For each valuation, $\nu$, for the query $(Q_1)^{{\cal M}\Sigma}$ and for an arbitrary (not necessarily just relevant) instance, the image of the relational subgoals of $(Q_1)^{{\cal M}\Sigma}$ under $\nu$ is a relevant instance.  
		These properties permit us to use a test for unconditional containment of $(Q_1)^{{\cal M}\Sigma}$ in $Q_2$ to correctly determine whether the original query $Q_1$ is contained in $Q_2$ w.r.t. all the relevant instances. (See Theorem~\ref{main-mendelz-thm} in Section~\ref{containment-procedure-sec}.)
	
Zhang and Mendelzon in their paper  \cite{ZhangM05} did precisely the above chase, with precisely the same goals and results, in the special case where no dependencies hold on the relevant instances. As an illustration, suppose that in Example~\ref{disj-neq-ex} we set $\Sigma$ $:=$ $\emptyset$, while keeping the remaining inputs as they are. Then the approach of \cite{ZhangM05} for these inputs would derive the $UCQ^{\neq}$ query $\{ Q^{(a)}_1$, $Q^{(b)}_1 \}$ of that example, call this query $Q''_1$. %{\bf [[[ Need to double check that this claim about $Q''_1$ and $(Q_1)^{{\cal M}\Sigma}$ is correct ]]]} 
		As $Q''_1$ is not unconditionally contained in the given query $Q_2$, the conclusion of \cite{ZhangM05} for these inputs would be that $Q_2$ does not contain $Q_1$ w.r.t. these inputs with $\Sigma$ $=$ $\emptyset$. 
		
Thus, in this current work we build  directly on the ideas and techniques of \cite{ZhangM05}.  At the same time, \cite{ZhangM05} does not make the chase process explicit, in the way in which it is explicit in the work (e.g., \cite{AbiteboulHV95,DeutschDiss,DeutschPods08,JohnsonK84})  on determining containment of queries in presence of %embedded 
dependencies. In particular, the paper \cite{ZhangM05}  does not introduce dependencies that look like %dependencies 
$\tau_V$ in Example~\ref{disj-neq-ex}. As a result, the authors of \cite{ZhangM05} do not have to deal with the (arguably inelegant) extensions of embedded dependencies to dependencies that may have disjunction and disequalities on the right-hand side. (Appendix~\ref{mendelzon-prelims-sec} provides some details of the approach of \cite{ZhangM05}.) 
		
In this current paper, when extending the approach of \cite{ZhangM05} to the case of dependencies holding on the instances of interest, it has proved convenient for us to make explicit the $MV$-induced dependencies, such as $\tau_V$ in Example~\ref{disj-neq-ex}. Thus, in this work we introduce (in Section~\ref{containment-dependencies-sec}) dependencies that have both disjunctions and disequalities on the right-hand side. Disequalities in dependencies are necessary in our approach for determining ${\cal M}\Sigma$-conditional containment, see Section~\ref{containment-procedure-sec}. (As a side note, we will see in Section~\ref{main-vv-dexchg-sec} that disequalities in dependencies are {\em not} necessary for essentially the same approach to work correctly when solving the problem of finding the set of certain answers to a CQ query w.r.t. a CQ weakly acyclic materialized-view setting.) 

Not surprisingly, for CQ weakly acyclic settings ${\cal M}\Sigma$ and CQ queries $Q_1$ and $Q_2$ of interest, $Q_1$ $\sqsubseteq_{{\cal M}\Sigma}$ $Q_2$ does not necessarily imply any of the following: 
\begin{itemize} 
			\item $Q_1$ $\sqsubseteq$ $Q_2$; 
\vspace{-0.1cm} 
			\item $Q_1$ $\sqsubseteq_{\Sigma}$ $Q_2$; and 
\vspace{-0.1cm} 
			\item $Q_1$ $\sqsubseteq_{{\cal M}\emptyset}$ $Q_2$; here, by ${\cal M}\emptyset$ we denote the result of replacing $\Sigma$ by $\emptyset$ in ${\cal M}\Sigma$.
\end{itemize} 

\vspace{-0.1cm} 

\noindent 
(See Appendix~\ref{sigma-noemptyset-ex-sec} for all the details.)

\subsection{The Dependencies and Chase Rules} 
\label{containment-dependencies-sec} 

We now introduce dependencies that are used in the algorithm of Section~\ref{containment-procedure-sec}. The input to each run of the algorithm is a triple of the form $({\cal M}\Sigma,Q_1,Q_2)$, with ${\cal M}\Sigma$ a CQ weakly acyclic setting, and $Q_1$ and $Q_2$ two CQ queries. We call such triples {\em CQ weakly acyclic input instances.} For each $({\cal M}\Sigma,Q_1,Q_2)$, the algorithm determines whether $Q_1 \sqsubseteq_{{\cal M}\Sigma} Q_2$ holds. To make the determination, a modification (via adding ${\cal C}^{exp}_{MV}$) of the query $Q_1$ is chased with the dependencies that we  introduce in the current subsection. 
%
%In Section~\ref{containment-procedure-sec} we will present a procedure for determining whether a CQ query $Q_1$ is ${\cal M}\Sigma$-contained in a CQ query $Q_2$ w.r.t. a CQ weakly acyclic materialized-view setting ${\cal M}\Sigma$. We will denote each (CQ weakly acyclic) problem instance by the triple $({\cal M}\Sigma,Q_1,Q_2)$. As discussed in Section~\ref{zmendelz-intuition-sec}, the intuition for the proposed procedure is that we use chase to transform the input query $Q_1$ into a query, $(Q_1)^{{\cal M}\Sigma}$, with two properties. First, $(Q_1)^{{\cal M}\Sigma}$ is equivalent to $Q_1$ on all the relevant instances. (The ``relevant instances'' are the $\Sigma$-valid base instances for the given $\cal V$ and $MV$.) Second, the definition of the query $(Q_1)^{{\cal M}\Sigma}$ ``exhibits the flavor of the relevant instances,'' as outlined in Section~\ref{zmendelz-intuition-sec}.  Due to these two properties, we can use a test for unconditional containment of $(Q_1)^{{\cal M}\Sigma}$ in $Q_2$ to correctly determine whether the original query $Q_1$ is contained in $Q_2$ w.r.t. all the relevant instances. (See Theorem~\ref{} in Section~\ref{} {\bf [[[ This theorem should be my main result for ZMendelzon containment ]]]}.) 
%
%Let ${\cal M}\Sigma$ be a CQ weakly acyclic materialized-view setting; 

\vspace{-0.1cm} 

\paragraph{Building blocks for the chase} 
All the dependencies used in Section~\ref{containment-procedure-sec} are constructed using the input CQ %weakly acyclic 
%materialized-view 
setting ${\cal M}\Sigma$. (For ease of exposition, in the remainder of this subsection we will assume that one such setting ${\cal M}\Sigma$ $=$ $(${\bf P}, $\Sigma$, $\cal V$, $MV)$ is fixed.)  
%set $\Sigma$ and the definitions of the views $\cal V$ in ${\cal M}\Sigma$. 
The construction uses {\em normalized} versions of conjunctions of relational atoms (see, e.g., \cite{ZhangO97}). That is, let $\phi$ be a conjunction of relational atoms. We replace in $\phi$ each duplicate occurrence of a variable or constant with a fresh distinct variable name. As we do each replacement, say of $X$ (or $c$) with $Y$, we add to the conjunction the equality atom $Y$ $=$ $X$ (or $Y$ $=$ $c$). As an illustration, if $\phi$ $=$ $P(X,X) \wedge S(c,c,X)$, then its normalized version is ${\phi}^{(n)}$ $=$ $P(X,Y) \wedge S(c,Z,W) \wedge Y = X \wedge Z = c \wedge W = X$. %Clearly, 
By construction, the normalized version of each $\phi$ is unique up to variable renamings. For the normalized version ${\phi}^{(n)}$ of a conjunction $\phi$, we will denote by ${\cal R}({\phi}^{(n)})$ the conjunction of all the relational atoms in ${\phi}^{(n)}$, and will denote by ${\cal E}({\phi}^{(n)})$ the conjunction of all the equality atoms in ${\phi}^{(n)}$. (If ${\phi}^{(n)}$ has no equality atoms, we set ${\cal E}({\phi}^{(n)})$ to $true$.) %Thus, ${\cal C}^{(n)}$ $=$ ${\cal R}({\cal C}^{(n)})$ $\wedge$ ${\cal E}({\cal C}^{(n)})$. 

A {\em non-egd} {\em (negd)} % for short)  
is a dependency of the form 
%
%\vspace{-0.2cm} 
%
\begin{equation} 
\label{negd-eqn}  
%\vspace{-0.1cm} 
\sigma: \phi(\bar{W}) \rightarrow X \neq Y.  
\end{equation} 

%\vspace{-0.1cm} 

\noindent 
Here, $\phi$ is a conjunction of relational atoms, and each of $X$ and $Y$ is an element of the set of variables $\bar W$. 

We also use %in Section~\ref{containment-procedure-sec} 
chase with ``implication constraints,'' %or Horn rules with  empty heads, 
see, e.g., \cite{ZhangO97}. An {\em implication constraint (ic)} is a dependency of the form $\tau: \phi(\bar{W}) \rightarrow false$, with $\phi(\bar{W})$ a conjunction of relational atoms. % (and $false$ a truth value). 

The algorithm of Section~\ref{containment-procedure-sec} performs chase of $CQ^{\neq}$ queries with ics, negds, egds, and tgds, by the  following rules. Let $Q$ be a  $CQ^{\neq}$ query. We say that {\em chase of $Q$ with an ic $\tau$ is applicable} whenever  there exists a homomorphism, $h$, from the antecedent $\phi$ of $\tau$ to the body of $Q$. Then we say that the {\em chase step of $Q$ with $\tau$ fails.}  Similarly, we say that a {\em chase step with a negd} $\sigma$ (as in Eq. (\ref{negd-eqn})) {\em applies} to $Q$ if there exists a homomorphism, $h$, from the antecedent $\phi$ of $\sigma$ to the body of $Q$. There are two cases: One, $h(X)$ and $h(Y)$ are the same variable (or the same constant) in $Q$. Then we say that {\em the chase step of $Q$ with $\sigma$ fails.} Otherwise, we form from $Q$ {\em the result} $Q^*$ of the chase step: $Q^*$ is a $CQ^{\neq}$ query obtained by conjoining  $body_{(Q)}$ with the atom $h(X) \neq h(Y)$. Chase steps with {\em tgds} are defined for $CQ^{\neq}$ queries in the same way as for CQ queries, see Section~\ref{dep-chase-prelims-sec}. Finally, for chase with {\em egds,} we extend the rules of Section~\ref{dep-chase-prelims-sec} by requiring that whenever chase of a $CQ^{\neq}$ query $Q$ with an egd $\tau$ is applicable, with some homomorphism $h$, and the consequent of $\tau$ is of the form $X$ $=$ $Y$, then {\em the chase step of $Q$ with $\tau$ fails} iff $body_{(Q)}$ has the %disequality 
atom $h(X) \neq h(Y)$ (or $h(Y) \neq h(X)$). (This generalizes the chase-step rule for CQ queries with egds, in the part where $h(X)$ and $h(Y)$ are distinct constants, see Section~\ref{dep-chase-prelims-sec}.) As we define $CQ^{\neq}$ queries as not having explicit equality atoms, our extended chase-step rules cover all possible cases for $CQ^{\neq}$ queries. 

\vspace{-0.1cm} 

\paragraph{Dependencies $\Phi_{(MV)}$  for CQ setting ${\cal M}\Sigma$} 
We now introduce one type of dependencies, {\em $MV$-induced dependencies} $\Phi_{(MV)}$, to be used in the chase in the algorithm of Section~\ref{containment-procedure-sec}. % to chase the input query $Q_1$. 
For the CQ setting ${\cal M}\Sigma$ with set $\cal V$ of views, 
let $V$ $\in$ $\cal V$ be a $k_V$-ary ($k_V$ $\geq$ $0$) view with definition $V({\bar X})$ $\leftarrow$ $\phi({\bar X},{\bar Y})$. We first normalize the body $\phi$ of $V$ into ${\cal R}(\phi^{(n)})$ $\wedge$ ${\cal E}(\phi^{(n)})$. %Note that 
The result \linebreak 
$\neg$ ${\cal E}(\phi^{(n)})$ of negating ${\cal E}(\phi^{(n)})$ is (obviously) a disjunction of disequality atoms. (E.g., $\neg$ $(X = Y \wedge Z = c)$ is $(X \neq Y$ $\vee$ $Z \neq c)$.) We now proceed for $V$ as follows. 

If $MV[V]$ $=$ $\emptyset$, we define the {\em $MV$-induced generalized implication constraint ($MV$-induced gic)} $\iota_V$ {\em for} $V$ as 
\begin{equation} 
\label{new-iota-eqn}  
\iota_V: \ {\cal R}(\phi^{(n)}) \rightarrow \ false \vee \neg {\cal E}(\phi^{(n)}).  
\end{equation}

% let $k_V$ $=$ $0$, and let $MV[V]$ be the set $\{ () \}$. Then we define the {\em $MV$-induced generalized egd} $\tau_V$ {\em for} $V$ as  

%\begin{tabbing} 
%$\tau_V:$ $\phi({\bar X},{\bar Y})$ $\rightarrow$ $true$.  
%\end{tabbing} 

%Finally, 

Now suppose $k_V$ $\geq$ $1$ and $MV[V]$ $=$ $\{ {\bar t}_1$, ${\bar t}_2$, $\ldots$, ${\bar t}_{m_V} \}$, with ${m_V}$ $\geq$ $1$. Then we define the {\em $MV$-induced generalized negd ($MV$-induced gnegd)} $\tau_V$ {\em for} $V$ as  
%
%\begin{tabbing} 
%$\tau_V:$ $\phi({\bar X},{\bar Y})$ $\rightarrow$ $\bigvee_{i=1}^{m_V} ({\bar X} = {\bar t}_i)$.  
%\end{tabbing} 
%\vspace{-0.1cm} 
\begin{equation} 
\label{new-tau-eqn}  
\tau_V: {\cal R}(\phi^{(n)}) \rightarrow \vee_{i=1}^{m_V} ({\bar X} = {\bar t}_i)  \vee \neg {\cal E}(\phi^{(n)}).  
\end{equation} 

%\vspace{-0.1cm} 

\noindent 
Here, ${\bar X}$ $=$ $[S_1,\ldots,S_{k_V}]$ is the head vector of the query for $V$, with $S_j$ $\in$ {\sc Const} $\cup$ {\sc Qvar} for  $j$ $\in$ $[1,$ $k_V]$. (By definition of ${\cal R}(\phi^{(n)})$, all the elements of $\bar X$ occur in ${\cal R}(\phi^{(n)})$.) For each $i$ $\in$ $[1,$ ${m_V}]$ and for the ground tuple $\bar t_i$ $=$ $(c_{i1},$ $\ldots,$ $c_{ik_V})$ $\in$   $MV[V]$, we abbreviate by ${\bar X} = {\bar t}_i$ the conjunction $\wedge_{j=1}^{k_V} (S_j = c_{ij})$. $MV$-induced gnegds are a straightforward generalization of disjunctive egds of  %as introduced in 
\cite{DeutschT01,FaginKMP05}, with negds added ``on top.'' %as explained below. 

For a CQ setting ${\cal M}\Sigma$ with set of view answers $MV$, {\em the set of $MV$-induced dependencies} $\Phi_{(MV)}$ {\em for} ${\cal M}\Sigma$ is the set of $MV$-induced gnegds and $MV$-induced gics constructed for all the views in ${\cal M}\Sigma$ as specified above.\footnote{We have shown that it is not necessary to use $MV$-induced dependencies for Boolean views $V$ with  %nonempty relations in 
$MV[V]$ $\neq$ $\emptyset$.}  

\vspace{-0.1cm} 

\paragraph{Dependencies $\Sigma_{(\neq)}$ for CQ setting ${\cal M}\Sigma$} 
We now outline how to obtain from the given CQ setting ${\cal M}\Sigma$ the second set of dependencies, $\Sigma_{(\neq)}$, to be used in chase in the algorithm of Section~\ref{containment-procedure-sec}. % to chase the input query $Q_1$. 
We convert each dependency in $\Sigma$ (in the given ${\cal M}\Sigma$) using a conversion rule that follows, and then produce $\Sigma_{(\neq)}$ as the union of the outputs. The conversion rule for a dependency $\sigma$ $\in$ $\Sigma$ of the form $\sigma: \phi({\bar X},{\bar Y}) \rightarrow \exists \bar{Z} \ \psi(\bar{X},\bar{Z})$ converts $\phi$ into ${\cal R}(\phi^{(n)})$ $\wedge$ ${\cal E}(\phi^{(n)})$, and then returns 

\vspace{-0.1cm} 

\begin{tabbing} 
$\sigma_{(\neq)}: {\cal R}(\phi^{(n)}) \rightarrow  \exists \bar{Z} \ \psi(\bar{X},\bar{Z}) \vee \neg {\cal E}(\phi^{(n)})$.  
\end{tabbing} 

\vspace{-0.3cm} 

%\noindent 
%(Observe that whenever $\sigma$ is an egd, $\sigma_{(\neq)}$ is a gnegd.) % See Example~\ref{disj-neq-ex} for an illustration.) 

%\vspace{-0.1cm} 

\paragraph{Chase of $CQ^{\neq}$ queries with $\Upsilon_{{\cal M}\Sigma}$ $=$ $\Phi_{(MV)}$ $\cup$ $\Sigma_{(\neq)}$} 
We now define chase of $CQ^{\neq}$ queries with the dependencies $\Upsilon_{{\cal M}\Sigma}$ $=$ $\Phi_{(MV)}$ $\cup$ $\Sigma_{(\neq)}$. For the fixed ${\cal M}\Sigma$, let 
$Q$ be a $CQ^{\neq}$ query over the schema {\bf P} in ${\cal M}\Sigma$. %Each chase step applied to $Q$ using a dependency in $\Upsilon_{{\cal M}\Sigma}$ results in a {\em set of} $CQ^{\neq}$ queries, by the rules as follows. 
Our definition of the chase steps can be seen as an extension of the definition of \cite{FaginKMP05} for their disjunctive egds, once we postulate that chase steps are to be applied to queries, rather than to instances as is done in \cite{FaginKMP05}. %Due to the page limit, we cannot provide all the details in the main text. Instead, we give 
Intuitively, we view each dependency $\upsilon$ $\in$ $\Upsilon_{{\cal M}\Sigma}$, of the form $\upsilon: \phi \rightarrow \psi_1 \vee \psi_2 \vee \ldots \vee \psi_m$, where each $\psi_i$ is a conjunction, as $m$ dependencies $\upsilon_1: \phi \rightarrow \psi_1$; $\ldots$; $\upsilon_m: \phi \rightarrow \psi_m$. Suppose there is a homomorphism, $h$, from the antecedent $\phi$ of $\upsilon$ to the query $Q$, and none of  $h(\psi_1)$, $h(\psi_2)$, $\ldots $, $h(\psi_m)$ is a tautology. Then we say that the chase step with $\upsilon$ applies to $Q$, and we output, as the result of the step, a set of $CQ^{\neq}$ queries such that each element of the set results from the application to $Q$ of one of $\upsilon_1$, $\ldots$, $\upsilon_m$, as defined above. %(When some $\upsilon_k$ is of the form $\phi \rightarrow \vee_{i=1}^{k} ({\bar X} = {\bar t}_i)$, see Eq. (\ref{new-tau-eqn}), we interpret $\upsilon_k$ as a set of $k$ egds, all of which must be applied to $Q$ in a chase step to produce a single element of the output. See Appendix~\ref{vv-dexchg-sec} for the details.) 

Whenever the chase step of $Q$ with $\upsilon_i$, for an $i$ $\in$ $[1,$ $m]$, fails (as is, e.g., always the case with an ic), then the chase step does not contribute anything to the output set. Thus, if the chase step of $Q$ fails with $\upsilon_i$ for {\em all} $i$ $\in$ $[1,$ $m]$, the output of the chase step of $Q$ with the (original) $\upsilon$ %$\in$ $\Upsilon_{{\cal M}\Sigma}$ 
is the empty set, i.e., a trivial $UCQ^{\neq}$ query. 

Once we have a formalization of chase steps of $CQ^{\neq}$ queries with dependencies $\Upsilon_{{\cal M}\Sigma}$, we can define {\em chase trees} and {\em chase results,} by generalizing the formalizations of \cite{FaginKMP05} of chase of instances with disjunctive egds. 
Due to the space limit, we are unable to provide detailed formalizations in the main text. (Example~\ref{disj-neq-ex} provides an illustration. Appendix~\ref{vv-dexchg-sec} has a detailed formalization of the special case where $\Upsilon_{{\cal M}\Sigma}$ does not contain any disequalities; an extension to disequalities is straightforward.) Intuitively, in a {\em chase tree} $\cal T$ constructed for a CQ setting ${\cal M}\Sigma$ and a $CQ^{\neq}$ query $Q$, the root represents $Q$, and each node represents either a $CQ^{\neq}$ query or (as a special case of a leaf) a trivial $UCQ^{\neq}$ query; we denote the node in this special case by $\epsilon$. A node $t$ in $\cal T$ has children $t_1$, $\ldots$, $t_k$ iff a chase step with some $\sigma$ $\in$ $\Upsilon_{{\cal M}\Sigma}$ applies to the $CQ^{\neq}$ query represented by $t$, and the result of the chase step is exactly all the queries represented by $t_1$, $\ldots$, $t_k$. A (non-$\epsilon$) node $t$ in $\cal T$ is a leaf iff no dependency in $\Upsilon_{{\cal M}\Sigma}$ applies in a chase step to the $CQ^{\neq}$ query represented by $t$. 

Each chase tree $\cal T$ can be associated with a   sequence (with repeated entries allowed) of dependencies in $\Upsilon_{{\cal M}\Sigma}$, according to the sequence of chase steps represented by $\cal T$ from the root downwards. The {\em result of the chase of} $Q$ with sequence $\sigma_1$, $\sigma_2$, $\ldots$ of dependencies in $\Upsilon_{{\cal M}\Sigma}$ is defined iff the associated $\cal T$ is a finite tree; then this result is either a trivial $UCQ^{\neq}$ query (iff each leaf of $\cal T$ is $\epsilon$), or is the union of all the $CQ^{\neq}$ queries represented by the leaves of $\cal T$. {\em A chase result of $Q$ with}  ${\cal M}\Sigma$, denoted $(Q)^{{\cal M}\Sigma}$, is the result (if defined) of the chase of $Q$ with any sequence of dependencies in $\Upsilon_{{\cal M}\Sigma}$. 

We now obtain the following result, in Proposition~\ref{ucq-tree-finite-prop}, for the case where ${\cal M}\Sigma$ is CQ weakly acyclic and  $Q$ is a CQ query. Let $MV$ be the set  of materialized views in ${\cal M}\Sigma$; then ${\cal C}^{exp}_{MV}$ is defined as in Section~\ref{probl-stmt-relationship-sec}. 
As is done in \cite{ZhangM05}, we denote by $Q'$ the CQ query obtained from $Q$ by conjoining the body of $Q$ with ${\cal C}^{exp}_{MV}$, after all the variables of ${\cal C}^{exp}_{MV}$ have been consistently renamed so that $Q$ and ${\cal C}^{exp}_{MV}$ do not share any variable names. We call $Q'$ {\em the ${\cal M}\Sigma$-expansion of} $Q$. 

\vspace{-0.1cm} 

\begin{proposition} 
\label{ucq-tree-finite-prop} % query-cont-algor-sec
Given  a CQ weakly acyclic setting ${\cal M}\Sigma$ and a CQ query $Q$: For the ${\cal M}\Sigma$-expansion $Q'$ of $Q$, each chase tree $\cal T$ for ${\cal M}\Sigma$ and $Q'$ is finite, of polynomial depth in the size of $Q$ and of $MV$ in ${\cal M}\Sigma$. Further, for any such $\cal T$ and for the $UCQ^{\neq}$ query $(Q)^{{\cal M}\Sigma}$ that is the result of the chase of $Q'$  with sequence of dependencies associated with $\cal T$, we have that: 
\begin{itemize} 
	\item The number of $CQ^{\neq}$ components of $(Q)^{{\cal M}\Sigma}$ is up to exponential in the size of $Q$ and $MV$, and 
\vspace{-0.1cm} 
	\item For each $CQ^{\neq}$ component, $q$, of $(Q)^{{\cal M}\Sigma}$, the size of $q$ is polynomial in the size of $Q$ and $MV$. 
\end{itemize} 
\vspace{-0.55cm} 
\end{proposition} 

\vspace{-0.1cm} 

The proof of Proposition~\ref{ucq-tree-finite-prop} is based on the results of \cite{FaginKMP05}, which construct a polynomial-size upper bound on the number of distinct values that can occur in chase of an instance with weakly acyclic tgds and egds. Appendix~\ref{vv-dexchg-sec} outlines a proof for a generalization over \cite{FaginKMP05}, in which a version of $\Upsilon_{{\cal M}\Sigma}$ is constructed without disequalities; the main observation is that $Q'$ already has all the constants that might be introduced in the chase by the $MV$-induced gnegds (as in Eq. (\ref{new-tau-eqn})) of $\Upsilon_{{\cal M}\Sigma}$. We then build on that result of Appendix~\ref{vv-dexchg-sec}, by observing that chase steps with {\em negds} do not add new values, and may add a number of disequality atoms that is only up to polynomial in the size of the given $Q$ and $MV$.

\subsection{The Containment-Checking Algorithm} 
\label{containment-procedure-sec} 

By Proposition~\ref{ucq-tree-finite-prop}, if a triple $({\cal M}\Sigma,Q_1,Q_2)$ is CQ weakly acyclic as defined in Section~\ref{containment-dependencies-sec}, then each chase tree for ${\cal M}\Sigma$ and $Q'_1$ is finite. Thus, the following procedure, given here by pseudocode, is an algorithm for CQ weakly acyclic inputs. (Testing whether $({\cal M}\Sigma,Q_1,Q_2)$ is CQ weakly acyclic can be done in polynomial time.) 

\mbox{} 

\noindent
{\em Algorithm} {\sc ${\cal M}\Sigma$-containment determination:} 

\noindent
{\bf Input:} CQ weakly acyclic instance $({\cal M}\Sigma,Q_1,Q_2)$. 

\noindent
{\bf Output:} Determination whether $Q_1$ $\sqsubseteq_{{\cal M}\Sigma}$ $Q_2$. 
%\vspace{-0.2cm}
%
%(Notation: For a mapping $\mu$ defined on a set of terms $T \neq \emptyset$ and for a set $S \subseteq T$, $\mu[S]$ denotes the restriction of $\mu$ to the domain $S$.) 
%
\begin{tabbing}
be \= ho \= be \= ho \= be \= ho \kill 
%1. Set $M$ to the set of multiset variables of $Q$; \\
1. Set $Q'_1$  to the ${\cal M}\Sigma$-expansion of $Q_1$; \\
%\> \> \> // Note that $Q^{min}$ $\in$ ${\cal Q}_{min}(Q)$  \\ 
%3. Set $S$ to the set of terms of $Q'$; \\
2. Obtain a chase result $(Q_1)^{{\cal M}\Sigma}$ of $Q'_1$ with $\Upsilon_{{\cal M}\Sigma}$;    \\
%\> $Q^{min}$ to itself such that $\mu(Q^{min})$ has fewer subgoals \\
%\> than $Q^{min}$) // Note that $\mu(Q^{min})$ $\in$ ${\cal Q}_{min}(Q)$  \\ 
%\> such that the image of $\mu[S]$ is a subset of $S$ \\
%\> {\bf begin} \\ 
%\> 4. If (the set of multiset variables of $Q'$ is $M$) then  \\ 
%\> \> \> \> {\bf begin} \\ 
3. {\em If ($(Q_1)^{{\cal M}\Sigma}$ is a trivial $UCQ^{\neq}$ query}  \\ 
4. \> {\em or} $(Q_1)^{{\cal M}\Sigma}$ $\sqsubseteq$ $Q_2$) {\em then} output ``yes''; {\em else} output ``no.''  
\end{tabbing} 

%\noindent 
(Recall that \cite{LevyMSS95} provides a containment test for the $UCQ^{\neq}$ query $(Q_1)^{{\cal M}\Sigma}$ and CQ query $Q_2$ in line 4.) %whether $(Q_1)^{{\cal M}\Sigma}$ $\sqsubseteq$ $Q_2$ using the results 

We now show that the algorithm {\sc ${\cal M}\Sigma$-containment determination} is correct for CQ weakly acyclic inputs. Our first observation is as follows. 

\vspace{-0.1cm} 

\begin{proposition} 
\label{zmendelz-containm-prop} 
For a CQ weakly acyclic ${\cal M}\Sigma$ $=$ $(${\bf P}, $\Sigma$, $\cal V$, $MV)$ and CQ query $Q$, let $(Q)^{{\cal M}\Sigma}$ be a chase result  of the ${\cal M}\Sigma$-expansion $Q'$ of $Q$ with ${\cal M}\Sigma$. Then: 
\begin{enumerate} 
	\item $(Q)^{{\cal M}\Sigma}$ $\sqsubseteq$ $Q$, and 
\vspace{-0.2cm} 
	\item $Q$ $\sqsubseteq_{{\cal M}\Sigma}$ $(Q)^{{\cal M}\Sigma}$. 
\end{enumerate} 
\vspace{-0.6cm} 
\end{proposition} 

\vspace{-0.1cm} 

The proof of item 2 of Proposition~\ref{zmendelz-containm-prop}  is by induction on the chase steps for $Q'$ and $\Upsilon_{{\cal M}\Sigma}$, once we fix an instance $I$ such that $\cal V$ $\Rightarrow_{I,{\Sigma}}$ $MV$. Specifically, the property in item 2 is an invariant for the output of each chase step of $Q'$ with $\Upsilon_{{\cal M}\Sigma}$, for any fixed such %instance 
$I$. % such that $\cal V$ $\Rightarrow_{I,{\Sigma}}$ $MV$. 

Our next observation concerns valuations for the query $(Q)^{{\cal M}\Sigma}$ and for {\em arbitrary} instances of schema {\bf P}. (The proof is by construction of each $UCQ^{\neq}$ query $(Q)^{{\cal M}\Sigma}$.) 

\vspace{-0.1cm} 

\begin{proposition} 
\label{zmendelz-valuation-prop} 
Given a CQ weakly acyclic setting ${\cal M}\Sigma$ $=$ $(${\bf P}, $\Sigma$, $\cal V$, $MV)$ and a CQ query $Q$. For any nontrivial chase result $(Q)^{{\cal M}\Sigma}$ of the ${\cal M}\Sigma$-expansion of $Q$ with ${\cal M}\Sigma$, denote by $B^*$ all the relational atoms in the body of $(Q)^{{\cal M}\Sigma}$. Then for every instance $I$ of schema {\bf P} and for each valuation $\nu$ for  $(Q)^{{\cal M}\Sigma}$ and $I$, $\nu(B^*)$ is a $\Sigma$-valid base instance for $\cal V$ and $MV$. %, call it $J$, such that $\cal V$ $\Rightarrow_{J,{\Sigma}}$ $MV$. 
\end{proposition} 

\vspace{-0.1cm} 

By Propositions \ref{zmendelz-containm-prop}--\ref{zmendelz-valuation-prop}, for CQ weakly acyclic instances $({\cal M}\Sigma,Q_1,Q_2)$, all the chase results $(Q_1)^{{\cal M}\Sigma}$ are trivial $UCQ^{\neq}$ queries (i.e., each of them is the empty set) iff (*) the answer to the input query $Q_1$ is empty on all $\Sigma$-valid base instances. Further, a $(Q_1)^{{\cal M}\Sigma}$ $=$ $\emptyset$  {\em only} if (*) holds. This justifies the ``yes'' output when the condition of line 3 of the algorithm  evaluates to true.

Propositions \ref{ucq-tree-finite-prop} through %, \ref{zmendelz-containm-prop},  and 
\ref{zmendelz-valuation-prop} permit us to establish 
correctness of the algorithm {\sc ${\cal M}\Sigma$-containment determination} for CQ weakly acyclic inputs: 

\vspace{-0.1cm} 

\begin{theorem} 
\label{main-mendelz-thm} 
Given a CQ weakly acyclic instance $({\cal M}\Sigma,Q_1,Q_2)$. Then $Q_1$ $\sqsubseteq_{{\cal M}\Sigma}$ $Q_2$ if and only if for any one chase result $(Q_1)^{{\cal M}\Sigma}$ of the ${\cal M}\Sigma$-expansion of $Q_1$ with ${\cal M}\Sigma$, either $(Q_1)^{{\cal M}\Sigma}$ $=$ $\emptyset$ or $(Q_1)^{{\cal M}\Sigma}$ $\sqsubseteq$ $Q_2$. 
\end{theorem} 

\vspace{-0.1cm} 

By this result, our solution to the certain-answer problem presented in Example~\ref{intro-main-three-ex}, for the tuple $(${\tt johnDoe}, {\tt 50000}$)$,  is correct for the setting of this example. (We solve that certain-answer problem via determining ${\cal M}\Sigma$-conditional containment, as stipulated in Theorem~\ref{problem-relationship-thm}.) 

As discussed earlier, the approach of \cite{ZhangM05} is exactly the algorithm {\sc ${\cal M}\Sigma$-containment determination} for the case $\Sigma$ $=$ $\emptyset$. The chase result $(Q_1)^{{\cal M}\Sigma}$, with $\Sigma$ $=$ $\emptyset$, is denoted in \cite{ZhangM05} by $Q''_1$. We have shown %(see Appendix~\ref{cannot-chase-mendelzon-sec}) 
that our extension of the approach of \cite{ZhangM05} to the cases where $\Sigma$ $\neq$ $\emptyset$ is not as simple as ``just chasing $Q''_1$ with the input dependencies $\Sigma$.'' In fact, even if we chase $Q''_1$ with our modified dependencies $\Sigma_{(\neq)}$, we are not guaranteed a correct output. (See Appendix~\ref{cannot-chase-mendelzon-sec} for all the details.) Thus, algorithm  {\sc ${\cal M}\Sigma$-containment determination} is not a trivial extension of the approach of \cite{ZhangM05}. 

We can also show that to chase the query $Q'_1$ with the dependencies $\Sigma$, and to then chase the resulting query with the  dependencies $\Phi_{(MV)}$, does not, in general, yield a correct determination of $Q_1$ $\sqsubseteq_{{\cal M}\Sigma}$ $Q_2$ when we apply the unconditional-containment ($\sqsubseteq$) test. However, it is by construction of the dependencies $\Sigma_{(\neq)}$ that chasing $Q'_1$ with $\Sigma_{(\neq)}$ (rather than $\Sigma$) only, followed by chase with the  dependencies $\Phi_{(MV)}$ only, yields correct chase results for the purpose of determining ${\cal M}\Sigma$-conditional containment for CQ weakly acyclic inputs. 

Finally, we note that the presence of disequality atoms is critical to ensure correctness of our algorithm. Specifically, if disequality atoms are not introduced into either $\Sigma_{(\neq)}$ or the dependencies $\Phi_{(MV)}$, then the result of Proposition \ref{zmendelz-valuation-prop} no longer holds. (See Appendix~\ref{deps-must-have-neqs-sec} for all the details.) As a result, %in the absence of negds 
it is no longer clear how to ensure that the only-if direction of Theorem~\ref{main-mendelz-thm} (in case where $(Q_1)^{{\cal M}\Sigma}$ $\neq$ $\emptyset$) goes through. 

\vspace{-0.2cm} 

\section{Finding all the Certain Answers} 
\label{main-vv-dexchg-sec} 

The results of Sections~\ref{problem-statement-sec}--\ref{query-cont-algor-sec} suggest an approach for %solving the problem of 
finding all certain-answer tuples for  CQ weakly acyclic inputs. For a $k$-ary query $Q$ and a setting ${\cal M}\Sigma$, the approach is to generate all the $k$-ary tuples of values in $consts({\cal M}\Sigma)$, and then for each such tuple, $\bar t$, to solve the certain-answer problem for $Q$, $\bar t$, and ${\cal M}\Sigma$, by using Theorem~\ref{problem-relationship-thm} and algorithm {\sc ${\cal M}\Sigma$-containment determination}. By our results above, this approach is a correct algorithm for  CQ weakly acyclic inputs. At the same time, its generate-and-test flavor may result in voluminous unnecessary computation for all those tuples $\bar t$ that are {\em not} certain answers for the given input. 

In this section we introduce an approach, called ``view-verified data exchange,'' which solves the same problem but is not based on the generate-and-test paradigm. As the name suggests, this approach is based on data exchange \cite{FaginKMP05,Barcelo09,LibkinDataExchange}. This approach is also intimately related to the techniques that we used in Section \ref{query-cont-algor-sec} to address ${\cal M}\Sigma$-conditional query containment. Specifically, view-verified data exchange uses a modification of the dependencies $\Upsilon_{{\cal M}\Sigma}$ of Section \ref{query-cont-algor-sec}, in which we do away with the disequality atoms in the dependencies. Due to the page limit, in this section we provide just a brief overview; all the details, including a full formalization and examples, can be found in Appendix~\ref{vv-dexchg-sec}. 

Given a CQ setting ${\cal M}\Sigma$ $=$ $(${\bf P}, $\Sigma$, $\cal V$, $MV)$, the idea of view-verified data exchange is very natural: We borrow from the standard data-exchange framework, in that we treat the relation symbols in $\cal V$ as the ``source schema'' and the schema {\bf P} as the ``target schema,'' with ``target constraints'' $\Sigma$. Further, we treat natural tgds arising from the definitions of the views in $\cal V$ as ``source-to-target dependencies'' $\Sigma_{st}$ for this ``data-exchange setting.'' Then we could treat the set $MV$ as a ``source instance,'' and pose the input query $Q$ on the ``target instances'' that are determined by this data-exchange setting  and by this source instance. (All the relevant formal definitions %of all these terms 
can be found in Appendices~\ref{dexchg-sec}--\ref{app-dexchg-sec}.) 

One special type of target instance used in data exchange is called ``canonical universal solution'' \cite{FaginKMP05} for the given data-exchange setting and source instance. Such instances are obtained by chase of the source instance with the dependencies $\Sigma_{st}$ $\cup$ $\Sigma$, and can be used to represent, in the following precise sense, all target instances of interest. % in the problem of finding certain answers to the query $Q$. 
When $\Sigma_{st}$ is a set of tgds, $\Sigma$ is weakly acyclic, and $Q$ is a UCQ query, the problem of computing certain answers for $Q$, w.r.t. the given data-exchange setting and source instance, can be solved %correctly 
via posing $Q$ on a canonical universal solution \cite{FaginKMP05}. It turns out that this result can be carried over directly to the problem of finding certain answers to a query in presence of a {\em materialized-view} setting, resulting in a sound and complete algorithm  \cite{StoffelS05} under {\em OWA} for the CQ weakly acyclic cases of the problem. 

Not surpisingly, the algorithm of \cite{StoffelS05} is %only sound but 
not complete under {\em CWA.} (See Appendices~\ref{dexchg-sec}--\ref{app-dexchg-sec} for the details.) In particular, applying the algorithm of \cite{StoffelS05} to our Example~\ref{intro-main-three-ex} would produce the empty set of certain-answer tuples. At the same time, using the results of Section~\ref{query-cont-algor-sec} we can show that $\bar t$ $=$ $(${\tt johnDoe}, {\tt 50000}$)$ {\em is} a certain answer for the setting of Example~\ref{intro-main-three-ex} under CWA.  As it would be straightforward for attackers to obtain that tuple $\bar t$ ``from first principles,'' our motivation was to come up with a correct algorithm for the CWA version of the problem of finding all certain query answers, as defined in Section~\ref{probl-stmt-defs-sec}. 
 Our view-verified data exchange does qualify, by being a sound and complete algorithm for all CQ weakly acyclic instances under CWA. 

We outline here the main idea of view-verified data exchange. (Due to the space limit, all the details can be found in Appendix~\ref{vv-dexchg-sec}.) Just as in the approach of \cite{StoffelS05}, we begin by obtaining a canonical universal solution, $J_{de}^{{\cal M}\Sigma}$, for the data-exchange setting that arises naturally from the input instance $({\cal M}\Sigma,Q)$. We then apply to $J_{de}^{{\cal M}\Sigma}$ {\em disjunctive chase,} as specified for our problem of ${\cal M}\Sigma$-conditional query containment, with two modifications. One, we chase the {\em instance} $J_{de}^{{\cal M}\Sigma}$, essentially by treating it as the body of a CQ query. % (as is done in our approach of Section~\ref{query-cont-algor-sec}). 
Two, we use in the chase a modification of the dependencies $\Upsilon_{{\cal M}\Sigma}$ $=$ $\Phi_{(MV)}$ $\cup$ $\Sigma_{(\neq)}$ of Section~\ref{query-cont-algor-sec}. The idea of this modification of $\Upsilon_{{\cal M}\Sigma}$ is that we do {\em not} normalize the left-hand side of any dependency. One consequence of this choice is that disequalities do not arise in the right-hand side of any resulting dependency. (In particular, $\Sigma$ remains unmodified, rather than giving rise to $\Sigma_{(\neq)}$ as in Section~\ref{query-cont-algor-sec}.) We show that disequalities are not necessary for correctness of the approach to the problem of finding certain-query answers. Intuitively, the instances that we obtain in the chase are used to characterize only $\Sigma$-valid instances for $\cal V$ and $MV$, rather than all possible instances of schema {\bf P}. (In the problem of Section~\ref{query-cont-algor-sec}, the chase enforces constraints that ensure that Proposition \ref{zmendelz-valuation-prop} holds for $(Q_1)^{{\cal M}\Sigma}$ on {\em all} instances of schema {\bf P}.) 

Finally, the view-verified data-exchange approach obtains a set of answers without nulls to the input query $Q$ on {\em each} of the instances in the chase result; the output is then the intersection of these sets. We have shown that for all CQ weakly acyclic inputs, the output of this approach is well defined and is the set of all certain answers to $Q$ w.r.t. the setting ${\cal M}\Sigma$. That is: 

\vspace{-0.1cm} 

\begin{theorem} 
View-verified data exchange is a sound and complete algorithm for finding certain answers for all CQ weakly acyclic instances under CWA. 
\end{theorem} 
 
 \vspace{-0.1cm} 

% In the next section we discuss an upper bound on the complexity of the algorithm. 

Interestingly, in view-verified data exchange one cannot always find all the certain answers correctly if one does the  chase ``in stages.'' That is, chase only with the input dependencies $\Sigma$, followed by chase only with the ``$MV$-induced dependencies,'' does not always yield a correct solution. The reverse order of the ``stages'' is not guaranteed to work %does not always work 
either. See Appendix~\ref{chase-not-staged-for-view-verif-app} for the details. 

\vspace{-0.2cm}

\section{Complexity of the Problems} 
\label{complexity-sec} 

In this section we consider the complexity of the CQ weakly acyclic cases of the three problems defined in Section~\ref{probl-stmt-defs-sec}. Our main focus is on %the results obtained under 
the security-relevant complexity measure introduced in \cite{ZhangM05}. Due to the page limit, the exposition in this section is just an outline of the results; Appendices~\ref{vv-dexchg-sec} and \ref{pi-p-2-hardness-proof-sec} provide the details. 

Generally, in studying the complexity of the certain-query-answer problem of Definition~\ref{certain-query-answer-def}, it is natural to build on the results of \cite{AbiteboulD98}, which were established w.r.t. %(two of) 
the complexity measures introduced in \cite{Vardi82}. For instance, for the CQ weakly acyclic case of the problem of Definition~\ref{certain-query-answer-def}, it is straightforward to obtain membership in coNP for the ``data complexity'' of the problem, that is, for the assumption that the set $MV$ is the only non-fixed part of %the instance 
$({\cal M}\Sigma,Q,{\bar t})$. Then one can use the coNP-hardness result of \cite{AbiteboulD98} for the special case $\Sigma$ $=$ $\emptyset$, to arrive at the overall coNP completeness of the CQ weakly acyclic case of the problem of Definition~\ref{certain-query-answer-def} w.r.t. the data-complexity measure of \cite{Vardi82}. 

Given the security focus of this current work, we concentrate here on a complexity measure that extends naturally that of \cite{ZhangM05}. Zhang and Mendelzon in \cite{ZhangM05} assumed for their ``conditional-containment'' problem that the base schema and the view definitions are fixed, where- as the set of view answers $MV$ and the queries posed on the base schema in presence of $MV$ can vary. (This assumption is natural in, e.g., database-access control  \cite{BertinoGK11}, where access-control views are typically defined once for each (class of) users, and where the only frequently changing parts of the problem instance would be the view answers, $MV$, seen by the users, as well as the ``secret queries'' $Q$.) \cite{ZhangM05} did not consider dependencies on the base schema; we follow the standard data-exchange assumption, see, e.g., \cite{FaginKMP05}, that the given dependencies are fixed, rather than being part of the problem input. 

Under this complexity metric, we consider first the complexity of the certain-query-answer problem (Definition~\ref{certain-query-answer-def}) and of the ${\cal M}\Sigma$-conditional containment problem (Definition~\ref{main-mendelz-def}). Given the tight relationship between these problems (see Theorem~\ref{problem-relationship-thm}), specifically between their CQ weakly acyclic cases, we can view the two problems together, using the following ``grid'': 
\begin{enumerate} 
	\item The CQ weakly acyclic case of the certain-query-answer problem with $\Sigma$ $=$ $\emptyset$; 

\vspace{-0.2cm} 

	\item The general (i.e., $\Sigma$ $\neq$ $\emptyset$ is possible) CQ weakly acyclic case of the certain-query-answer problem; 

\vspace{-0.2cm} 

	\item The CQ weakly acyclic case of the ${\cal M}\Sigma$-conditional-containment problem with $\Sigma$ $=$ $\emptyset$; and 

\vspace{-0.2cm} 

	\item The general (i.e., $\Sigma$ $\neq$ $\emptyset$ is possible) CQ weakly acyclic case of ${\cal M}\Sigma$-conditional containment.  

\end{enumerate} 

\noindent 
With the help of Theorem~\ref{problem-relationship-thm}, it is easy to show  that Problem 1 above is a special case of each of Problems 2 and 3, and that each of the latter problems is, in turn, a special case of Problem 4. 

Using these relationships, we have shown that each of Problems 1--4 is ${\Pi}^p_2$ complete w.r.t. our extension, above, of the complexity measure of \cite{ZhangM05}. These four results are immediate from the results of Theorems~\ref{complexity-hardness-for-certain-answers-thm}--\ref{complexity-membership-for-conditional-containment-thm}, to follow, and from our observations above on the inclusions between the four problems. 

\vspace{-0.1cm} 

\begin{theorem} 
\label{complexity-hardness-for-certain-answers-thm} 
The certain-query-answer problem of Definition~\ref{certain-query-answer-def} is ${\Pi}^p_2$ hard for CQ input instances $({\cal M}\Sigma,$$Q,$\linebreak ${\bar t})$ in which $\Sigma$ $=$ $\emptyset$ in the setting ${\cal M}\Sigma$, under the assumption that everything in the instance $({\cal M}\Sigma,Q,{\bar t})$ is fixed except for $Q$, $\bar t$, and the set $MV$ in ${\cal M}\Sigma$. 
\end{theorem} 

\vspace{-0.1cm} 

%\noindent 
(It is easy to show that in the setting of Theorem~\ref{complexity-hardness-for-certain-answers-thm}, it is enough to consider problem instances in which the size of the tuple $\bar t$ in $({\cal M}\Sigma,Q,{\bar t})$ is the arity of the query $Q$. See Appendix~\ref{pi-p-2-hardness-proof-sec} for the details.) 

The result of Theorem~\ref{complexity-hardness-for-certain-answers-thm} is by reduction from the $\forall$$\exists$-$CNF$ problem, which is known to be ${\Pi}^p_2$ complete \cite{Stockmeyer76}. %\footnote{Note that we cannot infer the ${\Pi}^p_2$ hardness of the problem just from the fact that our rewriting-based approach of Section~\ref{rewriting-sec}, which is sound and complete  for the case of $\Sigma$ $=$ $\emptyset$, uses a ${\Pi}^p_2$-complete containment test of \cite{ZhangM05}.} 
(Please see Appendix~\ref{pi-p-2-hardness-proof-sec}  for a detailed proof.) We start off from the reduction that was used by Millstein and colleagues in \cite{MillsteinHF03} for the problem of query containment for data-integration systems. We modify the reduction of \cite{MillsteinHF03} in the spirit that is similar to the modification of that reduction (of \cite{MillsteinHF03}) 
as suggested in \cite{ZhangM05}. (Recall that the full version of \cite{ZhangM05}, including any of its proofs, has never been published.) The goal of our modification is to comply with our assumptions about the input size, specifically with the assumption that the input view definitions are fixed. (In \cite{MillsteinHF03} it is assumed that both the queries and the view definitions can vary.) %The proof of Theorem~\ref{complexity-hardness-thm} can be found in Appendix~\ref{pi-p-2- 

\vspace{-0.1cm} 

\begin{theorem} 
\label{complexity-membership-for-conditional-containment-thm} 
The ${\cal M}\Sigma$-conditional containment\linebreak problem of Definition~\ref{main-mendelz-def} 
is in ${\Pi}^p_2$ for CQ weakly acyclic input instances  $({\cal M}\Sigma,Q_1,Q_2)$, under the assumption that everything in the instance $({\cal M}\Sigma,Q_1,Q_2)$ is fixed except  $Q_1$, $Q_2$, and the set $MV$ in ${\cal M}\Sigma$. % $({\cal M}\Sigma,Q,{\bar t})$. 
\end{theorem} 

\vspace{-0.1cm} 

(The proof is straightforward from the results of Section~\ref{query-cont-algor-sec}, specifically of Proposition~\ref{ucq-tree-finite-prop}.) 

Finally, consider the complexity of the CQ weakly acyclic case of the problem of finding all certain-answer tuples.  
Observe first that, in the special case where $Q$ is a Boolean query, the problem of finding all certain-answer tuples reduces to the certain-query-answer problem for the same ${\cal M}\Sigma$ and $Q$, with $\bar t$ $=$ $()$. %Thus, the upper bound on the time complexity of finding all certain-answer tuples for CQ weakly acyclic inputs is the tightest one can hope for, until the time when the  inclusion questions for the polynomial hierarchy get resolved.) 
Now recall that the view-verified data exchange %approach 
of Section~\ref{main-vv-dexchg-sec} is a sound and complete algorithm for the (general) CQ weakly acyclic case of this problem. Using this algorithm, we establish a singly-exponential upper bound on the time complexity %of the CQ weakly acyclic case 
 of the problem, under the same complexity measure as above, that is, assuming that in each instance $({\cal M}\Sigma,Q)$, everything is fixed except for the query $Q$ and for the set $MV$ in ${\cal M}\Sigma$. Further, under the same complexity measure, solving the CQ weakly acyclic case of the problem is in {\sc PSPACE} (provided the algorithm does certain things on-the-fly). Note that the output size is up to exponential in the arity of the input query $Q$. See Appendix~\ref{vv-dexchg-sec} for all the details.

%\vspace{-0.3cm} 

{\small 
\bibliographystyle{abbrv}
\bibliography{arxiv032014}  
}

} % end \mysize ghu

{\mysize % blue lagoon 

%\vspace{2cm}  

\newpage 

\appendix 

% end \nop alhambra 

\section{Certain Query Answers: \\ Example with $\Sigma$ $=$ $\emptyset$} 
\label{main-three-sec-app} 

In this appendix we show an example with $\Sigma$ $=$ $\emptyset$, of an input instance for the certain-query-answer problem of Definition~\ref{certain-query-answer-def} and for the problem of finding the set of  certain query answers w.r.t. a materialized-view setting, see Section~\ref{probl-stmt-defs-sec}. This example is to be used as an illustration in later appendices, e.g., in Appendix~\ref{app-rewriting-sec}. 

\begin{example} 
\label{app-main-three-ex} 
Consider a relation $E$ (for {\tt Employee}), which is used for storing information about employees of a company. Let the attributes of $E$ be {\tt Name}, {\tt Dept} (for the departments in which the employees work), %and {\tt Position} (of the employee in the department), 
and {\tt Salary}:  $E$({\tt Name, Dept, Salary}).

%\begin{tabbing} 
%{\tt Emp(Name, Dept, Salary)} 
%\end{tabbing} 

We assume that no integrity constraints hold on the database schema {\bf P} containing the relation $E$. (In particular, the only primary key of $E$ is all its attributes.) Thus, the set $\Sigma$ of dependencies holding on the schema {\bf P} is the empty set. 

Let a query {\tt Q} ask for the salaries of all the employees. We can formulate the query {\tt Q} in SQL as  

{\small 
\begin{verbatim} 
(Q): SELECT DISTINCT Name, Salary FROM E;   
\end{verbatim} 
} % end \small 

The query {\tt Q} is a CQ query, which can be expressed in Datalog as follows: 

\begin{tabbing} 
$Q(X, Z) \leftarrow E(X,Y,Z).$ 
\end{tabbing} 

%In this example, we use the schema {\bf P}, set $\Sigma$ of dependencies holding on {\bf P}, and secret query {\tt Q} of Example~\ref{main-ex}. Suppose a database user has access to a view, {\tt S}, that asks for the IDs, names, and salaries of the employees in the {\tt Sales} department.  
Consider two views, {\tt V} and {\tt W}, that are defined for some class(es) of users on the schema {\bf P}. The view {\tt V} returns the departments for each employee, and the view {\tt W} returns the salaries in each department. The Datalog definitions of these CQ views are as follows. (Please see Example~\ref{intro-main-three-ex} for the SQL definitions of {\tt V} and {\tt W}.)  

\begin{tabbing} 
$V(X, Y) \leftarrow E(X,Y,Z).$ \\ 
$W(Y, Z) \leftarrow E(X,Y,Z).$ 
\end{tabbing} 

Suppose that some user(s) are authorized to see the answers to {\tt V} and {\tt W}, and that at some point in time the user(s) can see the following set $MV$ of answers to these views. %materialized views $MV$ (i.e., of the answers to the queries {\tt V} and {\tt W} on a fixed instance of the {\tt Emp} relation) that the user can see  has one tuple in each view, as follows.  

\begin{tabbing} 
$MV$ $=$ $\{$ {\tt V(johnDoe,sales)}, {\tt W(sales,50000)}  $\} \ .$ 
\end{tabbing} 

Then one ``conjunctive fact-expression'' ${\cal C}_{MV}$ (see Section~\ref{probl-stmt-relationship-sec}) that the user(s) can put together based on this instance $MV$ is 

\begin{tabbing} 
${\cal C}_{MV}$ $=$ {\tt V(johnDoe,sales)} {\tt AND} {\tt W(sales,50000)}.  
\end{tabbing} 

%\noindent 
Let $\bar t$ $=$ $($  {\tt johnDoe}, {\tt 50000} $)$ be the tuple that the user hypothesizes is in the answer to the query {\tt Q} on all the instances of the relation {\tt Emp} that satisfy the (empty set of) dependencies $\Sigma$ and that generate the above instance $MV$. Observe that  the tuple $\bar t$ is made up from values {\tt johnDoe} and {\tt 50000}, which ``are generated by'' the expression ${\cal C}_{MV}$. Thus, knowing the associations between the values in the tuple $\bar t$ and the respective attribute names in ${MV}$, we can ``put together'' this expression ${\cal C}_{MV}$ and this tuple $\bar t$ as a SQL query, {\tt Rvw}, in terms of the views {\tt V} and {\tt W} and in presence of the constants from the instance $MV$, as follows: 

{\small 
\begin{verbatim} 
(Rvw): SELECT DISTINCT Name, Salary FROM V, W 
       WHERE Name = `johnDoe' AND V.Dept = W.Dept 
       AND V.Dept = `sales' AND Salary = `50000'; 
\end{verbatim} 
} % end \small 

\noindent 
That is, by defining the query {\tt Rvw} we formalize the rather natural process of the user ``putting together'' tuples in the available instance $MV$ and of his then using some of the values from the selected tuples to put forth a tuple of constants that is hypothetically in the answer to the query $Q$. We note that {\tt Rvw} is defined by a CQ query: 

\begin{tabbing} 
heheboomheheheboom he \= go \kill 
$R_{vw}(johnDoe,50000) \leftarrow V(johnDoe,sales),$ \\ 
\> $W(sales,50000).$ 
\end{tabbing}

By definition of {\tt Rvw}, the above tuple $\bar t$ $=$ $($  {\tt johnDoe}, {\tt 50000} $)$ is the only possible answer to  {\tt Rvw} on all possible instances of the relations {\tt V} and {\tt W}. It is easy to see that this answer to the query {\tt Rvw} is compatible with (i.e., can be obtained by asking the query {\tt Rvw} on) the above instance $MV$. (Intuitively, this is true because we have constructed {\tt Rvw} from the tuples in the above instance $MV$.) 
\end{example}

\section{Relationship between the cert- ain-query-answer problem w.r.t. a setting and the query-cont- ainment problem w.r.t. a setting} 
\label{app-rewriting-sec} 

In this appendix we provide the technical details on the main result of Section~\ref{probl-stmt-relationship-sec}. That result, Theorem~\ref{problem-relationship-thm}, establishes a direct relationship between the certain-query-answer problem for a given $Q$, $\bar t$, and a valid CQ setting ${\cal M}\Sigma$, %on the one hand, 
and the problem of ${\cal M}\Sigma$-conditional containment for $Q'$ and $Q$, % w.r.t. ${\cal M}\Sigma$, 
for the same $Q$ and ${\cal M}\Sigma$. Here, the query $Q'$ is constructed from the given $Q$, $\bar t$, and ${\cal M}\Sigma$. The proof of Theorem~\ref{problem-relationship-thm} is immediate from Theorem~\ref{rewr-approach-thm}, see Section~\ref{one-rewr-enough-sec} of this appendix. 

A relationship similar to that of Theorem~\ref{problem-relationship-thm} was observed in \cite{AbiteboulD98} %for finding certain query answers 
for the dependency-free case under OWA. % and in the absence of dependencies. 
In contrast, our result holds under CWA and in presence of dependencies on the schema {\bf P} in the setting ${\cal M}\Sigma$.   

\subsection{The Intuition} 
\label{intu-sec} 

The intuition for the relationship between the two problems can be illustrated via Example~\ref{app-main-three-ex}. 
That is, we formalize the  thought process of the presumed attackers concerning the answers to ``secret queries'' \cite{MiklauS07} (such as the query $Q$ in Example~\ref{app-main-three-ex}) that are posed on a a proprietary database. The attackers know a materialized-view setting ${{\cal M}\Sigma}$ $=$ ({\bf P}, $\Sigma$, $\cal V$, $MV$) and the definition of a query $Q$, and come up with candidate certain answers to $Q$ in this setting ${\cal M}\Sigma$. (This is, informally, the idea of the problem of database access control, see, e.g.,  \cite{BertinoGK11}.) 

We argue that the approach of putting together the tuples in the view answers is natural for the presumed 
attackers to use. Indeed, in any specific 
instance of the problem of the certain query answer w.r.t. a materialized-view setting, attackers deal directly with a ground instance $MV$. They know that $MV$ is a set of answers to the ``access-policy'' views $\cal V$ on the underlying instance of interest. Thus, intuitively, a question that is natural for the attackers to ask is which values in $consts({\cal M}\Sigma)$ can be put together to form an answer to the secret query $Q$, on all possible underlying instances of interest. (In general, the attackers could consider in their pursuit not just values in $consts({\cal M}\Sigma)$, but also constants mentioned in the queries for $\cal V$ and in the secret query $Q$. It is straightforward to reflect this in our setting, by adding extra head arguments to the definitions of the respective views. Thus, we do not explicitly consider this extension in this paper.)

How can this certain-query-answer question be answered deterministically, as required by %the definition of information leak (
Definition~\ref{certain-query-answer-def}? A natural approach would be to put together a query, call it $R$, in terms of the relations in the instance $MV$, and to then prove that $R$ is ``contained,'' in some precise sense (in particular, w.r.t. the views in $\cal V$), in the secret query $Q$. We will be referring  to all queries $R$ over schema $\cal V$ %are traditionally called 
as ``rewritings'' (in terms of $\cal V$), as indeed they would be defined in terms of the relation symbols in $\cal V$. (Another reason to refer to such queries $R$ as ``rewritings'' is that we will need to define their expansions shortly.) Hence our name for this approach to solving the problem of whether a ground tuple $\bar t$ is a certain answer to a query $Q$ w.r.t. a setting ${\cal M}\Sigma$. As we will see, this approach determines precisely containment between two queries w.r.t. the given setting ${\cal M}\Sigma$. Here, one of the two queries in question is the query $Q$ provided in the problem input, and the other query is the ``expansion'' \cite{LevyMSS95} of a  rewriting $R$, with head $\bar t$, such that the definition of $R$ is obtained from  ${\cal M}\Sigma$. %; it formalizes the attackers'  thought process that was discussed in Example~\ref{main-three-ex}. 

A challenge arises immediately when attackers pursue this train of thought: In the set $R(MV)$ of answers to a rewriting $R$ on an instance $MV$, not all the tuples in $R(MV)$ would necessarily be in the answer to the secret query $Q$. That is, the formal containment that we %attackers 
are looking for would not hold for all rewritings $R$. (As an illustration, suppose that in some instance of the certain-query-answer problem w.r.t. a setting, the input query $Q$ returns names of employees with high salaries, and $R$ returns names of employees in the accounting department. Clearly, the answer to $R$ is not necessarily a subset of the answer to $Q$, on any particular database of interest.) 

At the same time, for each individual tuple $\bar t$ $\in$ $R(MV)$, it makes sense to ask the question of whether {\em the query} $R({\bar t})$ is contained in $Q$ in the appropriate precise sense. The intuition is that $R({\bar t})$ is the result of binding the head vector of $R$ to a tuple, $\bar t$, in the relation $R(MV)$; as a result, $\bar t$ is the only answer to $R({\bar t})$ on the instance $MV$. We focus on such rewritings $R({\bar t})$ in this approach. %, and give the formal definitions in the next subsection. 

\subsection{Defining the Rewriting Approach}  
\label{rewr-approach-subsec} 

We now formalize the ``rewriting approach'' to determining whether a given ground tuple $\bar t$ is a certain answer to a given query $Q$ w.r.t. a given materialized-view setting ${\cal M}\Sigma$. As outlined in Section~\ref{intu-sec}, the intuition is that this rewriting approach works by determining containment between two queries w.r.t. the setting ${\cal M}\Sigma$, such that each of the two queries is obtained from some combination of the given inputs $Q$, $\bar t$, and ${\cal M}\Sigma$. Our intent is to tie the definitions of rewritings that attackers can formulate on view answers, to components of the given setting ${\cal M}\Sigma$. After defining rewritings of the form $R({\bar t})$, we recall the standard notion of {\em expansion} of a view-based rewriting \cite{LevyMSS95}; an expansion of a rewriting is its equivalent reformulation over the schema {\bf P} used to define the query $Q$. We then formalize the rewriting approach, using the notion of containment of queries over the same schema w.r.t. a set of view answers $MV$ and a set of dependencies $\Sigma$.

%\newpage 

{\bf Head-instantiated rewriting $R({\bar t})$.} Intuitively, an attackers' goal in this approach is to form candidate answers, $\bar t$, to the secret query $Q$, by using constants that are in $consts({\cal M}\Sigma)$ and that thus presumably originate from the actual instance $I$ of interest, $\cal V$ $\Rightarrow_{I,\Sigma}$ $MV$. (That is, the instance $I$ is the actual proprietary database, of interest to the attackers, that has been used to generate the instance $MV$.) Observe that not all $\cal V$-based rewritings could be used toward this goal. Consider, for instance, a rewriting $R_f(f) \leftarrow V(X)$, defined using a constant $f$ and a subgoal $V(X)$ for a view $V$ and variable $X$. Clearly, regardless of the contents of the set $MV$ of answers to the view $V$, the answer $R_f(f)(MV)$ to $R_f$ on $MV$ is always the set $\{$ $(f)$ $\}$. To rule out rewritings such as $R_f(f)$, we define a desirable type of rewritings as follows.

For an integer $k$ $\geq$ $0$ and for a $k$-tuple $\bar t$ of constants, consider a safe $k$-ary CQ query $R$ over schema $\cal V$ and with head vector $\bar t$. We say that $R$ is a {\em head-instantiated rewriting for} $\bar t$ iff there exists a safe $k$-ary CQ query $R^{(g)}({\bar X})$ over the schema $\cal V$, called a {\em grounding rewriting for} $R$, that satisfies two conditions. First, the head vector $\bar X$ of $R^{(g)}$ does not include constants. Second, there exists a mapping, $h$, that maps all the elements of $\bar X$ to constants and that maps the remaining terms in $R^{(g)}$ to themselves, such that the rewriting that results from applying $h$ to the definition of $R^{(g)}$ is exactly $R$. 

\begin{example} 
\label{grounding-ex} 
Consider rewritings $R_{vw}$ and $\tilde{R}_{vw}$ that use constants $c$, $d$, and $f$. ($\tilde{R}_{vw}$ also uses a variable $Z$.) 
\begin{tabbing} 
$R_{vw}(c,f) \leftarrow V(c,d), W(d,f) .$ \\ 
$\tilde{R}_{vw}(c,f) \leftarrow V(c,Z), W(Z,f) .$
\end{tabbing}   

\noindent 
Suppose that $c$, $d$, and $f$ stand for `$johnDoe$', `$sales$', and `$50000$', respectively; then $R_{vw}$ is the rewriting {\tt Rvw} of Example~\ref{app-main-three-ex}. By applying this ``translation of constants'' to the instance $MV$ of Example~\ref{app-main-three-ex}, we obtain an instance $MV$ $=$ $\{ V(c,d), W(d,f) \}$. 

Each of $R_{vw}$ and $\tilde{R}_{vw}$ is a head-instantiated rewriting for $(c,f)$, as the respective grounding rewritings are 
\begin{tabbing} 
$R^{(g)}_{vw}(X,Y) \leftarrow V(X,d), W(d,Y) .$ \\ 
$\tilde{R}^{(g)}_{vw}(X,Y) \leftarrow V(X,Z), W(Z,Y) .$
\end{tabbing}   
\vspace{-0.6cm} 
\end{example} 

By definition, for each head-instantiated rewriting $R$ for a tuple $\bar t$, the answer to $R$ on an instance $I$ of schema $\cal V$ is nonempty (and is exactly $\{$ $\bar t$ $\}$) iff there exists a valuation from the body of $R$ onto a subset $I'$ of $I$ such that $adom(I')$ contains all the constants in $\bar t$. Further, consider an arbitrary safe CQ query $R''$ over schema $\cal V$, and consider any instance $MV$ such that $R''(MV)$ $\neq$ $\emptyset$. Then for each tuple $\bar t$ in $R''(MV)$, the result $R''({\bar t})$ of binding the head vector of $R''$ to $\bar t$ (while consistently renaming the terms in the body of $R''$ as well) is a head-instantiated rewriting for $\bar t$, such that the answer to $R''({\bar t})$ on the instance $MV$ is not empty (and is, obviously, exactly $\{ {\bar t} \}$). 

{\bf Expansion of a rewriting.} We now take a step back, from head-instantiated rewritings to general CQ rewritings, to recall the standard notion of expansion of a CQ rewriting  \cite{LevyMSS95}. %Let $R$ be a CQ rewriting defined in terms of a set of CQ views $\cal V$ over schema {\bf P}. 
First, given a set of views $\cal V$ and a ground instance $I$ of schema {\bf P}, consider an instance over schema $\cal V$ $\cup$ {\bf P}, which results from adding to $I$ the relation $V(I)$ for each relation symbol $V$ $\in$ $\cal V$. We call the latter instance {\em the $\cal V$-enhancement of $I$,} and denote it by $I^{(+{\cal V})}$. Now given a rewriting $R$ over the schema $\cal V$, consider a query, $R'$, over the schema {\bf P} such that for each instance $I$ of {\bf P} we have $R'(I)$ $=$ $R(I^{(+{\cal V})})$. We call such a query $R'$ {\em an expansion of} $R$ {\em (over} {\bf P}), and denote it by $R^{exp}$. We will use the following straightforward %but important 
property of $R^{exp}$: % (the proof is immediate from the definitions): 

\begin{proposition} 
\label{expansion-prop} 
For a set $\cal V$ of views over schema {\bf P}: Let $R$ be a query over $\cal V$ such that $R^{exp}$ exists, and let $MV$ be an instance of schema $\cal V$. Then for each instance $I$ of schema {\bf P} such that $\cal V$ $\Rightarrow_{I,\emptyset}$ $MV$, we have  $R^{exp}(I)$ $=$ $R(MV)$. 
\end{proposition} 

%In case where $\cal V$ is a set of CQ views and where $R$ is a CQ rewriting over $\cal V$, the standard process in the literature of constructing $R^{exp}$ is to replace each subgoal of $R$ with the body of the query for the corresponding relation symbol in $\cal V$. First, given a set of views $\cal V$ and a ground instance $I$ of schema {\bf P}, consider an instance over schema $\cal V$ $\cup$ {\bf P}, which results from adding to $I$ the relation $V(I)$ for each relation symbol $V$ $\in$ $\cal V$. We call the latter instance {\em a $\cal V$-enhanced instance $I$ of} {\bf P}, and denote it by $I^{(+{\cal V})}$. 

%Now given a query $R$ over the schema $\cal V$, consider a query, $R'$, over schema {\bf P} such that for each instance $I$ of {\bf P} we have $R'(I)$ $=$ $R(I^{(+{\cal V})})$. We call such a query $R'$ {\em an expansion of} $R$ {\em (over} {\bf P}), and denote it by $R^{exp}$. 

In case where $\cal V$ is a set of CQ views and $R$ is a CQ query over $\cal V$, the standard process in the literature of constructing $R^{exp}$  is \cite{LevyMSS95} to replace each subgoal of $R$ with the body of the query for the corresponding relation symbol in $\cal V$. In this process, care is taken to perform two operations on each query, of the form $V({\bar X})$ $\leftarrow$ $body_{(V)}$, whose body in $R^{exp}$ corresponds to a subgoal of $R$ of the form $V({\bar Z})$. First, we {\em bind the arguments of the query for} $V$ {\em to the vector} $\bar Z$, % (of the subgoal $s_V({\bar Z})$ of $R$), 
in two steps, (A) and (B). The step  (A) is to extend the homomorphism,\footnote{It is easy to show that if such a $h$ cannot be constructed, then $R$ is unsatisfiable on all instances of the schema $\cal V$.} $h$, that maps each element of the head vector $\bar X$ of the query for $V$ to the same-position element of $\bar Z$, to a homomorphism $h_{V({\bar Z})}$, whose domain is the set of all arguments of $body_{(V)}$, such that $h_{V({\bar Z})}$ is the identity mapping for each value that is not in the domain of $h$. Then, (B) is to apply $h_{V({\bar Z})}$ to $body_{(V)}$, with conjunction  of relational atoms $h_{V({\bar Z})}(body_{(V)})$ as the output. 
Second, before conjoining $h_{V({\bar Z})}(body_{(V)})$ with the current body, $bodyR^{exp}_{curr}$, of the query $R^{exp}$, we rename all the variables in $h_{V({\bar Z})}(body_{(V)})$ consistently into ``fresh'' variables not occuring in $bodyR^{exp}_{curr}$. The query $R^{exp}$ that is obtained by this two-step process is (i) an expansion of $R$ over {\bf P}, and is (ii) unique up to variable renaming.  

{\bf Conditional containment:} We can now use containment to directly relate a rewriting $R$, via $R^{exp}$, to the given query $Q$. The notion of containment we will use is that of Definition~\ref{main-mendelz-def} in Section~\ref{probl-stmt-defs-sec}. 

For notational convenience in the results to follow, we now introduce $\Sigma$-conditional containment of a rewriting in a query modulo a set of views and a set of answers to the views: For a rewriting $R$ over $\cal V$ such that $R^{exp}$ exists, and for a query $Q$ over {\bf P}, we say that $R$ {\em  is $\Sigma$-conditionally contained in} $Q$ {\em w.r.t.} $MV$ {\em and modulo} $\cal V$, denoted $R$ $\sqsubseteq_{\Sigma,MV,{\cal V}}$ $Q$, iff $R^{exp}$ $\sqsubseteq_{{\cal M}\Sigma}$ $Q$ holds. 

{\bf The rewriting approach:} We are now ready to specify the rewriting approach to the problem of determining whether a tuple $\bar t$ is a certain answer to a query $Q$ w.r.t. a materialized-view setting ${\cal M}\Sigma$. For an instance $MV$ of schema $\cal V$, we say that a head-instantiated rewriting $R({\bar t})$ is {\em $MV$-validated} iff the set $R({\bar t})(MV)$ is not the empty set. (The rewritings $R_{vw}$ and $\tilde{R}_{vw}$ of Example~\ref{grounding-ex} are both $MV$-validated.) Given a valid\footnote{The view-verified data-exchange approach of Appendix~\ref{vv-dexchg-sec} can be used as a sound and complete algorithm for determining whether a given CQ weakly acyclic materialized-view setting is valid.} materialized-view setting ${\cal M}\Sigma$ with set of view answers $MV$, a $k$-ary ($k$ $\geq$ $0$) query $Q$, and a $k$-ary tuple $\bar t$ of constants in $consts({\cal M}\Sigma)$, the {\em rewriting approach to the certain-query-answer problem for $Q$ and $\bar t$ in} ${\cal M}\Sigma$ is to find an $MV$-validated head-instantiated rewriting $R$ for $\bar t$ such that $R$ $\sqsubseteq_{\Sigma,MV,{\cal V}}$ $Q$. This approach is sound: 

\begin{proposition} 
\label{rewr-approach-prop}  
Given a valid materialized-view setting ${\cal M}\Sigma$ $=$ $(${\bf P}, $\Sigma$, $\cal V$, $MV)$ and a query $Q$ of arity $k$ $\geq$ $0$. Let $\bar t$ be a $k$-tuple of values from $consts({\cal M}\Sigma)$. Suppose that there exists an $MV$-validated head-instantiated rewriting $R$ for $\bar t$ such that $R$ $\sqsubseteq_{\Sigma,MV,{\cal V}}$ $Q$. Then $\bar t$ is a certain answer to $Q$ w.r.t. ${\cal M}\Sigma$. 
\end{proposition} 

The proof is very simple: Any rewriting $R$ satisfying the conditions of Proposition~\ref{rewr-approach-prop} must have $\bar t$ as its only answer on the given instance $MV$. Thus, by Proposition~\ref{expansion-prop}, $R^{exp}$ (which exists because the containment $R$ $\sqsubseteq_{\Sigma,MV,{\cal V}}$ $Q$ is stated in Proposition~\ref{rewr-approach-prop} to be well defined)  has $\bar t$ as its only answer on all instances $I$ such that $\cal V$ $\Rightarrow_{I,\Sigma}$ $MV$. From the containment $R^{exp}$ $\sqsubseteq_{{\cal M}\Sigma}$ $Q$ we conclude that on all such instances $I$, the tuple $\bar t$ is an element of the set $Q(I)$. The claim of Proposition~\ref{rewr-approach-prop}  follows. 
%\newpage 

%{\bf [[[ Stopped here S07/06/13 ]]]} 

\subsection{One Rewriting Is Enough} 
\label{one-rewr-enough-sec} 

Suppose that we are given a materialized-view setting ${\cal M}\Sigma$ and a $k$-ary ($k$ $\geq$ $0$) query $Q$. One (e.g., attackers) can generate all $k$-tuples $\bar t$ with values in $consts({\cal M}\Sigma)$. Then, Proposition~\ref{rewr-approach-prop} gives the attackers a tool for testing each such $\bar t$ as a %candidate 
certain-answer tuple to $Q$ w.r.t. ${\cal M}\Sigma$, assuming that the attackers can come up with an ``appropriate'' rewriting $R$ for each $\bar t$, and that there exists an algorithm for checking the containment $R$ $\sqsubseteq_{\Sigma,MV,{\cal V}}$ $Q$ for each such $R$ and $\bar t$. We will consider in the next subsection some such algorithms. However, in this current subsection we show that to solve this generate-and-test problem for a given instance ${\cal M}\Sigma$, % of the information-leak problem, 
it is not necessary to also generate various bodies for rewritings $R$.  Each valid ${\cal M}\Sigma$ is associated with a single CQ rewriting for each $\bar t$, with all these rewritings (for ${\cal M}\Sigma$) having the same body. % which we call $R^{(MV)}_{max}$, such that $R^{(MV)}_{max}$ alone
The main result of this subsection is %in Theorem~\ref{rewr-approach-correct-thm} 
that  for all CQ %weakly acyclic 
 instances ${\cal M}\Sigma$, these rewritings alone can be used to  capture {\em exactly} the set of all certain answers to the input query. 

Intuitively, we are to construct the desired rewritings from the facts in the instance $MV$ given as part of ${\cal M}\Sigma$. Indeed, by the requirement that $MV$ in each ${\cal M}\Sigma$ be a ground instance, each fact in $MV$ can be viewed equivalently as a relational atom whose all arguments are constants. Given a fixed $MV$ and a $k$-ary ($k$ $\geq$ $0$) tuple $\bar t$ of values from $consts({\cal M}\Sigma)$, we say that a CQ rewriting $R$ with head vector $\bar t$ is an {\em $MV$-induced rewriting for} $\bar t$ iff each subgoal of $R$ is a fact in $MV$. Further, an $MV$-induced rewriting $R$ for $\bar t$ is a {\em maximal $MV$-induced rewriting for} $\bar t$ iff each fact in $MV$ is also a subgoal of $R$. In Example~\ref{grounding-ex}, $R_{vw}$ is a maximal $MV$-induced rewriting for the tuple $(c,f)$, and $\tilde{R}_{vw}$ is not an $MV$-induced rewriting. 

We now list useful properties of $MV$-induced rewritings.

\begin{proposition} 
\label{mv-induced-properties-prop} 
Given a valid materialized-view setting ${\cal M}\Sigma$ $=$ $(${\bf P}, $\Sigma$, $\cal V$, $MV)$. For a $k$ $\geq$ $0$, let $\bar t$, ${\bar t}_1$, and ${\bar t}_2$  be $k$-tuples of values from $consts({\cal M}\Sigma)$, for the $MV$ in ${\cal M}\Sigma$. Then:
\begin{itemize} 
	\item[(1)] Each $MV$-induced rewriting $R$ for $\bar t$ is an $MV$-validated head-instantiated rewriting for $\bar t$ whenever each element of $\bar t$ occurs in the body of $R$; %; that is, $R(MV)$ $=$ $\{ ({\bar t}) \}$ for each such $R$. 
	\item[(2)] For each $\bar t$, there is exactly one maximal $MV$-induced rewriting, which is an $MV$-validated head-instantiated rewriting for $\bar t$; and  
	\item[(3)] The maximal $MV$-induced rewritings for ${\bar t}_1$ and for ${\bar t}_2$ have the same body, for all choices of ${\bar t}_1$ and ${\bar t}_2$.  
\end{itemize} 
%\vspace{-0.55cm} 
\end{proposition} 

{\em Note 1.} In case where some constants in $consts({\cal M}\Sigma)$ are in definitions of the views in $\cal V$ but are not in $MV$, all the claims of Proposition~\ref{mv-induced-properties-prop} still go through once we modify the view definitions by adding all their body constants into their head vectors. This fix for this case also works for all the other results of this appendix that deal with head-instantiated rewritings for tuples $\bar t$ constructed from the elements of the set $consts({\cal M}\Sigma)$. 

The next result says that when we have the maximal $MV$-induced rewriting for some tuple $\bar t$ of values from $consts({\cal M}\Sigma)$, then we do not need to consider any other head-instantiated rewritings for $\bar t$ in our rewriting approach. (The proof is straightforward and is omitted.) 

%Given a valid materialized-view setting ${\cal M}\Sigma$ $=$ $(${\bf P}, $\Sigma$, $\cal V$, $MV)$ and a query $Q$ of arity $k$ $\geq$ $0$. Let $\bar t$ be a $k$-tuple of values from $consts({\cal M}\Sigma)$. Suppose that there exists an $MV$-validated head-instantiated rewriting $R$ for $\bar t$ such that $R$ $\sqsubseteq_{\Sigma,MV,{\cal V}}$ $Q$. Then $\bar t$ is a certain answer to $Q$ w.r.t. ${\cal M}\Sigma$. 

\begin{proposition} 
\label{need-only-max-mv-rewr-prop} 
Given a valid CQ materialized-view setting ${\cal M}\Sigma$ $=$ $(${\bf P}, $\Sigma$, $\cal V$, $MV)$ and a query $Q$ of arity $k$ $\geq$ $0$. Let $\bar t$  be a $k$-tuple of values from $consts({\cal M}\Sigma)$. %, for the $MV$ in ${\cal M}\Sigma$. 
Let $R$ be an $MV$-validated head-instantiated rewriting for $\bar t$ such that $R$ $\sqsubseteq_{\Sigma,MV,{\cal V}}$ $Q$. Then for the maximal $MV$-induced rewriting $R^*_{\bar t}$ for $\bar t$, we have $R^*_{\bar t}$ $\sqsubseteq_{\Sigma,MV,{\cal V}}$ $Q$. 
\end{proposition} 

The following result  says that maximal $MV$-induced rewritings alone can be used to  capture exactly the certain answers to %CQ 
queries w.r.t.  CQ materialized-view settings. This result is an immediate corollary of Propositions~\ref{mv-induced-properties-prop} and~\ref{need-only-max-mv-rewr-prop}. 

\begin{theorem} 
\label{rewr-approach-correct-thm} 
Given a valid CQ materialized-view setting ${\cal M}\Sigma$ and a %CQ 
query $Q$ of arity $k$ $\geq$ $0$. For a $k$-tuple $\bar t$ of values from $consts({\cal M}\Sigma)$: The tuple $\bar t$ is a certain answer to $Q$ w.r.t. ${\cal M}\Sigma$ iff for the maximal $MV$-induced rewriting $R^*_{\bar t}$ for $\bar t$, we have $R^*_{\bar t}$ $\sqsubseteq_{\Sigma,MV,{\cal V}}$ $Q$. 
\end{theorem}

\begin{proof} 
{\em If:} The proof of this direction parallels the proof of Proposition~\ref{rewr-approach-prop}. %{\bf [[[ Need to double check ]]]} 

{\em Only-If:} By Definition~\ref{certain-query-answer-def}, for the given tuple $\bar t$ we have that $\bar t$ is in the set $Q(I)$ for all instances $I$ of schema {\bf P} such that $\cal V$ $\Rightarrow_{I,\Sigma}$ $MV$. By Proposition~\ref{mv-induced-properties-prop}, we have that  $\bar t$ is the only answer on the instance $MV$ to the maximal $MV$-induced rewriting $R^*_{\bar t}$ for $\bar t$. Thus, for the expansion of $R^*_{\bar t}$, denote this expansion by $(R^*_{\bar t})^{exp}$, we have by Proposition~\ref{expansion-prop}  that 
 for each instance $I$ of schema {\bf P} such that $\cal V$ $\Rightarrow_{I,\Sigma}$ $MV$, we have  $(R^*_{\bar t})^{exp}(I)$ $=$ $\{ ({\bar t}) \}$.  (In more detail, we have by Proposition~\ref{expansion-prop}  that  for each instance $J$ of schema {\bf P} such that $\cal V$ $\Rightarrow_{J,\emptyset}$ $MV$, we have  $(R^*_{\bar t})^{exp}(J)$ $=$ $\{ ({\bar t}) \}$. The conclusion that $(R^*_{\bar t})^{exp}(I)$ $=$ $\{ ({\bar t}) \}$  for each instance $I$ of schema {\bf P} such that $\cal V$ $\Rightarrow_{I,\Sigma}$ $MV$ follows from the fact that the set of all such instances $I$ is a subset of the set of all such instances $J$.) Thus, by the definitions of expansions of rewriting and of the containment $\sqsubseteq_{\Sigma,MV,{\cal V}}$, we obtain immediately that $R^*_{\bar t}$ $\sqsubseteq_{\Sigma,MV,{\cal V}}$ $Q$.  % {\bf [[[ Need to finish this proof ]]]}
\end{proof} 

It follows from Theorem~\ref{rewr-approach-correct-thm} that the converse of Proposition~\ref{rewr-approach-prop} also holds. Hence we obtain the following result. 

\begin{theorem} 
\label{rewr-approach-thm}  
Given a valid CQ materialized-view setting ${\cal M}\Sigma$ and a %CQ 
query $Q$ of arity $k$ $\geq$ $0$. For a $k$-tuple $\bar t$ of values from $consts({\cal M}\Sigma)$: There exists an $MV$-validated head-instantiated rewriting $R$ for $\bar t$ such that $R$ $\sqsubseteq_{\Sigma,MV,{\cal V}}$ $Q$ iff $\bar t$ is a certain answer to $Q$ w.r.t. ${\cal M}\Sigma$. 
\end{theorem} 

{\em Note 2.} In the light of Theorems~\ref{rewr-approach-correct-thm} and~\ref{rewr-approach-thm}, we can use the results of this current paper on ${\cal M}\Sigma$-conditional query containment to determine correctly if a given ground $k$-tuple $\bar t$ is a certain answer to the ($k$-ary) query $Q$ w.r.t. the setting ${\cal M}\Sigma$, for the class of all problem instances where $Q$ is a CQ query, and the materialized-view settings ${\cal M}\Sigma$ is valid CQ weakly acyclic. Moreover, we can also {\em find} all the certain answers to $Q$ w.r.t. ${\cal M}\Sigma$ for the same class of instances (i.e., CQ queries and valid CQ weakly acyclic materialized-view settings), by first generating all the ground $k$-tuples of values in $consts({\cal M}\Sigma)$, and by then determining for each such tuple whether it is a certain answer to $Q$ w.r.t. ${\cal M}\Sigma$. By the results of this paper, the latter algorithm is sound and complete for this class of input instances under CWA.

\section{Conditional Containment for \\ CQ Queries} 
\label{mendelzon-prelims-sec} 

Zhang and Mendelzon in \cite{ZhangM05} addressed the problem of letting users access  authorized data, via rewriting the users' queries in terms of their authorization views. Toward that goal, \cite{ZhangM05} explored the notion of ``conditional query containment.'' The results of \cite{ZhangM05} include a powerful reduction of the problem of testing conditional containment of CQ queries to that of testing {\em unconditional} containment of modifications of the queries. In this appendix we review these results of \cite{ZhangM05}.  Appendix~\ref{trivial-example-sec} provides an illustrative example of conditional query containment. 

We begin by reviewing the definition of conditional containment of queries \cite{ZhangM05}. Some of the definitions here are restricted versions of the definitions given in Section~\ref{problem-statement-sec}. We provide the restricted definitions here for this appendix to be self contained. 

Suppose that we are given a schema {\bf P} and a set $\cal V$ of relation symbols not in {\bf P}, with each symbol  {\em (view name)}  $V$ $\in$ $\cal V$ of some arity $k_V$ $\geq$ $0$. Each symbol $V$ $\in$ $\cal V$ is defined via a $k_V$-ary query  on the schema {\bf P}. We call $\cal V$ a {\em set of views on} {\bf P}, and call the query for each $V$ $\in$ $\cal V$ the {\em definition of the view} $V$, or {\em the query for} $V$. We assume that  the query for each $V$ $\in$ $\cal V$ is associated with ($V$ in) the set $\cal V$. Consider a ground instance $MV$ of schema $\cal V$; we call $MV$ a {\em set of view answers for} $\cal V$. Then for a ground instance $I$ of schema {\bf P},  we say that $I$ {\em is a valid instance (of {\bf P}) for} $\cal V$ {\em and} $MV$ \cite{ZhangM05} whenever for each $V$ $\in$ $\cal V$, the answer $V(I)$ to the query for $V$ on the instance $I$ is identical to the relation $MV[V]$ for $V$ in the instance $MV$. For a given  set $MV$ of view answers for a set of views $\cal V$, we say that $MV$ {\em is a valid set of view answers for} $\cal V$ whenever there exists at least one valid instance for $\cal V$ and $MV$. 

Now given queries $Q_1$ and $Q_2$ on the schema {\bf P}, we say that $Q_1$ is {\em conditionally contained in $Q_2$ w.r.t. ($\cal V$ and)} $MV$ \cite{ZhangM05}, denoted\footnote{To avoid overcrowding the symbol $\sqsubseteq$, 
%To simplify the notation a bit, 
we assume that in the notation $\sqsubseteq_{MV}$, the name $MV$ of an instance of schema $\cal V$ uniquely identifies  the relevant set $\cal V$.} $Q_1$ $\sqsubseteq_{MV}$ $Q_2$, if and only if the relation $Q_1(I)$ is a subset of the relation $Q_2(I)$ for each valid instance $I$ for $\cal V$ and $MV$. 

It is easy to see that for all instances $MV$ of all schemas $\cal V$, the containment $Q_1$ $\sqsubseteq$ $Q_2$ is a sufficient condition for the containment $Q_1$ $\sqsubseteq_{MV}$ $Q_2$. Not surprisingly, $Q_1$ $\sqsubseteq_{MV}$ $Q_2$ does not imply $Q_1$ $\sqsubseteq$ $Q_2$; % for the same queries
something more sophisticated is clearly called for. The authors of \cite{ZhangM05} report the following powerful test for conditional containment of CQ queries. (We say that $\cal V$ is a set of {\em CQ} views if the query for each $V$ $\in$ $\cal V$ is a CQ query.) % The language of $UCQ^{\neq}$ queries, which are mentioned in Theorem~\ref{mendelz-thm}, is defined in Appendix~\ref{old-inst-query-sec}.) 

\begin{theorem}{\cite{ZhangM05}} 
\label{mendelz-thm} 
Given a schema {\bf P}, a set of CQ views $\cal V$ on {\bf P}, a valid set $MV$ of view answers for $\cal V$, and CQ queries $Q_1$ and $Q_2$ on the schema {\bf P}. Then $Q_1$ $\sqsubseteq_{MV}$ $Q_2$ if and only if for the $UCQ^{\neq}$ query $Q''_1$ constructed for $Q_1$ by an algorithm given in \cite{ZhangM05}, we have $Q''_1$ $\sqsubseteq$ $Q_2$. 
\end{theorem} 

Theorem~\ref{mendelz-thm} reduces the problem of testing conditional containment of CQ queries, $Q_1$ $\sqsubseteq_{MV}$ $Q_2$, to the problem of testing (unconditional) containment in $Q_2$ of a $UCQ^{\neq}$ modification of $Q_1$. The latter containment can be decided by a test due to \cite{LevyMSS95}. % (see Appendix~\ref{old-inst-query-sec} here). 
The required modification of $Q_1$ is done by an intricate algorithm given in \cite{ZhangM05}. %Intuitively, the relationship between the queries $Q_1$ and $Q''_1$ is that they are equivalent on all the ``relevant'' ground instances of schema {\bf P}, that is on all valid instances for $\cal V$ and $MV$. % in the language of unions of CQ queries with disequalities, we denote this language by . 
We outline here briefly the intuition for the construction of $Q''_1$ from $Q_1$. 

We say that an instance $I$ of schema {\bf P}  {\em underproduces} $MV$ if, for at least one $V$ $\in$ $\cal V$, the relation $V(I)$ is a proper subset of the relation $MV[V]$. %Symmetrically, an instance $I$ of schema {\bf P}  {\em overproduces} $MV$ if, for at least one $V$ $\in$ $\cal V$, the relation $V(I)$ is a proper superset of the relation $MV[V]$. 
By definition, %(i) 
each valid instance for $\cal V$ and $MV$ does not underproduce $MV$. 
%, and (ii) for each instance $I$ of schema {\bf P} such that $I$ does not underproduce $MV$, % (note that all valid instances for $\cal V$ and $MV$ are included in this class), we have that for each $V$ $\in$ $\cal V$, the relation $V(I)$ is a superset of the relation $MV[V]$. 

%{\bf Is this not needed? [[[ We say that an instance, $I$, of schema {\bf P}  belongs to the class of ``at-most-$MV$-instances-of-{\bf P}'' if, for each $V$ $\in $\cal V$, we have that $V(I)$ is a subset of the relation $MV[V]$. Symmetrically, an instance, $I$, of schema {\bf P}  belongs to the class of ``at-least-$MV$-instances-of-{\bf P}'' if, for each $V$ $\in $\cal V$, we have that $V(I)$ is a superset of the relation $MV[V]$. Clearly, each valid instance for $\cal V$ and $MV$ is both an at-most-$MV$-instance-of-{\bf P} and an at-least-$MV$-instance-of-{\bf P}. ]]] } 

The construction of $Q''_1$ from $Q_1$ proceeds in two steps. The first step guarantees that its output, a CQ query $Q^*_1$,  has the empty answer on all instances of {\bf P} that underproduce $MV$. This goal is achieved by defining $Q^*_1$ as having the same head vector as in $Q_1$, and by ($Q^*_1$) having the body that is a conjunction of the body of $Q_1$ with the conjunction ${\cal C}^{exp}_{MV}$ defined in Section~\ref{probl-stmt-relationship-sec}. 

The output of the second step in the construction, a $UCQ^{\neq}$ query $Q''_1$, has the same property as $Q^*_1$ does. In addition, for each instance, $I$, of schema  {\bf P} such that $I$ does {\em not} underproduce $MV$, and for each valuation, $\nu$, from the query $Q''_1$ to $I$, all the facts in $\nu(body_{(Q''_1)})$ collectively constitute a valid instance for $\cal V$ and $MV$. (This is done by adding to the body of $Q''_1$ disjunctions, equalities, and/or disequalities based on homomorphisms from the normalized bodies of the views in $\cal V$ to the body of (the current version of) $Q''_1$. The body of a CQ query is {\em normalized} whenever its relational part has only one occurrence of each variable and of each constant, and all the equalities between variables and/or constants are enforced by explicit equality atoms.) 

The result of Theorem~\ref{mendelz-thm} is shown in \cite{ZhangM05} to follow from these properties of $Q''_1$ and from the fact that for all valid instances $I$ for $\cal V$ and $MV$, $Q''_1(I)$ $=$ $Q_1(I)$. 

%We now outline in more detail the two steps in the construction of $Q''_1$. 
%The first step is as follows. Consider an arbitrary fact in $MV$, of the form $V({\bar t})$ for some $V$ $\in$ $\cal V$. We will use the notation $V({\bar X})$ $\leftarrow$ $body(V)$ for the CQ query for $V$. We obtain the {\em expansion} $V^{exp}({\bar t})$ of $V({\bar t})$, by applying to $body(V)$ the mapping that associates each element of the vector $\bar X$ with the same-position element in the vector $\bar t$. (Note that by $MV$ being a valid set of view answers for $\cal V$, this mapping from $\bar X$ to $\bar t$ always exists.) We then rename the variables in the expansions of all the facts in $MV$, in such a way that for each pair of distinct facts in $MV$, the respective expansions have no variable names in common. Denote by $P_{MV}$ the conjunction of all the resulting expansions. Finally, obtain a CQ query $Q^*_1$ by adding $P_{MV}$ to the body of $Q_1$. By construction, the query $Q^*_1$ has an empty answer on all instances of {\bf P} that underproduce $MV$. 
%
%The second step in the construction of $Q''_1$ is rather involved 

\section{Example of Query Containment w.r.t. a Set of View Answers} 
\label{trivial-example-sec} 

In the example in this appendix, one query is ${\cal M}\Sigma$-conditionally contained in the other, even though the bodies of the queries do not share any relational symbols. 

\begin{example} 
\label{trivial-ex} 
In this trivial example, one query is ${\cal M}\Sigma$-conditionally contained in the other (in the absence of dependencies), even though the bodies of the queries do not share any relational symbols. Consider Boolean CQ queries $Q_1$ and $Q_2$, a CQ view $V$, and a set of view answers $MV$, as follows. 

\begin{tabbing} 
$Q_1() \leftarrow P(X).$ \\ 
$Q_2() \leftarrow R(Y).$ \\ 
$V(Y) \leftarrow R(Y).$ \\ 
$MV$ $=$ $\{$ $V( c )$ $\}$.  
\end{tabbing} 

Let ${\cal M}\Sigma$ be $(\{ P, R \}, \emptyset, \{ V \}, MV)$, with $V$ and $MV$ as above. $P$ and $R$ in the schema {\bf P} $=$ $\{ P, R \}$ are unary relation symbols, and no dependencies hold on {\bf P}. % in ${\cal M}\Sigma$. 

For any base instance $I$ that is relevant to the setting ${\cal M}\Sigma$, the instance $I$ must have the ground atom $R( c )$. (This follows from the definitions of the view $V$ and of the instance $MV$.) As a result, the query $Q_2$ returns the empty tuple on any such instance $I$. It follows that {\em any} Boolean query, including $Q_1$, is  ${\cal M}\Sigma$-conditionally contained in %the query 
$Q_2$. The algorithm reported in this paper allows us to make this correct conclusion. 
\end{example}

\section{A Noncontainment Example} 
\label{sigma-noemptyset-ex-sec} 

%{\bf [[[ Stopped here D07/07/13 ]]]}  

In this appendix we show by example that when $Q_1$ $\sqsubseteq_{{\cal M}\Sigma}$ $Q$ holds for some choice of  $Q_1$, $\Sigma$, $MV$, and $Q$, then none of the following necessarily holds: 

\begin{itemize} 
	\item[(1)] $Q_1$ $\sqsubseteq$ $Q$, 
	\item[(2)] $Q_1$ $\sqsubseteq_{\Sigma}$ $Q$, and 
	\item[(3)] $Q_1$ $\sqsubseteq_{{\cal M}\emptyset}$ $Q$. 
\end{itemize} 

\noindent 
Here, by ${\cal M}\emptyset$ we denote the result of replacing $\Sigma$ by $\emptyset$ in the given setting ${\cal M}\Sigma$ $=$ $(${\bf P}, $\Sigma$, $\cal V$, $MV \}$. %Thus, from this claim (3) above (i.e., from the claim that $Q_1$ $\sqsubseteq_{{\cal M}\emptyset}$ $Q_2$ does not necessarily imply $Q_1$ $\sqsubseteq_{{\cal M}\Sigma}$ $Q_2$), it follows that our results in this current paper are not just immediate corollaries of the results of \cite{ZhangM05}. 

\begin{example} 
\label{main-sigma-ex} 
Recall the schema {\bf P} $=$ $\{$ {\tt Emp,}\linebreak {\tt HQDept},  {\tt OfficeInHQ} $\}$ of Example~\ref{intro-main-three-ex}. We abbreviate each relation name by its first letter (same as in Example~\ref{intro-formalized-ex}). As before, we assume that the only primary key of the relation $E$ is its three attributes together. The key of the relation $O$ is its first attribute, which we express using the following egd $\tau$: 

\begin{tabbing} 
$\tau:$ $O(X,Y) \wedge O(X,Z) \rightarrow Y = Z .$
\end{tabbing} 

%\noindent 
Suppose that for  all the departments located in the company headquarters, all their employees %of the departments 
have their offices in the headquarters. We express this constraint using a tgd $\sigma$ (which is the same as in Examples~\ref{intro-main-three-ex} and~\ref{intro-formalized-ex}): 

\begin{tabbing} 
$\sigma:$  $E(X,Y,Z) \wedge H(Y) \rightarrow \exists S \ O(X,S).$
\end{tabbing} 

We assume that $\sigma$ and $\tau$ constitute all the integrity constraints $\Sigma$ on the schema {\bf P}, that is, $\Sigma$ $=$ $\{$ $\sigma$, $\tau$ $\}$. 

Recall the views {\tt U}, {\tt V}, and {\tt W} introduced in Example~\ref{intro-main-three-ex}: 

{\small 
\begin{verbatim} 
(U): DEFINE VIEW U(Dept) AS SELECT * FROM HQDept;    
(V): DEFINE VIEW V(Name, Dept) AS 
     SELECT DISTINCT Name, Dept FROM Emp;   
(W): DEFINE VIEW W(Dept, Salary) AS
     SELECT DISTINCT Dept, Salary FROM Emp;   
\end{verbatim} 
} % end \small 

\noindent 
We denote by $\cal V$ the set $\{$ {\tt U}, {\tt V}, {\tt W} $\}$. 

Suppose that a user, or several users together, are authorized to see the answers to all three views {\tt U}, {\tt V}, and {\tt W}, and that at some point in time the user(s) can see the following set $MV$ of answers to these views (same as in Example~\ref{intro-main-three-ex}). %materialized views $MV$ (i.e., of the answers to the queries {\tt V} and {\tt W} on a fixed instance of the {\tt Emp} relation) that the user can see  has one tuple in each view, as follows.  

\begin{tabbing} 
hhabgnms \= jehf \kill 
%$MV$ $=$ $\{$ {\tt U(Sales)}, {\tt V(JohnDoe,Sales)}, \\ 
%\> {\tt W(Sales,\$50000)}  $\} .$ 
$MV$$=$$\{${\tt U(sales)},{\tt V(johnDoe,sales)},{\tt W(sales,50000)}$\}.$ 
\end{tabbing} 

\noindent 
We denote by ${\cal M}\Sigma$ the materialized-view setting $(${\bf P}, $\Sigma$, $\cal V$, $MV \}$.  

Now recall the secret query {\tt Q} of Example~\ref{intro-main-three-ex}; {\tt Q} returns the names and salaries of all the employees who work in the company headquarters. 

{\small 
\begin{verbatim} 
(Q): SELECT DISTINCT E.Name, Salary FROM Emp E, OfficeInHQ  
     WHERE E.Name = OfficeInHQ.Name, 
\end{verbatim} 
} % end \small 

Recall the tuple $\bar t$ $=$ $($  {\tt johnDoe}, {\tt 50000} $)$ and the query $R_{vw}$, over the schema $\cal V$, of Example~\ref{app-main-three-ex}: 

{\small 
\begin{verbatim} 
(Rvw): SELECT DISTINCT Name, Salary FROM V, W 
       WHERE V.Name = `johnDoe' AND V.Dept = W.Dept 
       AND V.Dept = `sales' AND Salary = `50000';  
\end{verbatim} 
} % end \small 

Using the results of this paper, we can show that the expansion $R^{exp}_{vw}$, in terms of the schema {\bf P}, of the query $R_{vw}$ is contained in the query {\tt Q} w.r.t. the setting ${\cal M}\Sigma$. At the same time, none of the following containments hold: $R^{exp}_{vw}$ $\sqsubseteq$ $Q$, $R^{exp}_{vw}$ $\sqsubseteq_{\Sigma}$ $Q$, and $R^{exp}_{vw}$ $\sqsubseteq_{{\cal M}\emptyset}$ $Q$. (Here, by ${\cal M}\emptyset$ we denote the result of replacing $\Sigma$ by $\emptyset$ in the setting ${\cal M}\Sigma$.) 

We now prove all the containment and non-containment statements of the preceding paragraph, for the queries $R^{exp}_{vw}$ and {\tt Q} over the schema {\bf P}. First, we render in Datalog the queries {\tt Rvw}, $R^{exp}_{vw}$, {\tt Q}, and the queries for the three views. (For conciseness, in the remainder of this example we will refer to the constants {\tt johnDoe}, {\tt sales}, and {\tt 50000} as $c$, $d$, and $f$, respectively.)  

\begin{tabbing} 
hitab muc \= mekill \kill 
$Q(X,Z)$ \> $\leftarrow E(X,Y,Z), O(X,S) .$ \\  
$U(X)$ \> $\leftarrow H(X) .$ \\  
$V(X,Y)$ \> $\leftarrow E(X,Y,Z) .$ \\  
$W(Y,Z)$ \> $\leftarrow E(X,Y,Z) .$ \\ 
$R_{vw}(c,f)$ \> $\leftarrow V(c,d), W(d,f) .$ \\  
$R^{exp}_{vw}(c,f)$ \> $\leftarrow E(c,d,Z), E(X,d,f) .$ 
\end{tabbing} 
 
(1) We are first determining whether $R^{exp}_{vw}$ $\sqsubseteq$ $Q$ holds. The noncontainment of  $R^{exp}_{vw}(c,f)$ in $Q$ is immediate from the containment test of \cite{ChandraM77} and from the absence in the definition of $R^{exp}_{vw}(c,f)$ of a subgoal with predicate {\tt OfficeInHQ}. (As a result, the subgoal of $Q$ with predicate {\tt OfficeInHQ} cannot be mapped into the body of $R^{exp}_{vw}(c,f)$.) We conclude that $R^{exp}_{vw}$ $\sqsubseteq$ $Q$ does not hold. 

(2) We are now determining whether $R^{exp}_{vw}$ $\sqsubseteq_{\Sigma}$ $Q$ holds. We observe that the result of chasing the query $R^{exp}_{vw}(c,f)$ with the dependencies $\Sigma$ is identical to $R^{exp}_{vw}(c,f)$. Recall that $R^{exp}_{vw}(c,f)$ $\sqsubseteq_{\Sigma}$ $Q$ holds iff that chase result (which is identical to $R^{exp}_{vw}(c,f)$) is contained in $Q$ in the absence of dependencies. We then use the reasoning of item (1) to conclude that $R^{exp}_{vw}$ $\sqsubseteq_{\Sigma}$ $Q$ does not hold. 

(3) We are now determining whether $R^{exp}_{vw}$ $\sqsubseteq_{{\cal M}\emptyset}$ $Q$ holds. Consider the following instance $I$ of schema {\bf P}: 

\begin{tabbing} 
$I = \{ H(d), E(c,d,f) \}.$ 
\end{tabbing} 

\noindent 
It is easy to verify that for the set $\cal V$ $=$ $\{$ {\tt U}, {\tt V}, {\tt W} $\}$, the result of applying the queries for $\cal V$ to $I$ is exactly the instance $MV$ as given above. (Observe that the instance $I$ does not satisfy the tgd $\sigma$ in the set $\Sigma$ in the setting ${\cal M}\Sigma$. At the same time, we're checking here for the containment of $R^{exp}_{vw}$ in $Q$ w.r.t. the setting ${\cal M}\emptyset$, in which $\Sigma$ $=$ $\emptyset$.) We verify that $Q(I)$ $=$ $\emptyset$ and that $R^{exp}_{vw}(c,f)(I)$ $=$ $\{$ $(c,f)$ $\}$. As a result, the containment $R^{exp}_{vw}$ $\sqsubseteq_{{\cal M}\emptyset}$ $Q$ does not hold. %We conclude that $R_{vw}$ $\sqsubseteq_{{\cal V},MV}$ $Q$ does not hold either. 

(4) Finally, let us determine whether $R^{exp}_{vw}$ $\sqsubseteq_{{\cal M}\Sigma}$ $Q$ holds. For ease of exposition, we denote the query $R^{exp}_{vw}$ by $Q_1$. Using the results of this current paper, we first transform $Q_1$ into $Q'_1$, by conjoining the body of $Q_1$ with ${\cal C}^{exp}_{MV}$ $=$ $E(c,d,A) \wedge E(B,d,f) \wedge H(d)$. (We then minimize the resulting query to obtain $Q'_1$; the minimization does not affect any of our results in this current paper.)  

\begin{tabbing} 
hitab muc \= mekill \kill 
$Q_1(c,f)$ \> $\leftarrow E(c,d,Z), E(X,d,f).$ \\ 
$Q'_1(c,f)$ \> $\leftarrow E(c,d,Z), E(X,d,f), H(d).$ 
\end{tabbing} 

Then we chase $Q'_1$ using the $MV$-induced dependencies $\tau_U$, $\tau_V$, and $\tau_W$, as well as the modifications $\sigma'$ and $\tau'$ of the dependencies $\sigma$ and $\tau$, respectively, in the given set $\Sigma$ of dependencies on the schema {\bf P}: 

\begin{tabbing} 
hitab \= muc hhhhhhhhhhhg \= mekill \kill 
$\tau_U:$ \> $H(X)$ \> $\rightarrow  X = d.$ \\  
$\tau_V:$ \> $E(X,Y,Z)$ \> $\rightarrow  X = c \wedge Y = d.$ \\  
$\tau_W:$ \> $E(X,Y,Z)$ \> $\rightarrow Y = d \wedge Z = f.$ \\ 
$\sigma':$ \> $E(X,Y,Z) \wedge H(T) \rightarrow (\exists S \ O(X,S)) \vee (Y \neq T).$ \\ 
$\tau':$ \> $O(X,Y) \wedge O(T,Z) \rightarrow (Y = Z) \vee (X \neq T).$
\end{tabbing} 

The result of the chase of $Q'_1$ with the dependencies $\tau_U$, $\tau_V$, $\tau_W$, $\sigma'$, and $\tau'$ is the following query $(Q_1)^{{\cal M}\Sigma}$: 

\begin{tabbing} 
hitab muc \= mekill \kill 
$(Q_1)^{{\cal M}\Sigma}(c,f) \leftarrow E(c,d,f), H(d), O(c,Z).$ 
\end{tabbing} 

It is easy to verify that, by the results of \cite{ChandraM77}, the CQ query $(Q_1)^{{\cal M}\Sigma}$ is unconditionally contained in the input query $Q$. Using the results of this current paper, we obtain that  $Q_1$ $\sqsubseteq_{{\cal M}\Sigma}$ $Q$ holds as well. On replacing $Q_1$ by the original notation $R^{exp}_{vw}$, we conclude that $R^{exp}_{vw}$ $\sqsubseteq_{{\cal M}\Sigma}$ $Q$ also holds. 
\end{example} 

\section{For ${\cal M}\Sigma$-Conditional Containm-\\ ent, Cannot Just Chase with $\Sigma$ \\ (or even with $\Sigma_{(\neq)}$) the Query $Q_1''$} 
\label{cannot-chase-mendelzon-sec} 
 
This appendix illustrates via an example that, in determining ${\cal M}\Sigma$-conditional containment of two CQ queries w.r.t. a weakly acyclic materialized-view setting, just ``chasing with $\Sigma$ (or even with $\Sigma_{(\neq)}$)'' the query $Q''_1$ of \cite{ZhangM05} may yield incorrect conclusions about the ${\cal M}\Sigma$-conditional containment of the input queries.  (``The query $Q''_1$ of \cite{ZhangM05}'' is the result of chasing one of the input queries in the algorithm of \cite{ZhangM05} for determining conditional containment of the two input queries in the absence of dependencies. The algorithm of \cite{ZhangM05} checks this query $Q''_1$  for unconditional containment in the input query $Q_2$; it is shown in \cite{ZhangM05} that, in case of the positive answer to the unconditional-containment test, the input query $Q_1$ is {\em conditionally} contained in $Q_2$.) 

\nop{ % bufgu 
\begin{example} 
\label{tgd-nonlayered-ex} 
Let schema {\bf P} have a unary relation symbol $R$ and binary relation symbols $P$ and $S$. Consider a tgd  $\sigma$ and three view definitions: 

\begin{tabbing} 
$\sigma: S(X,Y) \rightarrow \exists \ T \ \ P(Y,T).$ \\ 
$U(X) \leftarrow R(X).$ \\ 
$V(X) \leftarrow P(Y,X).$ \\ 
$W(X) \leftarrow S(Y,X).$ 
\end{tabbing} 

Finally, for the set  $\{ U,V,W \}$ of the above views, let the set of view answers $MV$ be $\{ U( c ), V( c ), W(f) \}$. Then the materialized-view setting $(${\bf P}, $\{ \sigma \}$, $\{ U,V,W \}$, $MV)$, which we denote by ${\cal M}\Sigma$, is CQ weakly acyclic.

Now let $Q_1$ and $Q_2$ be two CQ queries, as follows. 

\begin{tabbing} 
$Q_1(X) \leftarrow S(X,Y).$ \\ 
$Q_2(X) \leftarrow S(X,Y), P(Y,Z), R(Z).$ 
\end{tabbing} 

It is easy to show that neither $Q_1 \sqsubseteq Q_2$ nor $Q_1 \sqsubseteq_{\{ \sigma \}} Q_2$ holds. (Please see Section~\ref{prelim-sec}.) However, we can show that  $Q_1 \sqsubseteq_{{\cal M}\Sigma} Q_2$ does hold. 
The result of the chase of $Q_1$ in our approach is the following CQ query:  

\begin{tabbing} 
$(Q_1)^{{\cal M}\Sigma}(X) \leftarrow S(X,f), P(f,c), R( c ), P(Z,c), S(T,f).$ 
\end{tabbing} 

We can then determine that $Q_1$ $\sqsubseteq_{{\cal M}\Sigma}$ $Q_2$ holds, by using the results of \cite{ChandraM77} to check that the unconditional containment $(Q_1)^{{\cal M}\Sigma} \sqsubseteq Q_2$ holds.   
\end{example} 

} % end \nop bufgu

{\em Note.} As shown in this current paper, the result $(Q_1)^{{\cal M}\Sigma}$ of chasing a CQ query $Q_1$ using weakly acyclic dependencies $\Sigma$ and the $MV$-induced dependencies 
is, in general,  a $UCQ^{\neq}$ query. $(Q_1)^{{\cal M}\Sigma}$ in Example~\ref{extended-tgd-nonlayered-ex} is a CQ query, because the view/tgd definitions and $MV$ are particularly simple here.

\begin{example} 
\label{extended-tgd-nonlayered-ex} 
%This is a detailed version of Example~\ref{tgd-nonlayered-ex}. 
Let schema {\bf P} have a unary relation symbol $R$ and binary relation symbols $P$ and $S$. Consider a tgd  $\sigma$ and three view definitions: 

\begin{tabbing} 
$\sigma: S(X,Y) \rightarrow \exists \ T \ \ P(Y,T).$ \\ 
$U(X) \leftarrow R(X).$ \\ 
$V(X) \leftarrow P(Y,X).$ \\ 
$W(X) \leftarrow S(Y,X).$ 
\end{tabbing} 

Finally, for the set  $\{ U,V,W \}$ of the above views, let the set of view answers $MV$ be $\{ U( c ), V( c ), W(f) \}$. Then the materialized-view setting $(${\bf P}, $\{ \sigma \}$, $\{ U,V,W \}$, $MV)$, which we denote by ${\cal M}\Sigma$, is CQ weakly acyclic.

Now let $Q_1$ and $Q_2$ be two CQ queries, as follows. 

\begin{tabbing} 
$Q_1(X) \leftarrow S(X,Y).$ \\ 
$Q_2(X) \leftarrow S(X,Y), P(Y,Z), R(Z).$ 
\end{tabbing} 

It is easy to show that neither $Q_1 \sqsubseteq Q_2$ nor $Q_1 \sqsubseteq_{\{ \sigma \}} Q_2$ holds. (Please see Section~\ref{prelim-sec}.) However, it turns out that  $Q_1 \sqsubseteq_{{\cal M}\Sigma} Q_2$ does hold. A way to prove this fact is to chase the query $Q_1$ using both the given dependencies, $\{ \sigma \}$, on the schema {\bf P}, and the ``$MV$-induced'' dependencies that we introduce in this paper. (In this particular example, the input dependency $\sigma$ is the same as its ``neq-transformation.'' That is, the set of dependencies $\Sigma_{(\neq)}$ required for the chase in our approach for checking ${\cal M}\Sigma$-conditional containment, is the same as the input set $\Sigma$ $=$ $\{ \sigma \}$. Thus, chase with $\sigma$ and with the $MV$-induced dependencies as shown in this example is correct w.r.t. the approach for checking ${\cal M}\Sigma$-conditional containment as introduced in this paper.) 

The result of the chase is the following CQ query:  

\begin{tabbing} 
$(Q_1)^{{\cal M}\Sigma}(X) \leftarrow S(X,f), P(f,c), R( c ), P(Z,c), S(T,f).$ 
\end{tabbing} 

We can then determine that the ${\cal M}\Sigma$-conditional containment of $Q_1$ in $Q_2$ holds, by using the results of \cite{ChandraM77} to check that the unconditional containment $(Q_1)^{{\cal M}\Sigma} \sqsubseteq Q_2$ holds.  

We now provide the details of chasing the query   $Q_1$ using both the given dependencies, $\{ \sigma \}$, on the schema {\bf P}, and the $MV$-induced dependencies. The process of obtaining the query $(Q_1)^{{\cal M}\Sigma}$ from the query $Q_1$ via this process can be represented using four stages, as follows: 

{\em Stage I:} We first add to the body of the query $Q_1$ the conjunction ${\cal C}^{exp}_{MV}$ $=$ $R( c ) \wedge P(Z,c) \wedge S(T,f)$: 

\begin{tabbing} 
$Q_1'(X) \leftarrow S(X,Y), R( c ), P(Z,c), S(T,f).$ 
\end{tabbing}

{\em Stage II:} In this stage, we choose to chase the query $Q'_1$ using the following $MV$-induced dependencies $\tau_U$, $\tau_V$, and $\tau_W$. (Intuitively, each of the egds $\tau_U$, $\tau_V$, and $\tau_W$ below arises from the ground fact for the respective view in the given instance $MV$. If $MV$ had more than one fact for any view symbol, call it $V^*$, then all these facts together would give rise to a single dependency for $V^*$, with a disjunction on the right-hand side of the dependency. The reason each $MV$-induced dependency in this example is ``just'' an egd is that both the definitions of the views and the instance $MV$ are particularly simple here.) 

\begin{tabbing} 
hehe boo heh \= doo \kill 
$\tau_U: R(X)$ \> $\rightarrow X = c.$ \\ 
$\tau_V: P(Y,X)$ \> $\rightarrow X = c.$ \\ 
$\tau_W: S(Y,X)$ \> $\rightarrow X = f.$ 
\end{tabbing} 

Applying these three dependencies to the query $Q'_1$ results in the following query: 

\begin{tabbing} 
$Q''_1(X) \leftarrow S(X,f), R( c ), P(Z,c), S(T,f).$ 
\end{tabbing} 

\noindent 
The difference between the queries $Q'_1$ and $Q''_1$ is that an application of the egd $\tau_W$ to the first subgoal, $S(X,Y)$, of the query $Q'_1$ turns this subgoal into the first subgoal $S(X,f)$ of the query $Q''_1$. None of the dependencies $\tau_U$, $\tau_V$, and $\tau_W$ is applicable to the query $Q''_1$. 

{\em Stage III:} In this stage, we choose to chase the query $Q''_1$ with the dependency $\sigma$ (on the schema {\bf P} in ${\cal M}\Sigma$); the outcome of the chase is a CQ query $Q'''_1$, see below. The two chase steps with $\sigma$ on $Q''_1$ result in the addition to the body of $Q''_1$ of two subgoals, $P(f,A)$ and $P(f,B)$. These two subgoals are what is different between the queries $Q''_1$ and $Q'''_1$. The tgd $\sigma$ is not applicable to the query $Q'''_1$. 

\begin{tabbing} 
$Q'''_1(X) \leftarrow S(X,f), R( c ), P(Z,c), S(T,f), P(f,A), P(f,B).$ 
\end{tabbing} 

Note that if we stop here, unconditional containment of the query $Q'''_1$ in $Q_2$ does not hold. (The reason is, the body of $Q'''_1$  does not have any pattern of three subgoals with predicates $S$, $P$, and $R$, such that the body of the query $Q_2$ would subsume the pattern.)  

{\em Stage IV:} Now, after stage III  in chasing the original query $Q_1$ is over, we can apply again chase steps with the $MV$-induced dependencies. Indeed, we can do two chase steps with the egd $\tau_V$. As a result, the subgoals $P(f,A)$ and $P(f,B)$ of the query $Q'''_1$  are transformed into two identical atoms $P(f,c)$. We call the resulting query 

\begin{tabbing} 
$(Q_1)^{{\cal M}\Sigma}(c) \leftarrow S(X,f), P(f,c), R( c ), P(Z,c), S(T,f).$ 
\end{tabbing} 

Chase with $\sigma$ or with the $MV$-induced dependencies does not apply to this query $(Q_1)^{{\cal M}\Sigma}$. It is easy to see that the query $(Q_1)^{{\cal M}\Sigma}$ is unconditionally contained in the query $Q_2$. We conclude that the given query $Q_1$ is ${\cal M}\Sigma$-conditionally contained in %the query 
$Q_2$. 
\end{example} 

\section{Ensuring that the Image of \\ $Body((Q_1)^{{\cal M}\Sigma})$ under each Valuation Be a $\Sigma$-Valid Base Instance for $\cal V$ and $MV$}  
\label{deps-must-have-neqs-sec} 

In this appendix, in the context of the problem of determining ${\cal M}\Sigma$-conditional query containment, we demonstrate via two examples the role of disequality atoms, both in $MV$-induced dependencies (Example~\ref{mv-deps-must-have-neqs-ex}) and in the dependencies $\Sigma_{(\neq)}$ (Example~\ref{sigma-deps-must-have-neqs-ex}). As a summary of the observations  illustrated by these examples, when at least some of the disequalities are missing in either kind of dependencies, then there may exist a valuation, call it $\nu$, from the body, $B$, of the resulting query to some instance, such that the image of (the relational part of) $B$ under $\nu$ is  not a $\Sigma$-valid base instance for $\cal V$ and $MV$. As a result, we would not be able to prove correctness of the algorithm for determining ${\cal M}\Sigma$-conditional containment of CQ queries w.r.t. CQ weakly acyclic materialized-view settings, both in case $\Sigma$ $=$ $\emptyset$ (this is the setting of \cite{ZhangM05}) and in case $\Sigma$ $\neq$ $\emptyset$. (Example~\ref{mv-deps-must-have-neqs-ex} illustrates the former case, and Example~\ref{sigma-deps-must-have-neqs-ex} -- the latter case.) Specifically, we would not be able to prove the claim that $Q_1 \sqsubseteq_{{\cal M}\Sigma} Q_2$ implies $(Q_1)^{{\cal M}\Sigma}$ $\sqsubseteq$ $Q_2$. 

\begin{example} 
\label{mv-deps-must-have-neqs-ex} 
On a schema {\bf P} with one binary relation symbol $P$, consider two CQ queries, $Q_1$ and $V$, as follows: 

\begin{tabbing} 
$Q_1(X) \leftarrow P(X,Y).$ \\ 
$V(X) \leftarrow P(X,X).$ 
\end{tabbing} 

Let $\cal V$ be the set $\{ V \}$, with the set of view answers $MV$ $=$ $\{$ $V(c)$ $\}$.  We will show chase of the query $Q_1$ with the dependencies arising in the setting ${\cal M}\Sigma$ $=$ ({\bf P}, $\emptyset$, $\cal V$, $MV$). 

{\bf Scenario A.} In this scenario, we show the correct chase of the query $Q_1$ for the setting ${\cal M}\Sigma$, as introduced in \cite{ZhangM05}. The first step of the approach is to conjoin the body of $Q_1$ with the ${\cal C}^{exp}_{MV}$ $=$ $P(c,c)$: 

\begin{tabbing} 
$Q'_1(X) \leftarrow P(X,Y), P(c,c).$  
\end{tabbing} 

We then chase the query $Q'_1$ with the $MV$-induced dependency $\tau_V$: 

\begin{tabbing} 
$\tau_V: P(X,Y)  \rightarrow X = c \vee X \neq Y.$  
\end{tabbing} 

The result of the chase of $Q'_1$ with $\tau_V$ is a $UCQ^{\neq}$ query  $(Q_1)^{{\cal M}\Sigma}$ $=$ $\{ Q_1^{(a)}, Q_1^{(b)} \}$, with the $CQ^{\neq}$ components specified as follows: 

\begin{tabbing} 
$Q_1^{(a)}(c ) \leftarrow P(c,Y), P(c,c).$  \\ 
$Q_1^{(b)}(X) \leftarrow P(X,Y), X \neq Y, P(c,c).$  
\end{tabbing} 

We can show that for each instance $I$ of the schema {\bf P} and for each valuation, $\nu$, for either one of the two $CQ^{\neq}$ components of the query $(Q_1)^{{\cal M}\Sigma}$ and for $I$, the image under $\nu$ of (the relational part of) the body of the relevant $CQ^{\neq}$ component of $(Q_1)^{{\cal M}\Sigma}$ is a $\Sigma$-valid base instance for $\cal V$ and $MV$. (Here, $\Sigma$ $=$ $\emptyset$ as specified above, and $\cal V$ and $MV$ are also as above.) 

{\bf Scenario B.} In this scenario, we show chase of the query $Q_1$ using a version of the dependency $\tau_V$ (which is defined as in Scenario A here)  that does {\em not} use disequalities. We refer to this version of $\tau_V$ as $\tilde{\tau}_V$: 

\begin{tabbing} 
$\tilde{\tau}_V: P(X,X)  \rightarrow X = c.$  
\end{tabbing} 

\noindent 
(Recall that the body of the view $V$ is $P(X,X)$.) 

Similarly to Scenario A, we first obtain the query $Q'_1$, and then chase it with $\tilde{\tau}_V$. The result of chase of $Q'_1$ with $\tilde{\tau}_V$ is a CQ query  $\tilde{Q}_1$ specified as follows: 

\begin{tabbing} 
$\tilde{Q}_1(X) \leftarrow P(X,Y), P(c,c).$  
\end{tabbing} 

\noindent 
As the dependency $\tilde{\tau}_V$ does not apply to (either $Q'_1$ or) $\tilde{Q}_1$, the query $\tilde{Q}_1$ is the result of the chase of the query $Q'_1$ with  $\tilde{\tau}_V$.

Consider an instance  $I$ $=$ $\{ P(c,c), P(d,d) \}$ of the schema {\bf P}, and the valuation $\nu: \{ X \rightarrow d, Y \rightarrow d \}$ for $\tilde{Q}_1$ and $I$. Clearly, the image $J$ $=$ $\{ P(d,d), P(c,c) \}$ of (the relational part of) the body of the query $\tilde{Q}_1$ under $\nu$ is an instance (of schema {\bf P}) that does not generate the above set of view answers $MV$ under CWA. 
That is, $J$ is not a $\emptyset$-valid base instance for the $\{ V \}$ and $MV$  as above. The reason is, the relation $V(J)$ has the tuple $V(d)$, which is not in the instance $MV$ as specified in the beginning of this example. 
\end{example}

\begin{example} 
\label{sigma-deps-must-have-neqs-ex} 
On a schema {\bf P} with a binary relation symbol $P$ and a unary relation symbol $S$, consider two CQ queries, $Q_1$ and $V$, and a tgd $\sigma$, all defined as follows: 

\begin{tabbing} 
$Q_1(X) \leftarrow P(X,Y).$ \\ 
$V(X) \leftarrow P(X,Y).$ \\ 
$\sigma: P(X,X) \rightarrow S(X)$ 
\end{tabbing} 

Let $\cal V$ be the set $\{ V \}$, with the set of view answers $MV$ $=$ $\{$ $V(c)$ $\}$.  We will show chase of the query $Q_1$ with the dependencies arising in the setting ${\cal M}\Sigma$ $=$ ({\bf P}, $\{ \sigma \}$, $\cal V$, $MV$). 

{\bf Scenario A.} In this scenario, we show the correct chase of the query $Q_1$ for the setting ${\cal M}\Sigma$, as introduced in this current paper. The first step of the approach is to conjoin the body of $Q_1$ with the ${\cal C}^{exp}_{MV}$ $=$ $P(c,Z)$: 

\begin{tabbing} 
$Q'_1(X) \leftarrow P(X,Y), P(c,Z).$  
\end{tabbing} 

We then chase the query $Q'_1$ with the $MV$-induced dependency $\tau_V$ and with the ``neq-transformation'' $\sigma'$ of the tgd $\sigma$, defined as follows: 

\begin{tabbing} 
$\tau_V: P(X,Y)  \rightarrow X = c$  \\ 
$\sigma': P(X,Y) \rightarrow S(X) \vee X \neq Y$ 
\end{tabbing} 

The result of the chase of $Q'_1$ with $\tau_V$ and $\sigma'$ is a $UCQ^{\neq}$ query  $(Q_1)^{{\cal M}\Sigma}$ $=$ $\{ Q_1^{(a)}, Q_1^{(b)}, Q_1^{(c )}, Q_1^{(d)}  \}$, with the $CQ^{\neq}$ components specified as follows: 

\begin{tabbing} 
$Q_1^{(a)}(c ) \leftarrow P(c,Y), S(c ), P(c,Z).$  \\ 
$Q_1^{(b)}(c ) \leftarrow P(c,Y), S(c ), P(c,Z), Z \neq c.$  \\ 
$Q_1^{(c )}(c ) \leftarrow P(c,Y), Y \neq c,  P(c,Z), S(c ).$  \\ 
$Q_1^{(d)}(c ) \leftarrow P(c,Y), Y \neq c, P(c,Z), Z \neq c.$  
\end{tabbing} 

We can show that for each instance $I$ of the schema {\bf P} and for each valuation, $\nu$, for any one of the four $CQ^{\neq}$ components of the query $(Q_1)^{{\cal M}\Sigma}$ and for $I$, the image under $\nu$ of (the relational part of) the body of the relevant $CQ^{\neq}$ component of $(Q_1)^{{\cal M}\Sigma}$ is a $\Sigma$-valid base instance for $\cal V$ and $MV$. (Here, $\Sigma$ $=$ $\{ \sigma \}$ as specified above, and $\cal V$ and $MV$ are also as above.) 

{\bf Scenario B.} In this scenario, we show chase of the query $Q_1$ using the dependency $\tau_V$ (defined as in Scenario A here), as well as the original tgd $\sigma$ (which does {\em not} use disequalities), instead of using the dependency $\sigma'$ of Scenario A. 

Similarly to Scenario A, we first obtain the query $Q'_1$, and then chase it with $\tau_V$ and $\sigma$. The result of the chase is a CQ query  $\tilde{Q}_1$ specified as follows: 

\begin{tabbing} 
$\tilde{Q}_1(c ) \leftarrow P(c,Y), P(c,Z).$  
\end{tabbing} 

\noindent 
Neither $\tau_V$ nor $\sigma$ applies to $\tilde{Q}_1$. 

Consider an instance  $I$ $=$ $\{ P(c,c), S(c ) \}$ of the schema {\bf P}, and the valuation $\nu: \{ Y \rightarrow c, Z \rightarrow c \}$ for $\tilde{Q}_1$ and $I$. Clearly, the image $J$ $=$ $\{ P(c,c) \}$ of (the relational part of) the body of the query $\tilde{Q}_1$ under $\nu$ is an instance (of schema {\bf P}) that does not satisfy the input tgd $\sigma$ (even though the instance $I$ does). (The instance $J$ does generate the above set of view answers $MV$ under CWA.) We conclude that  
$J$ is not a $\Sigma$-valid base instance for the $\{ V \}$ and $MV$  as above. 
\end{example}

\section{The Data-Exchange Approach} 
\label{dexchg-sec} 

In this appendix we outline an approach to finding the set of certain answers to a CQ query w.r.t. a CQ weakly acyclic materialized-view setting under CWA. (Please see Appendix~\ref{app-dexchg-sec} for all the technical details.) This approach is based on data exchange \cite{FaginKMP05,Barcelo09,LibkinDataExchange}, hence the name. 

This approach is the result of our having rediscovered independently the idea and methods of the 2005 paper \cite{StoffelS05} by Stoffel and colleagues. The work \cite{StoffelS05} explicitly uses techniques that arise in data exchange, to solve the problem of finding the set of certain answers to a query w.r.t. a materialized-view setting under the {\em open-world assumption (OWA).}  At the same time, Brodsky and colleagues in their paper \cite{BrodskyFJ00}, which was published in 2000, used the same approach as Stoffel and colleagues did in \cite{StoffelS05}, without calling their approach (of  \cite{BrodskyFJ00}) ``data exchange.'' (Arguably,  %, when the work \cite{BrodskyFJ00} was published, 
``data exchange'' was not a household term in the year 2000.) Both \cite{BrodskyFJ00} and \cite{StoffelS05} solve the problem of finding the set of certain answers to a query w.r.t. a materialized-view setting under the open-world assumption (OWA). (Please see Section~\ref{related-work-section} for the details on the query languages and classes of dependencies to which the work of \cite{BrodskyFJ00} and \cite{StoffelS05} applies.) 

In this Appendix~\ref{dexchg-sec} we show that, not surprisingly, the approach of \cite{BrodskyFJ00} and \cite{StoffelS05} is sound but not complete under the closed-world assumption (CWA), even in case when the given (base) schema comprises a single relation, and even in the absence of dependencies on this schema. (A counterexample can be found in Appendix~\ref{main-three-sec-app}.) Our ``view-verified data exchange'' of Appendix~\ref{vv-dexchg-sec} then provides a correct algorithm for solving the problem of finding the set of certain answers to a query w.r.t. a materialized-view setting under CWA, for CQ queries and CQ weakly acyclic materialized-view settings. 

The idea of using data exchange \cite{FaginKMP05,Barcelo09,LibkinDataExchange} as a tool arises naturally in the context of the problem of finding the set of certain query answers w.r.t. a materialized-view setting. In the remainder of this appendix, we outline the resulting ``data-exchange'' approach. 
%We outline this intuition in this current appendix, and provide all the details in Appendix~\ref{app-dexchg-sec}. 
%
%The main, perhaps surprising, result of this appendix is that the high-tech data-exchange approach is not as powerful as the rewriting approach of Section~\ref{rewriting-sec}, even for those CQ instances ${\cal M}\Sigma$ where the set $\Sigma$ of dependencies is the empty set. (For instance, the sound rewriting approach returns a nonempty set of disclosed information leaks in the setting of Example~\ref{main-three-ex}, while the data-exchange approach returns the empty set for the same input.) As it turns out, the data-exchange approach can be extended to yield a sound and complete algorithm for the problem of information-leak disclosure for all CQ weakly-acyclic inputs, including those for which Section~\ref{rewriting-sec} provides no solutions. We introduce the sound and complete extension in Section~\ref{vv-dexchg-sec}. 
%
We begin by reviewing  the basics of data exchange in Section~\ref{dexchg-prelim-sec}, by generally following the excellent detailed survey \cite{Barcelo09}. Then, in Section~\ref{dexchg-work-sec} we introduce and discuss the sound but not complete data-exchange approach to finding the set of certain answers to a CQ query w.r.t. a CQ weakly acyclic materialized-view setting under CWA. All the technical details of the discussion can be found in Appendix~\ref{app-dexchg-sec}. 

%\vspace{-0.1cm} 

\subsection{Reviewing Data Exchange} 
\label{dexchg-prelim-sec} 

Given schemas {\bf S} $=$ $<S_1$, $\ldots$, $S_m>$ and {\bf T} $=$ $<T_1$, $\ldots$, $T_n>$, with no relation symbols in common, denote by $<${\bf S}, {\bf T}$>$ the schema $<S_1$, $\ldots$, $S_m$, $T_1$, $\ldots$, $T_n>$. If $I$ is an instance of {\bf S} and $J$ an instance of {\bf T}, then $(I$, $J)$ denotes an instance $K$ of $<$ {\bf S}, {\bf T} $>$ such that $K[S_i]$ $=$ $I[S_i]$ and $K[T_j]$ $=$ $J[T_j]$, for $i$ $\in$ $[1$, $m]$ and $j$ $\in$ $[1$, $n]$. 

%\vspace{-0.2cm} 

\begin{definition}{Data-exchange setting} 
\label{dexchg-def} 
A {\em data-exchange setting} $\cal M$ is a triple $(${\bf S}, {\bf T}, $\Sigma)$, where {\bf S} and {\bf T} are disjoint schemas and $\Sigma$ is a finite set of dependencies over $<${\bf S}, {\bf T}$>$. {\bf S} in $\cal M$ is called the {\em source schema,} and {\bf T} is called the {\em target schema}. 
\end{definition} 

%\vspace{-0.2cm} 

Instances of {\bf S} are called {\em source} instances and are always ground instances. Instances of {\bf T} are {\em target} instances. %It is usual in the data-exchange literature to assume the existence of two disjoint and infinite sets of values that populate instances. One is the set of {\em constants,} denoted by {\sc Const}, and the other one is the set of {\em nulls,} denoted by {\sc Var}. 
%A source instance is always a ground instance. % contained in {\sc Const}, while the active domain of a target instance is contained in {\sc Const} $\cup$ {\sc Var}. %We usually denote constants by lowercase letters $a$, $b$, $c$, $\ldots$, while nulls are denoted by symbols $\perp$, $\perp_1$, $\perp_2$, $\ldots$. 
%Target instances that are consistent with both the source instance and the specification $\Sigma$ are called {\em solutions.} Formally, 
Given a source instance $I$, we say that a target instance $J$ is a {\em solution for} $I$ {\em (under $\cal M$)} %(or simply a {\em solution for} $I$ if $\cal M$ is clear from the context) 
if $(I,$ $J)$ $\models$ $\Sigma$. 
 
%Admitting the full expressive power of FO as a language for specifying dependencies in data exchange, easily leads to unsatisfiability of some fundamental problems, such as checking for the existence of solutions. Thus, 
It is customary in the data-exchange literature to restrict the study to the class of settings whose set $\Sigma$ can be split into two sets $\Sigma_{st}$ and $\Sigma_t$, as follows: % that satisfy the following conditions: 

%\vspace{-0.2cm} 

\begin{enumerate} 
	\item $\Sigma_{st}$ is a set of {\em source-to-target} dependencies {\em (stds),} that is, tgds of the form $\varphi_{{\bf S}}({\bar X})$ $\rightarrow$ $\exists$ $\bar y$ $\psi_{{\bf T}}({\bar X}, {\bar Y})$, where $\varphi_{{\bf S}}({\bar X})$ and $\psi_{{\bf T}}({\bar X}, {\bar Y})$ are conjunctions of relational atoms in {\bf S} and {\bf T}, respectively; and 

%\vspace{-0.1cm} 

	\item $\Sigma_{t}$, the  set of {\em target} dependencies, is the union of a set of tgds and egds defined over the schema {\bf T}. 
\end{enumerate} 

In this current paper, we assume all data-exchange settings to be of the form $\cal M$ $=$ ({\bf S}, {\bf T}, $\Sigma)$, where $\Sigma$ $=$ $\Sigma_{st}$ $\cup$ $\Sigma_t$, for $\Sigma_{st}$ a set of stds and $\Sigma_t$ a set of target dependencies. Intuitively, the stds can be viewed as a tool for specifying how the source data get translated into target data. In addition, the target dependencies are the usual database constraints, to be satisfied by the translated data. The data-exchange settings of this form are not restrictive from the database point of view. % \cite{Barcelo09}. %(See our discussion of embedded dependencies in Section~\ref{dep-chase-prelims-sec}.) %Indeed, tuple-generating dependencies together with equality-generating dependencies precisely capture the class of {dep-chase-prelims-sec\em embedded} dependencies. The latter class contains all relevant dependencies that appear in relational databases, e.g., it contains functional and inclusion dependencies, among others. 

Solutions for a given source instance are not necessarily unique, and there are source instances that have no solutions. {\em Universal solutions} are, intuitively, ``the most general'' solutions among all possible solutions. %To give a precise mathematical definition of which solutions are the most general, we first have to define what a homomorphism between data-exchange instances is. Let $J$ and $J'$ be two instances over the target schema {\bf T}, with values in {\sc Const} $\cup$ {\sc Var}. A {\em homomorphism} $h$: $J$ $\rightarrow$ $J'$ is a mapping from the domain of $J$ into the domain of $J'$, that is the identity on constants, and such that $\bar t$ $=$ $(t_1,\ldots,t_n)$ $\in$ $J(R)$ implies $h({\bar t})$ $=$ $(h(t_1),\ldots,h(t_n))$ is in $J'(R)$ for all $R$ $\in$ {\bf T}. 
Formally, given a solution $J$ for source instance $I$, we say that $J$ is a {\em universal} solution for $I$ if for every solution $J'$ for $I$, there exists a homomorphism from $J$ to $J'$. 
Constructing a universal solution for a given source instance $I$ can be done by chasing $I$ with %chase. The basic idea is the following. The chase starts with the source instance $I$, and then triggers every dependency in 
$\Sigma_{st}$ $\cup$ $\Sigma_t$. The chase may never terminate or may fail; in the latter case, no solution exists \cite{FaginKMP05}. If the chase does not fail and terminates, then the resulting target instance is guaranteed to be a universal solution for $I$.  %Nothing can be said in the case when the chase does not terminate. %Example 4.2 in \cite{Barcelo09} shows an application of the chase procedure. 

The problem of checking for the existence of solutions is known to be undecidable, please see \cite{Barcelo09}. %Thus,  to restrict the class of dependencies allowed in data-exchange settings, in such a way that it satisfies the following: (C1) The existence of solutions implies the existence of universal solutions; (C2) checking the existence of solutions is a decidable (ideally, tractable) problem; and (C3) for every source instance that has a solution, at least one universal solution can be computed (hopefully, in polynomial time). 
At the same time, the following positive result is due to \cite{FaginKMP05}. % the chase always terminates for data-exchange settings with a weakly acyclic set $\Sigma_t$  \cite{FaginKMP05}; in this case the chase of the source instance terminates in at most polynomially many stages. This good behavior also implies the good behavior of this class of settings w.r.t. data exchange, as summarized below: 

%\vspace{-0.1cm} 

\begin{theorem}{\cite{FaginKMP05}}  
\label{barcelo-four-four-thm} 
Let  $\cal M$ $=$ ({\bf S}, {\bf T}, $\Sigma_{st}$ $\cup$ $\Sigma_t$) be a fixed data-exchange setting, such that $\Sigma_{t}$ is weakly acyclic. Then there is a polynomial-time algorithm such that for every source instance $I$, the algorithm decides whether a solution for $I$ exists. Then, whenever a solution for $I$ exists, the algorithm computes a universal solution for $I$ in polynomial time. 
\end{theorem} 

%\vspace{-0.1cm} 

%Thus, the class of settings with a weakly acyclic set of tgds satisfies conditions C1, C2, and C3, as defined above, and therefore, it constitutes a good class for data exchange according to our definition. 
The universal solution of Theorem~\ref{barcelo-four-four-thm}, called the {\em canonical} universal solution \cite{FaginKMP05}, is the result of the chase. 

{\bf Query answering:} Assume that a user poses a query $Q$ over the target schema {\bf T}, and $I$ is a given source instance. %Then, what does it mean to answer $Q$ with respect to $I$? %Clearly, there is an ambiguity here, since there may be many solutions for $I$, and the evaluation of $Q$ over different solutions may give different answers. The fact that one materializes only a single solution does not mean that others should not be taken into account when defining the semantics of the query. 
Then the usual semantics for the query answering is that of  ``certain answers,'' defined as follows. %\cite{Barcelo09}. %Intuitively, an answer is {\em certain} if it occurs in the result of evaluation of query $Q$ over every possible solution $J$. This semantics is independent of a particular solution that is materialized. 
Let $\cal M$ be a data-exchange setting, let $Q$ be a query over the target schema {\bf T} of $\cal M$, and let $I$ be a source instance. We define $certain_{\cal M}(Q, I)$, the set of {\em certain answers of} $Q$ {\em with respect to} $I$ {\em under} $\cal M$, as 

%\vspace{-0.1cm} 

% $\cal M$ $=$ ({\bf S}, {\bf T}, $\Sigma_{st}$ $\cup$ $\Sigma_t$)
\begin{tabbing} 
$certain_{\cal M}(Q, I)$ $=$ $\bigcap \ \{ \ Q(J) \ | \ J$ is a solution for $I \ \}.$
\end{tabbing} 

%\vspace{-0.1cm} 

Computing certain answers for arbitrary FO queries is an undecidable problem. For unions of CQ queries (UCQ queries) 
%which is a query language including all CQ queries, 
we have the following positive result: 

%\vspace{-0.1cm} 

\begin{theorem}{\cite{FaginKMP05}}  
\label{barcelo-five-two-thm} 
Let $\cal M$ $=$ ({\bf S}, {\bf T}, $\Sigma_{st}$ $\cup$ $\Sigma_t$) be a data-exchange setting with $\Sigma_t$ a weakly acyclic set, and let $Q$ be a $UCQ$ query. 
Then the problem of computing certain answers for $Q$ under $\cal M$ can be solved in polynomial time. 
\end{theorem} 

%\vspace{-0.1cm} 

To compute the certain answers to a UCQ query $Q$ w.r.t. a source instance $I$, we first check whether a solution for $I$ exists. If there is no solution, the setting is inconsistent w.r.t. $I$. Otherwise, compute an arbitrary universal solution $J$ for $I$, and then compute the set $Q_{\downarrow}(J)$ of all those tuples in $Q(J)$ that do not contain nulls. It can be shown that $Q_{\downarrow}(J)$ $=$ $certain_{\cal M}(Q, I)$. %Notice that for the class of settings with a weakly acyclic set of tgds 

%\vspace{-0.2cm} 

\subsection{Data Exchange for Finding Certain Query Answers w.r.t. Materialized-View Setting} 
\label{dexchg-work-sec} 

Suppose we are given a valid CQ weakly acyclic mater- ialized-view setting ${\cal M}\Sigma$ $=$ $(${\bf P}, $\Sigma$, $\cal V$, $MV)$ and a CQ query $Q$ of arity $k$ $\geq$ $0$. We consider the problem of finding the set of certain answers to $Q$ w.r.t. the setting ${\cal M}\Sigma$ under CWA. That is, by the definition given in Section~\ref{probl-stmt-defs-sec}, we are interested in finding all (and only) the $k$-ary tuples $\bar t$ of elements of $consts({\cal M}\Sigma)$, such that for all the instances $I$ with $\cal V$ $\Rightarrow_{I,{\Sigma}}$ $MV$,  we have $\bar t$ $\in$ $Q(I)$. 

In this subsection we show how a straightforward reformulation of the pair $({\cal M}\Sigma, Q)$ turns the above problem %, for any such tuple $\bar t$, 
into an instance of the problem of computing certain answers in data exchange. We first construct a set $\Sigma_{st}$ of tgds, as follows. For a view $V$ in the set of views $\cal V$ in  ${\cal M}\Sigma$, consider the query $V({\bar X})$ $\leftarrow$ $body_{(V)}({\bar X}, {\bar Y})$ for $V$. (As ${\cal M}\Sigma$ is a CQ setting, the query for each $V$ $\in$ $\cal V$ is a CQ query.) We associate with this $V$ $\in$ $\cal V$ the tgd $\sigma_V:$ $V({\bar X})$ $\rightarrow$ $\exists {\bar Y}$ $body_{(V)}({\bar X}, {\bar Y})$. We then define the set $\Sigma_{st}$ to be the set of tgds $\sigma_V$ for all $V$ $\in$ $\cal V$. Then the components {\bf P}, $\Sigma$, and $\cal V$ of ${\cal M}\Sigma$ can be reformulated into the following data-exchange setting: 

%\vspace{-0.1cm} 

\begin{tabbing} 
${\cal S}^{(de)}({\cal M}\Sigma)$ $=$ $({\cal V}$, {\bf P}, $\Sigma_{st} \cup \Sigma)$.  
\end{tabbing} 

%\vspace{-0.1cm} 

%\noindent 
Further, we interpret $MV$ in ${\cal M}\Sigma$ as a source instance for ${\cal S}^{(de)}({\cal M}\Sigma)$, and interpret the input query $Q$ as a query on the target schema {\bf P} in ${\cal S}^{(de)}({\cal M}\Sigma)$. We call the triple $({\cal S}^{(de)}({\cal M}\Sigma)$, $MV$, $Q)$ {\em the associated data-exchange instance for} $({\cal M}\Sigma, Q)$.  %Observe that whenever ${\cal M}\Sigma$ is a weakly acyclic CQ instance, $\Sigma$  

%\noindent 
For valid CQ weakly acyclic settings ${\cal M}\Sigma$ and for CQ queries $Q$, we introduce the following algorithm, which we call the {\em data-exchange approach to finding the set of certain query answers w.r.t. a materialized-view setting.} First, we compute the canonical universal solution, $J_{de}^{{\cal M}\Sigma}$, for the source instance $MV$ in the data-exchange setting ${\cal S}^{(de)}({\cal M}\Sigma)$. If $J_{de}^{{\cal M}\Sigma}$ does not exist, then we output the empty set of answers. Otherwise we output, as a set of certain answers to the query $Q$ w.r.t. the setting ${\cal M}\Sigma$, the set %$ground(Q(J_{de}^{{\cal M}\Sigma}))$ 
of all those tuples in $Q(J_{de}^{{\cal M}\Sigma})$ that do not contain nulls. When we assume, similarly to \cite{ZhangM05}, %, same as in Section~\ref{rewriting-sec}, 
that everything in ${\cal M}\Sigma$ is fixed except for $MV$ and $Q$, then from Theorem~\ref{barcelo-five-two-thm} due to \cite{FaginKMP05} we obtain immediately that this algorithm always terminates and runs in polynomial time. %(Please see Section~\ref{dexchg-prelim-sec} for the details.)  
We have shown that this data-exchange approach is sound. (Please see Appendix~\ref{app-dexchg-sec}.) % for CQ inputs. 

It turns out that our data-exchange approach is not complete for (CQ queries and) CQ weakly acyclic settings ${\cal M}\Sigma$ with $\Sigma$ $=$ $\emptyset$, nor for those with $\Sigma$ $\neq$ $\emptyset$. (Please see Appendix~\ref{incomplete-dexchg-approach-sec} for all the details.) We now discuss a feature of the data-exchange approach that prevents us from using  it as a complete algorithm for the problem of finding the set of certain query answers w.r.t. a materialized-view setting under CWA. In Appendix~\ref{vv-dexchg-sec} we will eliminate this feature of the data-exchange approach, in a modification that will yield a sound and complete algorithm for finding the set of certain answers to CQ queries w.r.t. CQ weakly acyclic materialized-view settings under CWA. 

Why is the data-exchange approach not complete when applied to (CQ queries and) CQ weakly acyclic material- ized-view settings? Intuitively, the problem is that its canonical universal solution $J_{de}^{{\cal M}\Sigma}$  ``may cover too many target instances'' (i.e., $J_{de}^{{\cal M}\Sigma}$ is an OWA rather than CWA solution). Let us rewrite the set $MV$ of Example~\ref{app-main-three-ex} using, to save space, constants $c$, $d$, and $f$, as $MV$ $=$ $\{ V(c,d), W(d,f) \}$. Now let us evaluate the queries for the views $V$ and $W$ of Example~\ref{app-main-three-ex} over the canonical solution $J_{de}^{{\cal M}\Sigma}$ $=$ $\{ E(c, d, \perp_1),$ $E(\perp_2, d, f) \}$ for that example. We obtain that the answer to the view $V$ on $J_{de}^{{\cal M}\Sigma}$ is $\{ V(c, d)$, $V(\perp_2, d) \}$.  Similarly, the answer to $W$ on $J_{de}^{{\cal M}\Sigma}$ is $\{ W(d, \perp_1)$, $W(d, f) \}$. Thus, if we replace $\perp_1$ in $J_{de}^{{\cal M}\Sigma}$ by any constant except $f$, or replace $\perp_2$ by any constant except $c$, then any ground instance obtained from $J_{de}^{{\cal M}\Sigma}$ using these replacements  would ``generate too many tuples'' (as compared with $MV$) % in Example~\ref{dexchg-approach-incomplete-sigma-empty-ex}) 
in the answer to either  $V$ or $W$.

We now generalize over this observation. Fix a valid CQ weakly acyclic instance ${\cal M}\Sigma$, and consider the canonical universal solution (if one exists) $J_{de}^{{\cal M}\Sigma}$ generated by the data-exchange approach with ${\cal M}\Sigma$ as input. (In the remainder of this paper, we will refer to $J_{de}^{{\cal M}\Sigma}$ as {\em the canonical data-exchange solution for} ${\cal M}\Sigma$.) By definition of $J_{de}^{{\cal M}\Sigma}$, for each $V$ $\in$ $\cal V$, % we have that 
the answer to the query for $V$ on $J_{de}^{{\cal M}\Sigma}$ is a superset of the relation $MV[V]$. Suppose that the answer on $J_{de}^{{\cal M}\Sigma}$ to at least one view $V$ $\in$ $\cal V$  is not a subset of $MV[V]$, as it is the case in the example we have just discussed.  Then $J_{de}^{{\cal M}\Sigma}$, as a template for instances of schema {\bf P}, describes not only instances that ``generate'' exactly the set $MV$ in ${\cal M}\Sigma$, but also those instances that generate proper supersets of $MV$. The latter instances are not of interest to us. (Recall that we take the CWA viewpoint, and thus are interested only in the instances $I$ of schema {\bf P} such that $\cal V$ $\Rightarrow_{I,{\Sigma}}$ $MV$.) As a result, when the data-exchange approach uses $J_{de}^{{\cal M}\Sigma}$ to obtain certain answers to the input query $Q$, it can easily miss those certain answers that characterize only those instances of interest to us.

%\newpage 

%\newpage 

%\vspace{-0.3cm} 

\section{Technical Details of the Data-Exchange Approach} 
\label{app-dexchg-sec} 

In this appendix we discuss the technical details of the data-exchange approach of Appendix~\ref{dexchg-sec} to finding the set of certain answers to a query w.r.t. a materialized-view setting, under CWA for CQ queries and CQ weakly acyclic settings. %In Section~\ref{vv-dexchg-sec} we will build on this approach to develop a sound and complete algorithm for disclosing information leaks in the CQ weakly acyclic setting. 

\subsection{A Sound Data-Exchange Approach} 
\label{sound-dexchg-approach-sec} 

Suppose that we are given a valid CQ materialized-view setting ${\cal M}\Sigma$ $=$ $(${\bf P}, $\Sigma$, $\cal V$, $MV)$ and a CQ query $Q$ of arity $k$ $\geq$ $0$. 
By the definition given in Section~\ref{probl-stmt-defs-sec}, we are interested in finding all (and only) $k$-ary tuples $\bar t$ of elements of $consts({\cal M}\Sigma)$, such that for all the instances $I$ satisfying $\cal V$ $\Rightarrow_{I,{\Sigma}}$ $MV$,  we have $\bar t$ $\in$ $Q(I)$. 

We now show how a straightforward reformulation of ${\cal M}\Sigma$ turns the above problem %, for any such tuple $\bar t$, 
into an instance of the problem of computing certain answers in data exchange. We first construct a set $\Sigma_{st}$ of tgds, as follows. For a view $V$ in the set of views $\cal V$ in ${\cal M}\Sigma$, consider the query $V({\bar X})$ $\leftarrow$ $body_{(V)}({\bar X}, {\bar Y})$ for $V$. (As ${\cal M}\Sigma$ is a CQ instance, the query for each $V$ in $\cal V$ is a CQ query.) We associate with this $V$ $\in$ $\cal V$ the tgd $\sigma_V:$ $V({\bar X})$ $\rightarrow$ $\exists {\bar Y}$ $body_{(V)}({\bar X}, {\bar Y})$. We then define the set $\Sigma_{st}$ to be the set of tgds $\sigma_V$ for all $V$ $\in$ $\cal V$. Then ${\cal M}\Sigma$ can be reformulated into a data-exchange setting 

\begin{tabbing} 
${\cal S}^{(de)}({\cal M}\Sigma)$ $=$ $($ $\cal V$, {\bf P}, $\Sigma_{st} \cup \Sigma$ $)$, % $J_{de}^{{\cal M}\Sigma}$
\end{tabbing} 

\noindent 
with a source instance $MV$ and a query $Q$ on the target schema {\bf P}. We call the triple $({\cal S}^{(de)}({\cal M}\Sigma),MV,Q)$ {\em the associated data-exchange instance for} ${\cal M}\Sigma$ and $Q$.  %Observe that whenever ${\cal M}\Sigma$ is a weakly acyclic CQ instance, $\Sigma$  

The following observation is immediate from Definition~\ref{certain-query-answer-def}  and from the definitions in Section~\ref{dexchg-prelim-sec}. 

\begin{proposition} 
\label{dexchg-subsumed-by-potential-leaks-prop} 
Given a valid CQ materialized-view setting ${\cal M}\Sigma$ and a CQ query $Q$, with their associated data-exchange instance $({\cal S}^{(de)}({\cal M}\Sigma),MV,Q)$. Then for each tuple $\bar t$ that is a certain answer of $Q$ with respect to $MV$ under the data-exchange setting ${\cal S}^{(de)}({\cal M}\Sigma)$, we have that $\bar t$ is a certain answer to the query $Q$ w.r.t. ${\cal M}\Sigma$ under CWA. 
\end{proposition} 

%\noindent 
For valid CQ weakly acyclic settings ${\cal M}\Sigma$ and CQ queries $Q$, we introduce the following algorithm, which we call the {\em data-exchange approach to finding certain query answers w.r.t. a materialized-view setting.} First, we compute the canonical universal solution, $J_{de}^{{\cal M}\Sigma}$, for the source instance $MV$ in the data-exchange setting ${\cal S}^{(de)}({\cal M}\Sigma)$. If $J_{de}^{{\cal M}\Sigma}$ does not exist, then we output the empty set of answers. Otherwise  we output, as a set of certain answers to $Q$ w.r.t. ${\cal M}\Sigma$, the set %$ground(Q(J_{de}^{{\cal M}\Sigma}))$ 
of all those tuples in $Q(J_{de}^{{\cal M}\Sigma})$ that do not contain nulls. When we assume, same as in \cite{ZhangM05}, that everything in ${\cal M}\Sigma$ is fixed except for $MV$, and assume that $Q$ is not fixed, then from Theorem~\ref{barcelo-five-two-thm} due to \cite{FaginKMP05} we obtain immediately that this algorithm always terminates, constructs the instance $J_{de}^{{\cal M}\Sigma}$ 
in polynomial time, 
 and returns each certain-answer tuple in polynomial time. %(Please see Section~\ref{dexchg-prelim-sec} for the details.)  
By Proposition~\ref{dexchg-subsumed-by-potential-leaks-prop}, this data-exchange approach is sound. 

% that a tuple $\bar t$ of values from $consts({\cal M}\Sigma)$ is a potential information leak in the instance ${\cal M}\Sigma$ if and only it $\bar t$ is a certain answer of the target query $Q$ w.r.t. to the source instance $MV$ under the data-exchange setting ${\cal S}^{(de)}({\cal M}\Sigma)$. 

\subsection{The Data-Exchange Approach Is Not \\ Complete} 
\label{incomplete-dexchg-approach-sec}

By Theorem~\ref{rewr-approach-thm} and Proposition~\ref{dexchg-subsumed-by-potential-leaks-prop}, we have that for CQ weakly acyclic
materialized-view settings ${\cal M}\Sigma$ and for CQ queries $Q$, for all the certain answers to $Q$ w.r.t. ${\cal M}\Sigma$ that can be found using the data-exchange approach of Section~\ref{sound-dexchg-approach-sec}, these certain answers can in principle also be disclosed by an adaptation of our rewriting approach of Section~\ref{rewr-approach-subsec}. (See Note 2 in Section~\ref{one-rewr-enough-sec}.) %(This result is by Proposition~\ref{dexchg-subsumed-by-potential-leaks-prop} and by .) 

\begin{theorem} 
\label{dexchg-subsumed-by-rewr-thm} 
Given a valid CQ weakly acyclic ma- terialized-view setting ${\cal M}\Sigma$ and a CQ query $Q$. Let $\cal T$ be the set of tuples output by the data-exchange approach of Section~\ref{sound-dexchg-approach-sec} when it is applied to the inputs ${\cal M}\Sigma$ and $Q$. Then for each $\bar t$ $\in$ $\cal T$, there exists an $MV$-validated head-instantiated rewriting $R$ for $\bar t$, such that $R$ $\sqsubseteq_{\Sigma,MV,{\cal V}}$ $Q$.   
\end{theorem} 

Even in the light of the result of Theorem~\ref{dexchg-subsumed-by-rewr-thm}, we cannot just abandon the data-exchange approach in favor of the rewriting approach when working with valid CQ weakly acyclic materialized-view settings and CQ queries. It is true that we already have %introduced 
a sound and complete rewriting approach to finding all certain answers to CQ queries w.r.t. valid CQ weakly acyclic materialized-view settings under CWA. However, the rewriting approach works via an explicit generate-and-test paradigm for all the candidate certain-answer tuples, please see Note 2 in Section~\ref{one-rewr-enough-sec}. The advantage of the data-exchange approach in this regard is that we can obtain all the certain-answer tuples for the query $Q$ that are determinable by this (sound) approach, simply by processing $Q$ once on the instance $J_{de}^{{\cal M}\Sigma}$, and by then filtering out all the answer tuples that contain null values. In Appendix~\ref{vv-dexchg-sec} we will introduce a sound {\em and complete} algorithm for finding all the certain-answer tuples to CQ queries w.r.t. CQ weakly acyclic materialized-view settings under CWA. The algorithm of Appendix~\ref{vv-dexchg-sec} (i) uses the idea and approach of data exchange, and is in fact based on the approach of Section~\ref{sound-dexchg-approach-sec}; and (ii) has the same desirable property of ``returning all the certain-answer tuples by processing the input query $Q$ once'' as just discussed in this paragraph in regard to the data-exchange approach of Section~\ref{sound-dexchg-approach-sec}. 

(As each of the generate-and-test rewriting algorithm of Section~\ref{one-rewr-enough-sec} and the algorithm to be introduced in Appendix~\ref{vv-dexchg-sec} is sound and complete for input instances with CQ queries and CQ weakly acyclic materialized-view settings under CWA, these two approaches have of course the same asymptotic complexity, w.r.t. any relevant complexity measure. Our only argument in the previous paragraph in favor of the algorithm of Appendix~\ref{vv-dexchg-sec} is that that algorithm is, in a sense, more streamlined (than the data-exchange approach), as it does not use the generate-and-test paradigm w.r.t. the candidate certain-answer tuples.) 

The reason we are to introduce the algorithm of Appendix~\ref{vv-dexchg-sec} is that, not surprisingly, the data-exchange approach of Section~\ref{sound-dexchg-approach-sec} is not complete under CWA for CQ queries, either for CQ weakly acyclic settings ${\cal M}\Sigma$ with $\Sigma$ $=$ $\emptyset$, or for those with $\Sigma$ $\neq$ $\emptyset$. In the remainder of this appendix, we discuss a feature of the data-exchange approach that prevents us from using  it as a complete algorithm for this class of input instances under CWA. In Appendix~\ref{vv-dexchg-sec} we will eliminate this feature of the data-exchange approach, in a modification that will give us a sound and complete algorithm for finding all the certain-answer tuples for this class of input instances under CWA. % of information leak. 

We now prove that  the data-exchange approach is not complete for CQ instances ${\cal M}\Sigma$ with $\Sigma$ $=$ $\emptyset$. 

\begin{example} 
\label{dexchg-approach-incomplete-sigma-empty-ex} 
We recall the CQ query $Q$ and the CQ views $V$ and $W$ of Example~\ref{app-main-three-ex}: 

\begin{tabbing} 
$Q(X, Z)$ $\leftarrow$ $E(X, Y, Z)$. \\ 
$V(X, Y)$ $\leftarrow$ $E(X, Y, Z)$. \\ 
$W(Y, Z)$ $\leftarrow$ $E(X, Y, Z)$. 
\end{tabbing} 

Using the agreement as in Example~\ref{grounding-ex} for the constants used in Example~\ref{app-main-three-ex}, we represent the set of view answers of Example~\ref{app-main-three-ex} as $MV$ $=$ $\{ \ V(c,d), W(d,f) \ \}$. In the same notation, the tuple $\bar t$ of Example~\ref{app-main-three-ex} is recast as %${\bar t}'$ $=$ 
$(c, f)$. 

Consider the materialized-view setting ${\cal M}\Sigma$ $=$ $(\{ E \}$, $\emptyset$, $\{ V, W \}$, $MV)$, with all the elements as defined above. By definition, ${\cal M}\Sigma$ is a CQ weakly acyclic setting. (${\cal M}\Sigma$ is also valid, by the existence of the instance $\{ (c,d,f) \}$ of schema $\{ E \}$.) The data-exchange approach of Section~\ref{sound-dexchg-approach-sec} applied to ${\cal M}\Sigma$ and $Q$ yields the following canonical universal solution, $J_{de}^{{\cal M}\Sigma}$, for the source instance $MV$ in the data-exchange setting ${\cal S}^{(de)}({\cal M}\Sigma)$: 

\begin{tabbing} 
$J_{de}^{{\cal M}\Sigma}$ $=$ $\{$ $E(c, d, \perp_1)$, $E(\perp_2, d, f)$ $\}$. 
\end{tabbing} 

\noindent 
(The first tuple in $J_{de}^{{\cal M}\Sigma}$ is due to the tuple $V(c, d)$ in $MV$, and the second tuple is due to $W(d,f)$ in $MV$.) It is easy to see that each of the two answers to the query $Q$ on the instance $J_{de}^{{\cal M}\Sigma}$ has nulls and thus cannot qualify as a certain answer to $Q$ w.r.t. ${\cal M}\Sigma$. 
\end{example} 

When given as inputs the setting ${\cal M}\Sigma$ and query $Q$ of Example~\ref{dexchg-approach-incomplete-sigma-empty-ex}, the data-exchange approach of Section~\ref{sound-dexchg-approach-sec} outputs the empty set of candidate-answer tuples. As $Q$ is a CQ query and ${\cal M}\Sigma$ is CQ weakly acyclic (with $\Sigma$ $=$ $\emptyset$),  the sound and complete rewriting-based algorithm of Section~\ref{one-rewr-enough-sec} for finding all the candidate-answer tuples is applicable to ${\cal M}\Sigma$ and $Q$, and outputs $\{$ $(c, f)$ $\}$ when given ${\cal M}\Sigma$ and $Q$ in its input. We conclude that the data-exchange approach is incomplete when applied to CQ queries and CQ weakly acyclic settings with $\Sigma$ $=$ $\emptyset$. Further, we can use the example of Appendix~\ref{sigma-noemptyset-ex-sec} to show that the data-exchange approach is also incomplete when applied to (CQ queries and) CQ weakly acyclic materialized-view settings with $\Sigma$ $\neq$ $\emptyset$. 

Why is the data-exchange approach not complete when applied to (CQ queries and) CQ weakly acyclic material- ized-view settings? Intuitively, the problem is that its canonical universal solution $J_{de}^{{\cal M}\Sigma}$  ``may cover too many target instances'' (i.e., $J_{de}^{{\cal M}\Sigma}$ is an OWA rather than CWA solution). Let us evaluate the queries for the views $V$ and $W$ of Example~\ref{dexchg-approach-incomplete-sigma-empty-ex} over the solution $J_{de}^{{\cal M}\Sigma}$ of that example. We obtain that the answer to the view $V$ on $J_{de}^{{\cal M}\Sigma}$ is $\{ V(c, d)$, $V(\perp_2, d) \}$.  Similarly, the answer to $W$ on $J_{de}^{{\cal M}\Sigma}$ is $\{ W(d, \perp_1)$, $W(d, f) \}$. Thus, if we replace $\perp_1$ in $J_{de}^{{\cal M}\Sigma}$ by any constant except $f$, or replace $\perp_2$ by any constant except $c$, then any ground instance obtained from $J_{de}^{{\cal M}\Sigma}$ using these replacements  would ``generate too many tuples'' (as compared with $MV$) % in Example~\ref{dexchg-approach-incomplete-sigma-empty-ex}) 
in the answer to either  $V$ or $W$. 

We now generalize over this observation. Fix a valid CQ weakly acyclic instance ${\cal M}\Sigma$, and consider the canonical universal solution (if one exists) $J_{de}^{{\cal M}\Sigma}$ generated by the data-exchange approach with ${\cal M}\Sigma$ as input. (In the remainder of this paper, we will refer to $J_{de}^{{\cal M}\Sigma}$ as {\em the canonical data-exchange solution for} ${\cal M}\Sigma$.) By definition of $J_{de}^{{\cal M}\Sigma}$, for each $V$ $\in$ $\cal V$, % we have that 
the answer to the query for $V$ on $J_{de}^{{\cal M}\Sigma}$ is a superset of the relation $MV[V]$. Suppose that the answer on $J_{de}^{{\cal M}\Sigma}$ to at least one view $V$ $\in$ $\cal V$  is not a subset of $MV[V]$, as it is the case in the example we have just discussed.  Then $J_{de}^{{\cal M}\Sigma}$, as a template for instances of schema {\bf P}, describes not only instances that ``generate'' exactly the set $MV$ in ${\cal M}\Sigma$, but also those instances that generate proper supersets of $MV$. The latter instances are not of interest to us. (Recall that we take the CWA viewpoint, and thus are interested only in the instances $I$ of schema {\bf P} such that $\cal V$ $\Rightarrow_{I,{\Sigma}}$ $MV$.) As a result, when the data-exchange approach uses $J_{de}^{{\cal M}\Sigma}$ to obtain certain answers to the input query $Q$, it can easily miss those certain answers that characterize only those instances of interest to us.

\section{View-Verified Data Exchange} % Approach} 
\label{vv-dexchg-sec} 

The problem with the natural data-exchange approach, % of Appendices~\ref{dexchg-sec}--~\ref{app-dexchg-sec}, 
as introduced in \cite{BrodskyFJ00,StoffelS05}, is that its canonical universal solution, when turned into a ground instance, %\footnote{The method is to substitute a distinct ``fresh'' constant for each labeled null.} 
may produce a proper superset of the given set of view answers $MV$. (See Appendices~\ref{dexchg-sec}--~\ref{app-dexchg-sec} in this current paper.) %Sections~\ref{dexchg-work-sec} and~\ref{incomplete-dexchg-approach-sec}.) 
That is, the canonical data-exchange solution does not necessarily describe ground solutions for ${\cal M}\Sigma$ ``tightly enough.'' (Recall that we take the CWA viewpoint, and thus are interested only in the instances $I$ of schema {\bf P} such that $\cal V$ $\Rightarrow_{I,{\Sigma}}$ $MV$. At the same time, the canonical data-exchange solution describes not only these ``CWA'' instances, but also those that are relevant to the inputs under OWA.) 

The approach %to disclosure of information leak 
that we introduce in this appendix builds on data exchange, by ``tightening'' its universal solutions using $consts({\cal M}\Sigma)$. This approach, which we call {\em view-verified data exchange,} solves correctly the problem of finding all the candidate-answer tuples w.r.t. a CQ query and a valid CQ weakly acyclic materialized-view setting. We also use the approach of this appendix to solve the problem of deciding whether a given materialized-view setting is valid. 

%\vspace{-0.2cm} 

\subsection{Chase with {\em MV}-Induced Dependencies} 
\label{mv-chase-sec} 

In Section~\ref{vv-dexchg-def-sec} we will define view-verified data exchange for CQ weakly acyclic input instances. (Throughout this appendix, we use the term ``CQ weakly acyclic input instance'' to refer to a pair $({\cal M}\Sigma,Q)$, where ${\cal M}\Sigma$ is a CQ weakly acyclic materialized-view setting, and $Q$ is a CQ query $Q$ over the schema {\bf P} in ${\cal M}\Sigma$.) Given a ${\cal M}\Sigma$ with set of views $\cal V$ and set of view answers $MV$, the idea of the approach is to force the canonical data-exchange solution $J_{de}^{{\cal M}\Sigma}$  for ${\cal M}\Sigma$ to generate only the relations in $MV$ as answers to the queries for $\cal V$. (By definition of $J_{de}^{{\cal M}\Sigma}$, the answer on $J_{de}^{{\cal M}\Sigma}$ to the query for each $V$ $\in$ $\cal V$ is always a superset of the relation $MV[V]$.) We achieve this goal by chasing $J_{de}^{{\cal M}\Sigma}$ using ``$MV$-induced'' dependencies. Intuitively, applying $MV$-induced dependencies to the instance $J_{de}^{{\cal M}\Sigma}$ forces some nulls in $J_{de}^{{\cal M}\Sigma}$ to become constants in $consts({\cal M}\Sigma)$. As a result of such a chase step, we obtain that for at least one view $V$ $\in$ $\cal V$, some formerly non-ground tuples in the answer to $V$ on the instance become ground tuples in $MV[V]$. 

We now formally define $MV$-induced dependencies. Let $V({\bar X})$ $\leftarrow$ $\phi({\bar X},{\bar Y})$ be a CQ query of arity $k_V$ $\geq$ $0$, and $MV$ be a ground instance of a schema that includes the $k_V$-ary relation symbol $V$. First, in case where $MV[V]$ $=$ $\emptyset$, we define the {\em $MV$-induced implication constraint ($MV$-induced ic)} $\iota_V$ {\em for} $V$ as 
%\begin{tabbing} 
%$\iota_V:$ $\phi({\bar X},{\bar Y})$ $\rightarrow$ $false$.  
%\end{tabbing} 
%\vspace{-0.1cm} 
\begin{equation} 
\label{iota-eqn}  
\iota_V: \ \phi({\bar X},{\bar Y}) \rightarrow \ false.  
\end{equation} 

%\vspace{-0.1cm} 

\noindent 
(Each $MV$-induced ic is an implication constraint, i.e., a Horn rule with the empty head. See \cite{ZhangO97} for the discussion and references on implication constraints.) 

% let $k_V$ $=$ $0$, and let $MV[V]$ be the set $\{ () \}$. Then we define the {\em $MV$-induced generalized egd} $\tau_V$ {\em for} $V$ as  

%\begin{tabbing} 
%$\tau_V:$ $\phi({\bar X},{\bar Y})$ $\rightarrow$ $true$.  
%\end{tabbing} 

%Finally, 

Second, in case where $k_V$ $\geq$ $1$, suppose $MV[V]$ $=$ $\{ {\bar t}_1$, ${\bar t}_2$, $\ldots$, ${\bar t}_{m_V} \}$, with ${m_V}$ $\geq$ $1$. Then we define the {\em $MV$-induced generalized egd ($MV$-induced ged)} $\tau_V$ {\em for} $V$ as  
%\begin{tabbing} 
%$\tau_V:$ $\phi({\bar X},{\bar Y})$ $\rightarrow$ $\bigvee_{i=1}^{m_V} ({\bar X} = {\bar t}_i)$.  
%\end{tabbing} 
%\vspace{-0.1cm} 
\begin{equation} 
\label{tau-eqn}  
\tau_V: \phi({\bar X},{\bar Y}) \rightarrow \vee_{i=1}^{m_V} ({\bar X} = {\bar t}_i).  
\end{equation} 

%\vspace{-0.1cm} 

\noindent 
Here, ${\bar X}$ $=$ $[S_1,\ldots,S_{k_V}]$ is the head vector of the query for $V$, with $S_j$ $\in$ {\sc Const} $\cup$ {\sc Qvar} for  $j$ $\in$ $[1,$ $k_V]$. For each $i$ $\in$ $[1,$ ${m_V}]$ and for the ground tuple $\bar t_i$ $=$ $(c_{i1},$ $\ldots,$ $c_{ik_V})$ $\in$   $MV[V]$, we abbreviate by ${\bar X} = {\bar t}_i$ the conjunction $\wedge_{j=1}^{k_V} (S_j = c_{ij})$. $MV$-induced geds are a straightforward generalization of disjunctive egds of  %as introduced in 
\cite{DeutschT01,FaginKMP05}. 

We now define chase of instances with $MV$-induced dependencies. Consider first $MV$-induced implication constraints.  Given an instance $K$ of schema {\bf P} and an $MV$-induced ic $\iota_V$ as in Eq. (\ref{iota-eqn}), suppose there exists a homomorphism $h$ from the antecedent $\phi({\bar X},{\bar Y})$ of $\iota_V$ to $K$. The intuition here is that we want to make sure that $K$ does not ``generate'' any tuples in the relation $MV[V]$; however, by the existence of $h$, the instance $K$ does generate at least one such tuple. We then say that {\em chase with $\iota_V$ (and $h$) fails on the instance} $K$ {\em and produces the set} $\{ \epsilon \}$, with $\epsilon$ denoting the empty instance. 

Now %for a $V$ $\in$ $\cal V$, 
let $\tau$ as in Eq. (\ref{tau-eqn}) be an $MV$-induced generalized egd for a $V$ $\in$ $\cal V$. % in $\Sigma^{({\cal M}\Sigma)}$. 
The intuition here is that $K$ must ``generate'' {\em only} the tuples in the relation $MV[V]$; we make this happen %via chase, 
by assigning nulls in $K$ to constants in $MV[V]$. (If such assignments are not possible, chase with $\tau$ fails on $K$.) %Our definition of the chase step with $\tau$ is a straightforward extension of the definition of \cite{FaginKMP05} for their disjunctive egds. 
Example~\ref{best-ex} is the running example. 

Our definition of the chase step with $\tau$ as in Eq. (\ref{tau-eqn}) is a straightforward extension of the definition of \cite{FaginKMP05} for their disjunctive egds, as follows. Consider the consequent of $\tau$, of the form $\vee_{i=1}^{m_V} ({\bar X} = {\bar t}_i)$. Recall that for each $i$ $\in$ $[1,$ $m_V]$, the expression ${\bar X} = {\bar t}_i$ is of the form $\wedge_{j=1}^{k_V} (S_j = c_{ij})$. Denote by $\tau^{(1)}$, $\ldots$, $\tau^{(m_V)}$ the following $m_V$ dependencies obtained from $\tau$: $(\phi({\bar X},{\bar Y})$ $\rightarrow$ ${\bar X} = {\bar t}_1)$, $\ldots$, $(\phi({\bar X},{\bar Y})$ $\rightarrow$ ${\bar X} = {\bar t}_{m_V})$, and call them {\em the dependencies associated with} $\tau$. For each $i$ $\in$ $[1,$ $m_V]$, $\tau^{(i)}$ is an embedded dependency that can be equivalently represented by $k_V$ egds %\footnote{The definitions in \cite{FaginKMP05} permit constants in egds} 
$\tau^{(i,1)}$, $\ldots$, $\tau^{(i,k_V)}$.  Here, for each $j$ $\in$ $[1,$ $k_V]$, the egd $\tau^{(i,j)}$ is $\phi({\bar X},{\bar Y})$ $\rightarrow$ $S_j = c_{ij}$. 

Given a $\tau$ as in Eq. (\ref{tau-eqn}) and an instance $K$ of schema {\bf P}, suppose that there exists a homomorphism $h$ from $\phi({\bar X},{\bar Y})$ to $K$ such that $\wedge_{j=1}^{k_V} (h(S_j) = h(c_{ij}))$ is not a tautology for any $i$ $\in$ $[1,$ $m_V]$. Then we say that $\tau$ {\em is applicable to $K$ with the homomorphism} $h$. It is easy to see that it is also the case that each of $\tau^{(1)}$, $\ldots$, $\tau^{(m_V)}$ can be applied to $K$ with % fthe homomorphism 
$h$. That is, for each $i$ $\in$ $[1,$ $m_V]$, the chase of $K$ is applicable with at least one egd $\tau^{(i,j)}$ in the equivalent representation of $\tau^{(i)}$ as a set of egds. For each  $i$ $\in$ $[1,$ $m_V]$, let $K_i$ be the result of applying all the egds $\tau^{(i,1)}$, $\ldots$, $\tau^{(i,k_V)}$ to $K$ with %the homomorphism 
$h$. Note that chase with $\tau^{(i,j)}$ and $h$ can fail on $K$ for some $i$ and $j$. For each such $i$, we say that {\em chase with $\tau^{(i)}$ %(and $h$) 
fails on} $K$ {\em and produces the empty instance} $\epsilon$.  %(If $\wedge_{j=1}^{k_V} (h(S_j) = h(c_{ij}))$ evaluates to $false$, then we say that {\em chase with $\tau$ %(and $h$) fails on} $K$ {\em and produces the set} $\{ \epsilon \}$.)  

Similarly to \cite{FaginKMP05}, we distinguish two cases: 
\begin{itemize} 
	\item If the set $\{ K_{1},$ $\ldots$, $K_{m_V} \}$ contains only empty instances, %chase with $\tau^{(i)}$ fails on $K$ for all $i$ $\in$ $[1,$ $m_V]$, 
we say that {\em chase with $\tau$ (and $h$) fails on} $K$ {\em and produces the set}   $\{ \epsilon \}$. 
%\vspace{-0.2cm} 
	\item Otherwise, let ${\cal K}^{(\tau)}$ $=$ $\{ K_{i_1},$ $\ldots$, $K_{i_p} \}$ be the set of all nonempty elements of $\{ K_{1},$ $\ldots$, $K_{m_V} \}$. We say that ${\cal K}^{(\tau)}$ {\em is the result of applying $\tau$ to $K$ with} $h$. % is the set $\{ K_{i_1},$ $\ldots$, $K_{i_p} \}$. 
\end{itemize} 

%Note that in the case where $\tau$ as in Eq. (\ref{tau-eqn}) has a single equality atom in its consequent, the above definition degenerates to the standard definition of chase step with an egd, see Section~\ref{dep-chase-prelims-sec}. 

Similarly to the approach of \cite{FaginKMP05}, in addition to chase steps with $MV$-induced dependencies we will also use chase steps with egds and tgds as in Section~\ref{dep-chase-prelims-sec}. For the chase step of each type, we will use the set notation for uniformity: $K$ $\Rightarrow^{\sigma,h}$ ${\cal K}'$ denotes that a chase step with dependency $\sigma$ and homomorphism $h$ applied to instance $K$ yields a set of instances ${\cal K}'$. Whenever chase with an egd fails on $K$, the set ${\cal K}'$ is the set $\{ \epsilon \}$ by convention; in all other cases where $\sigma$ is an egd or a tgd, the set ${\cal K}'$ is a singleton set.  For $\sigma$ of the form as in Eq. (\ref{iota-eqn})--(\ref{tau-eqn}), the set ${\cal K}'$ is in some cases $\{ \epsilon \}$ as defined above.  

\begin{definition}{$MV$-enhanced chase} Let $\Sigma$ be a set of egds and tgds, let $\Sigma^{(MV)}$  be a set of $MV$-induced dependencies, and let $K$ be an instance. 
\begin{itemize}
%\vspace{-0.1cm} 
	\item A {\em chase tree of $K$ with} $\Sigma$ $\cup$ $\Sigma^{(MV)}$ is a tree (finite or infinite) such that: 
	\begin{itemize}
		\item The root is $K$, and 
%\vspace{-0.1cm} 
		\item For every node $K_j$ in the tree, let ${\cal K}_j$ be the set of its children. Then there must exist some dependency $\sigma$ in $\Sigma$ $\cup$ $\Sigma^{(MV)}$ and some homomorphism $h$ such that $K_j$ $\Rightarrow^{\sigma,h}$ ${\cal K}_j$. 
	\end{itemize} 
%\vspace{-0.2cm} 
	\item A {\em finite $MV$-enhanced chase of $K$ with} $\Sigma$ $\cup$ $\Sigma^{(MV)}$ is a finite chase tree $\cal T$, such that for each leaf $K_p$ of $\cal T$, we have that either (a) $K_p$ is $\epsilon$, or (b) %($K_p$ is not $\epsilon$, and)  
	there is no dependency $\sigma$ in $\Sigma$ $\cup$ $\Sigma^{(MV)}$ and no homomorphism $h$ such that $\sigma$ can be applied to $K_p$ with $h$. 
\end{itemize} 
%\vspace{-0.6cm} 
\end{definition} 

\begin{example} 
\label{best-ex} 
Consider $Q$ as in Example~\ref{app-main-three-ex} and ${\cal M}\Sigma$ $=$ $(\{ E \}$, $\emptyset$, $\{ V, W \}$, $MV)$, with all the elements except $MV$ as in Example~\ref{app-main-three-ex}.\footnote{Please see Example~\ref{dexchg-approach-incomplete-sigma-empty-ex} for the details.} For this current example, we define the set $MV$ as 

\begin{tabbing} 
$MV$ $=$ $\{ \ V(c,d), V(g,d), W(d,f) \ \}$. % is obtained by adding the tuple $V(g,d)$ to the $MV$ of Example~\ref{dexchg-approach-incomplete-sigma-empty-ex}.  
\end{tabbing} 

By definition, ${\cal M}\Sigma$ paired with $Q$ is a CQ %weakly acyclic 
instance with $\Sigma$ $=$ $\emptyset$. ${\cal M}\Sigma$ is also valid, as witnessed by the instance $\{ E(c,d,f),$ $E(g,d,f) \}$. % of $\{ E \}$.  
The data-exchange approach of Appendix~\ref{dexchg-sec} %applied to ${\cal M}\Sigma$ 
yields the following canonical data-exchange solution $J_{de}^{{\cal M}\Sigma}$ for ${\cal M}\Sigma$: 

%\vspace{-0.1cm} 

\begin{tabbing} 
$J_{de}^{{\cal M}\Sigma}$ $=$ $\{$ $E(c, d, \perp_1)$, $E(g, d, \perp_2)$, $E(\perp_3, d, f)$ $\}$. 
\end{tabbing} 

%\vspace{-0.1cm} 

\noindent 
The set of answers without nulls to the query $Q$ on $J_{de}^{{\cal M}\Sigma}$ is empty. Thus, the data-exchange approach applied to %the instance 
${\cal M}\Sigma$ discovers no certain answers to the query $Q$ w.r.t. the setting ${\cal M}\Sigma$. % when . 

In applying the view-verified data-exchange approach to the input $({\cal M}\Sigma,Q)$, we first construct the $MV$-induced generalized egds, $\tau_V$ and $\tau_W$, one for each of the two views in ${\cal M}\Sigma$. (As $MV$ has no empty relations, we do not need to construct $MV$-induced ics for ${\cal M}\Sigma$.)
 
%\vspace{-0.1cm} 

\begin{tabbing} 
$\tau_V: E(X,Y,Z) \rightarrow (X = c \wedge Y = d) \vee (X = g \wedge Y = d).$ \\ 
$\tau_W: E(X,Y,Z) \rightarrow (Y = d \wedge Z = f).$
\end{tabbing} 

%\vspace{-0.1cm} 

\noindent 
The two dependencies associated with $\tau_V$ are $\tau^{(1)}_V:$\linebreak $E(X,Y,Z) \rightarrow (X = c \wedge Y = d)$ and $\tau^{(2)}_V: E(X,Y,Z) \rightarrow (X = g \wedge Y = d).$ Each of $\tau^{(1)}_V$ and $\tau^{(2)}_V$ can be equivalently represented by two egds. For instance, the egd representation for $\tau^{(1)}_V$ is via $\tau^{(1,1)}_V: E(X,Y,Z) \rightarrow X = c$ and $\tau^{(1,2)}_V: E(X,Y,Z) \rightarrow Y = d$.  Similarly, there is one dependency $\tau^{(1)}_W$ ($=$ $\tau_W$) associated with $\tau_W$; an equivalent representation of $\tau^{(1)}_W$ is via two egds. 

Consider a homomorphism $h^{(1)}_V:$ $\{ X \rightarrow c,$ $Y \rightarrow d,$ $Z \rightarrow \perp_1 \}$ from the antecedent $E(X,Y,Z)$ of $\tau_V$ to the instance $J_{de}^{{\cal M}\Sigma}$. As applying $h^{(1)}_V$ to the consequent of $\tau^{(1)}_V$ gives us the tautology $(c = c \wedge d = d)$, we conclude that $\tau_V$ is not applicable to $J_{de}^{{\cal M}\Sigma}$ with %the homomorphism 
$h^{(1)}_V$. 

Consider now the homomorphism $h^{(2)}_V:$ $\{ X \rightarrow \perp_3,$ $Y \rightarrow d,$ $Z \rightarrow f \}$ from the antecedent  of $\tau_V$ to $J_{de}^{{\cal M}\Sigma}$. 
Applying $h^{(2)}_V$ to the consequent of $\tau_V$ gives us the expression $(\perp_3 = c \wedge d = d) \vee (\perp_3 = g \wedge d = d),$ which has no tautologies among its disjuncts. Thus, $\tau_V$ is applicable to $J_{de}^{{\cal M}\Sigma}$ with %the homomorphism 
$h^{(2)}_V$. The chase step with $\tau_V$ and $h^{(2)}_V$ transforms %the instance 
$J_{de}^{{\cal M}\Sigma}$ into % two 
instances $J_1$ and $J_2$, as follows. 

%\vspace{-0.1cm} 

\begin{tabbing} 
$J_1$ $=$ $\{$ $E(c, d, \perp_1)$, $E(g, d, \perp_2)$, $E(c, d, f)$ $\}$. \\ 
$J_2$ $=$ $\{$ $E(c, d, \perp_1)$, $E(g, d, \perp_2)$, $E(g, d, f)$ $\}$. 
\end{tabbing} 

%\vspace{-0.1cm} 

\noindent 
($J_1$ results from assigning $\perp_3 := c$, and $J_2$ from $\perp_3 := g$.) 

We then use the same procedure to apply $\tau_W$ to each of $J_1$ and $J_2$. % using the same procedure. % as the one used for $\tau_V$. 
In each case, the chase steps assign the value $f$ to each of $\perp_1$ and $\perp_2$. As a result, the following instance $J_{vv}^{{\cal M}\Sigma}$ is obtained from each of $J_1$ and $J_2$: 

%\vspace{-0.1cm} 

\begin{tabbing} 
$J_{vv}^{{\cal M}\Sigma}$ $=$ $\{$ $E(c, d, f)$, $E(g, d, f)$ $\}$. 
\end{tabbing} 
%\vspace{-0.6cm} 
\end{example}  

%\vspace{-0.1cm} 

\subsection{Solving CQ Weakly Acyclic Instances} % ${\cal M}\Sigma$} 
\label{vv-dexchg-def-sec} 

% In this subsection we 
We now define the view-verified data-exchange approach to the problem of finding all certain answers to queries w.r.t. materialized-view-settings. 

Let ${\cal M}\Sigma$ $=$ $(${\bf P}, $\Sigma$, $\cal V$, $MV)$ be a CQ materialized-view setting. Then {\em the set $\Sigma^{({\cal M}\Sigma)}$ of $MV$-induced dependencies for} ${\cal M}\Sigma$ is a set of up to $|{\cal V}|$ elements, %defined 
as follows. For each $V$ $\in$ $\cal V$ such that $k_V$ $\neq$ $0$ or $MV[V]$ $\neq$ $\{ () \}$, %the set 
$\Sigma^{({\cal M}\Sigma)}$ has one $MV$-induced implication constraint or one $MV$-induced generalized egd, by the rules as in Eq. (\ref{iota-eqn})--(\ref{tau-eqn}) in Section~\ref{mv-chase-sec}.\footnote{We omit from $\Sigma^{({\cal M}\Sigma)}$ the dependencies, of the form $\phi({\bar X},{\bar Y}) \rightarrow true$, for the case where $k_V = 0$ {\em and} $MV[V]$ $\neq$ $\emptyset$. By the results in this appendix, adding these dependencies to $\Sigma^{({\cal M}\Sigma)}$ would not change any chase results.}  

For CQ weakly acyclic input instances $({\cal M}\Sigma,Q)$ we introduce the following {\em view-verified data-exchange approach to finding certain query answers w.r.t. a materi- alized-view setting.} First, we compute (as in Appendix~\ref{dexchg-sec}) the canonical universal solution $J_{de}^{{\cal M}\Sigma}$ for the source instance $MV$ in the data-exchange setting ${\cal S}^{(de)}({\cal M}\Sigma)$. If $J_{de}^{{\cal M}\Sigma}$ does not exist, %then 
we stop and output the answer that  ${\cal M}\Sigma$ is not valid. Otherwise we obtain a chase tree of $J_{de}^{{\cal M}\Sigma}$ with $\Sigma$ $\cup$ $\Sigma^{({\cal M}\Sigma)}$, where $\Sigma^{({\cal M}\Sigma)}$ is the set of $MV$-induced dependencies for ${\cal M}\Sigma$. If the chase tree is finite, denote by ${\cal J}^{{\cal M}\Sigma}_{vv}$ the set of all the nonempty leaves of the tree. %denoted by ${\cal J}^{{\cal M}\Sigma}_{vv}$, and 
We call each $J$ $\in$ ${\cal J}^{{\cal M}\Sigma}_{vv}$ a {\em view-verified universal solution for} ${\cal M}\Sigma$. 
If ${\cal J}^{{\cal M}\Sigma}_{vv}$ $=$ $\emptyset$, then we stop and output the answer that  ${\cal M}\Sigma$ is not valid. 
Otherwise, for each $J$ $\in$ ${\cal J}^{{\cal M}\Sigma}_{vv}$ we compute the set $Q_{\downarrow}(J)$ of all the tuples in $Q(J)$ that do not contain nulls. Finally, the output of the approach for the input $({\cal M}\Sigma,Q)$  
is the set 
%\vspace{-0.1cm} 
\begin{equation} 
\label{vv-eqn} 
\bigcap_{J \in {\cal J}^{{\cal M}\Sigma}_{vv}} Q_{\downarrow}(J). 
\end{equation} 

%\vspace{-0.1cm} 

The view-verified data-exchange approach to the problem of finding all certain answers to queries w.r.t. materi- alized-view-settings addresses the shortcoming of the data-exchange approach, see Appendix~\ref{dexchg-sec}. Recall that the canonical universal solution $J_{de}^{{\cal M}\Sigma}$ of the latter approach might not cover ``tightly enough'' all the instances of interest to the attackers. In the view-verified approach, we address this problem, by using our extension of the chase to generate from $J_{de}^{{\cal M}\Sigma}$ a set ${\cal J}^{{\cal M}\Sigma}_{vv}$ of instances that are each ``tighter'' than $J_{de}^{{\cal M}\Sigma}$ in this sense. 

In Section~\ref{vv-correctness-sec} we will show that the view-verified data-exchange approach is a sound and complete algorithm for the problem of finding all certain answers to queries w.r.t. materialized-view-settings, in all cases where the input instances are CQ weakly acyclic. In particular, we will see that the set ${\cal J}^{{\cal M}\Sigma}_{vv}$ is well defined, in that the chase tree in the view-verified data-exchange approach is always finite. We will also see that the set ${\cal J}^{{\cal M}\Sigma}_{vv}$ is ``just tight enough,'' in the following sense: Recall (see Section~\ref{probl-stmt-defs-sec}) the definition of $certain_{{\cal M}\Sigma}(Q)$, i.e., of the set of certain answers of query $Q$ w.r.t. materialized-view setting ${\cal M}\Sigma$. 
%tuples that are in the answer to the secret query $Q$ on {\em all} the ground instances $I$ such that $\cal V$ $\Rightarrow_{I,\Sigma}$ $MV$. 
 Then the expression in Eq. (\ref{vv-eqn}), which is the intersection of all the ``certain-answer expressions'' for $Q$ and for the individual elements of the set ${\cal J}^{{\cal M}\Sigma}_{vv}$, is exactly the set $certain_{{\cal M}\Sigma}(Q)$. %, $Q({\cal J}^{{\cal M}\Sigma}_{vv})$. 
 
\begin{example} 
\label{second-best-ex} 
Recall the input instance $({\cal M}\Sigma,Q)$ of Example~\ref{best-ex}, and the instance $J_{vv}^{{\cal M}\Sigma}$ obtained in that example. 
%\noindent  
$J_{vv}^{{\cal M}\Sigma}$ is the (only) view-verified universal solution for ${\cal M}\Sigma$. The set of answers without nulls to the query $Q$ on $J_{vv}^{{\cal M}\Sigma}$ is $\{ (c,f), (g, f) \}$. Thus, $(c,f)$ and $(g, f)$ are certain answers of the query $Q$ w.r.t. the materialized-view setting ${\cal M}\Sigma$, as computed for the instance $({\cal M}\Sigma,Q)$ by the view-verified data-exchange approach. Both $(c,f)$ and $(g, f)$ %tuples 
(and nothing else) are also discovered by the rewriting algorithm of Appendix~\ref{app-rewriting-sec}, %~\ref{sigma-emptyset-sec}, 
which is sound and complete for $({\cal M}\Sigma,Q)$. (See Note 2 in Section~\ref{one-rewr-enough-sec}.) 
\end{example} 

\subsection{Correctness, Validity, and Complexity} 
\label{vv-correctness-sec} 

In this subsection, we show that the  view-verified data-exchange approach is sound and complete for all CQ weakly acyclic input instances, and discuss its runtime and space complexity.  %instances ${\cal M}\Sigma$. % of information leak. 
We also show how the approach %view-verified data exchange 
can be used to decide whether a CQ weakly acyclic materialized-view setting ${\cal M}\Sigma$ is valid.

 %We also discuss the implications of this result for the problem of {\em finding all certain-answer tuples} for a query w.r.t. a materialized-view setting.  

%We now prove that the view-verified data-exchange approach is a sound and complete algorithm for the problem of disclosing information leaks for CQ weakly acyclic inputs, and determine the runtime complexity of the algorithm. For the remainder of the paper we assume, same as before and similarly to \cite{ZhangM05}, that the {\em size} of a given instance ${\cal M}\Sigma$ is the size of its instance $MV$ and query $Q$, with the remaining elements of ${\cal M}\Sigma$ being fixed. 

{\bf View-verified data exchange is an algorithm.} % for CQ weakly acyclic instances  of information leak.} 
We begin by obtaining a basic observation that builds on the results of \cite{FaginKMP05} for chase with tgds and disjunctive egds (as they are defined in \cite{FaginKMP05}). It is immediate from Proposition~\ref{vv-sizes-prop} that view-verified data exchange always terminates 
in finite time  for CQ weakly acyclic inputs. % instances  of information leak. 

\begin{proposition} 
\label{vv-sizes-prop} 
Given a CQ weakly acyclic mate- rialized-view setting ${\cal M}\Sigma$, such that its canonical data-exchange solution $J_{de}^{{\cal M}\Sigma}$ exists. Assume that everything in ${\cal M}\Sigma$ is fixed except for the instance $MV$. Then we have that: 
\begin{itemize} 
	\item[(1)] $MV$-enhanced chase of $J_{de}^{{\cal M}\Sigma}$ with $\Sigma$ $\cup$ $\Sigma^{({\cal M}\Sigma)}$ is a finite tree, $\cal T$, such that: 
%\vspace{-0.1cm} 
		\begin{itemize} 
	  		\item[(a)] $\cal T$ is of polynomial depth in the size of $MV$, and 
			\item[(b)] The number of leaves in $\cal T$ is up to exponential in the size of $MV$; and 
		\end{itemize} 
%\vspace{-0.2cm} 
	\item[(2)] For each nonempty leaf $J$ of %the tree 
	$\cal T$, we have that: 
%\vspace{-0.1cm} 
		\begin{itemize} 
	  		\item[(a)] $J$ is of polynomial size in the size of $MV$, and 
			\item[(b)] Each grounded version of $J$ is a $\Sigma$-valid base instance for $\cal V$ and $MV$.  
		\end{itemize} 
\end{itemize} 
%\vspace{-0.6cm} 
\end{proposition}  

%\vspace{-0.1cm} 

\noindent 
(A {\em grounded version} of instance $K$ %is the 
results from replacing consistently all its nulls %in $K$ 
with distinct %fresh 
new constants.) 

%We will be using the following notion. For an instance $K$, let us refer to the process of replacing consistently all nulls in $K$ with distinct fresh constants as {\em grounding} $K$; we call each result of this process {\em a grounded version of} $K$. %We can show that the grounded version of a given instance is unique up to isomorphism. 

The proof of Proposition~\ref{vv-sizes-prop} relies heavily on the results of \cite{FaginKMP05}, particularly on its Theorem 3.9. Recall the ``decomposition,'' in Section~\ref{mv-chase-sec}, of $MV$-induced generalized egds into egds that are defined as in Section~\ref{dep-chase-prelims-sec}. Intuitively, given a CQ weakly acyclic materialized-view setting ${\cal M}\Sigma$ and for each node $K$ on each path from the root $J_{de}^{{\cal M}\Sigma}$ of the tree $\cal T$ for ${\cal M}\Sigma$, we can obtain $K$ by chasing the root of $\cal T$ using only egds and weakly acyclic tgds.\footnote{Besides the egds and tgds of Section~\ref{dep-chase-prelims-sec}, chase on each path in $\cal T$ may use  $MV$-induced implication constraints. However, the only role of the latter constraints is to obtain the %empty 
instance $\epsilon$ and thus to terminate the respective path in $\cal T$.} The key observation here is that even though the set $\Sigma^{({\cal M}\Sigma)}$ of dependencies is not fixed (in fact, its size is linear in the size of the instance $MV$ in ${\cal M}\Sigma$), all the constants that contribute to the size of $\Sigma^{({\cal M}\Sigma)}$ are already used in the root $J_{de}^{{\cal M}\Sigma}$ of the tree $\cal T$, by definition of $J_{de}^{{\cal M}\Sigma}$. In addition, the antecedent of each $MV$-induced generalized egd in $\Sigma^{({\cal M}\Sigma)}$ is of constant size, by definition of the size of ${\cal M}\Sigma$. %(Recall that the sizes of all the views in ${\cal M}\Sigma$ are fixed.) 
As a result, we can build on Theorem 3.9 and Proposition 5.6 of \cite{FaginKMP05} to obtain items (1)(a) and (2)(a) of our Proposition~\ref{vv-sizes-prop}. 

Item (2)(b) of Proposition~\ref{vv-sizes-prop} is by definition of $MV$-enhanced chase, and % item 
(1)(b) is by construction of the %chase 
tree $\cal T$. Appendix~\ref{expon-vv-number-app} provides a lower bound, via an example where for a CQ instance ${\cal M}\Sigma$ with $\Sigma$ $=$ $\emptyset$, the number of leaves in a chase tree  is exponential in the size of ${\cal M}\Sigma$. 

{\bf Soundness and completeness.} By Proposition~\ref{vv-sizes-prop} (2)(b), the view-verified data-exchange approach  is a {\em complete} algorithm when applied to CQ weakly acyclic input instances $({\cal M}\Sigma,Q)$. (That is, for each certain-answer tuple $\bar t$ for a problem input in this class, view-verified data exchange outputs $\bar t$.) We now make a key observation toward a proof that this algorithm is also {\em sound} for such instances. (Soundness means that for each tuple $\bar t$ that this approach outputs for an input $({\cal M}\Sigma,Q)$ in this class, $\bar t$ is a  certain-answer tuple for $({\cal M}\Sigma,Q)$.) 

%\vspace{-0.1cm} 

\begin{proposition} 
\label{each-valid-ground-instance-mapped-from-vv-prop} 
Given a CQ weakly acyclic mate- rialized-view setting ${\cal M}\Sigma$ $=$ $(${\bf P}, $\Sigma$, $\cal V$, $MV)$ and a CQ query $Q$. %, with set of views $\cal V$ and set of view answers $MV$. 
Then, for each instance $I$ such that $\cal V$ $\Rightarrow_{I,\Sigma}$ $MV$, %$\Sigma$-valid ground base instance $I$ for $\cal V$ and $MV$, 
there exists a homomorphism from some view-verified universal solution for ${\cal M}\Sigma$ to $I$.  
\end{proposition} 

%\vspace{-0.1cm} 

The intuition for the proof of Proposition~\ref{each-valid-ground-instance-mapped-from-vv-prop} is as follows. For a given ${\cal M}\Sigma$, whenever an instance $I$ exists such that $\cal V$ $\Rightarrow_{I,\Sigma}$ $MV$, a canonical data-exchange solution $J_{de}^{{\cal M}\Sigma}$ for ${\cal M}\Sigma$ must also exist. By definition of $J_{de}^{{\cal M}\Sigma}$, there must be a homomorphism from $J_{de}^{{\cal M}\Sigma}$ to the instance $I$. We then start applying $MV$-enhanced chase to $J_{de}^{{\cal M}\Sigma}$, to simulate some rooted path, $P_{({\cal T})}$, in the chase tree $\cal T$ for ${\cal M}\Sigma$. (The tree is finite by Proposition~\ref{vv-sizes-prop}.) In following the path $P_{({\cal T})}$ via the chase, we make sure that there is a homomorphism from each node in the path to $I$, by always choosing an ``appropriate'' associated dependency $\tau^{(i)}$ for each $MV$-induced generalized egd $\tau$ that we are applying in the chase. By $\cal V$ $\Rightarrow_{I,\Sigma}$ $MV$, such a choice always exists, and the path $P_{({\cal T})}$ terminates in finite time in a nonempty instance, $J$. By definition, $J$ is a view-verified universal solution for ${\cal M}\Sigma$. By our simulation of the path $P_{({\cal T})}$ ``on the way to'' $I$, there exists a homomorphism from $J$ to $I$.  

{\bf Validity of setting ${\cal M}\Sigma$.} By the results of \cite{FaginKMP05}, when for a given ${\cal M}\Sigma$ no canonical data-exchange solution %$J_{de}^{{\cal M}\Sigma}$ 
exists, then ${\cal M}\Sigma$ is not a valid setting. We refine this observation into a sufficient and necessary condition for validity of CQ weakly acyclic materialized-view settings ${\cal M}\Sigma$. (The only-if part of Proposition~\ref{when-ds-valid-prop} follows from  % is immediate from 
Proposition~\ref{each-valid-ground-instance-mapped-from-vv-prop}, and its if part is by % immediate from 
Proposition~\ref{vv-sizes-prop} (2)(b).) 

%{\bf [[[ Validity of ${\cal M}\Sigma$: (I)  When the canonical dexchg solution does not exist then no data-exchange solution exists; but a data-exchange solution produces {\em at least} $MV$ and may or may not be $\Sigma$-compliant; (II) When the canonical dexchg solution does exist, conceptually there can still be no $\Sigma$-valid ground base instance for $\cal V$ and $MV$. Our result is as follows: Given a CQ weakly acyclic instance ${\cal M}\Sigma$ of information leak, such that its canonical data-exchange solution $J_{de}^{{\cal M}\Sigma}$ exists. Then ${\cal M}\Sigma$ is a valid instance of information leak if and only if the set ${\cal J}^{{\cal M}\Sigma}_{vv}$ of view-verified universal solutions for ${\cal M}\Sigma$ is not empty. ]]]} 

%\vspace{-0.1cm} 

\begin{proposition} 
\label{when-ds-valid-prop} 
Given a CQ weakly acyclic mate- rialized-view setting ${\cal M}\Sigma$, the setting ${\cal M}\Sigma$ is valid iff the set ${\cal J}^{{\cal M}\Sigma}_{vv}$ of view-verified universal solutions for ${\cal M}\Sigma$ is not empty. 
\end{proposition} 

%\vspace{-0.1cm} 

{\bf Correctness of view-verified data exchange.} By Proposition~\ref{each-valid-ground-instance-mapped-from-vv-prop}, view-verified data exchange is sound. By Proposition~\ref{when-ds-valid-prop}, it outputs a set of certain-answer tuples iff its input is valid. % instance. 
We now conclude: 
% as follows: 

%\vspace{-0.1cm} 

\begin{theorem} 
\label{vv-correctness-thm} 
View-verified data exchange is a sound and complete algorithm for finding all certain answers to CQ queries w.r.t. CQ weakly acyclic materialized-view-settings. 
\end{theorem} 

%\vspace{-0.1cm} 

{\bf Complexity of view-verified data exchange for CQ weakly acyclic input instances.} By Theorem~\ref{vv-correctness-thm}, view-verified data exchange is an algorithm for all CQ weakly acyclic input instances. We now obtain an expon- ential-time upper bound on the runtime complexity of the view-verified data-exchange approach, %\footnote{Recall that for a given CQ weakly acyclic input instance $({\cal M}\Sigma,Q)$, the view-verified data-exchange approach of Section~\ref{vv-dexchg-def-sec} solves the problem of finding all certain-answer tuples for a query and a setting, rather than the problem of determining whether a given ground tuple is a certain answer to $Q$ w.r.t. ${\cal M}\Sigma$.} of Section~\ref{vv-dexchg-def-sec}, 
as follows. 

Given a CQ weakly acyclic input instance $({\cal M}\Sigma,Q)$, the runtime of the approach of Section~\ref{vv-dexchg-def-sec} is exponential in the size of $Q$ and of the set of answers $MV$ in ${\cal M}\Sigma$, assuming that the rest of ${\cal M}\Sigma$ is fixed. This complexity setting extends naturally that of \cite{ZhangM05}: Zhang and Mendelzon in \cite{ZhangM05} assumed for their problem that the base schema and the view definitions are fixed, whereas the set of view answers $MV$ and the queries posed on the base schema in presence of $MV$ can vary. The authors of \cite{ZhangM05} did not consider dependencies on the base schema; we follow the standard data-exchange assumption, see, e.g., \cite{FaginKMP05}, that the given dependencies are fixed rather than being part of the problem input. 

% (as we already did in Section~\ref{vv-correctness-sec}), that all elements of ${\cal M}\Sigma$ except $MV$ are fixed, and that $Q$ is not fixed. That is, the {\em size} of a given input instance $({\cal M}\Sigma,Q,{\bar t})$ is the size of its set of view answers $MV$ and of its query $Q$, with the remaining elements of $({\cal M}\Sigma,Q,{\bar t})$ being fixed. ) 

To obtain the above exponential-time upper bound for the problem of view-verified data exchange for CQ weakly acyclic input instances $({\cal M}\Sigma,Q)$, we analyze the following flow for the view-verified data-exchange algorithm of Section~\ref{vv-dexchg-def-sec}. First, we spend exponential time in the arity $k$ of $Q$ to generate all the $k$-ary ground tuples $\bar t$ out of the set $consts({\cal M}\Sigma)$. (Generating each such $\bar t$ gives rise to one iteration of the {\em main loop} of the algorithm.) For each such tuple $\bar t$, we then do the following: 
\begin{itemize} 

	\item Construct the query $Q({\bar t})$, as the result of applying to the query $Q$ the homomorphism\footnote{It is easy to verify that if a homomorphism $\mu$ specified by (i)-(ii) does not exist, then $\bar t$ cannot be a certain answer to $Q$ w.r.t. ${\cal M}\Sigma$.} $\mu$, such that (i) $\mu$ maps the head vector $\bar X$ of $Q$ to $\bar t$, and (ii) $\mu$ is the identity mapping on each term that occurs in $Q$ but not in its head vector $\bar X$; 

	\item Enumerate all the (up to an exponential number of) view-verified universal solutions $J$ for ${\cal M}\Sigma$ (recall that generating each such $J$ takes polynomial time in the size of $MV$, see Proposition~\ref{vv-sizes-prop}); and then 
	
	\item For each such $J$ that is not the empty instance, verify whether the query $Q({\bar t})$ has a nonempty set of answers, which would be precisely $\{ {\bar t} \}$, on the instance $J$. (For each $\bar t$ generated as above, we use a one-bit flag to track whether $\bar t$ is an answer to $Q$ on all such instances $J$; each $\bar t$ that is an answer to $Q$ on all the instances $J$ is returned as an answer tuple by the view-verified data-exchange algorithm.) The runtime for this verification step is polynomial in the size of $MV$ (because the size of $J$ is polynomial in the size of $MV$, see Proposition~\ref{vv-sizes-prop}) and is exponential in the number of subgoals of $Q$. (As the schema {\bf P} in ${\cal M}\Sigma$ is fixed, each subgoal of the query $Q$ has up to constant arity.)  

\end{itemize} 

Observe that for each tuple $\bar t$ generated in the main loop of the algorithm, the respective iteration of the main loop runs in {\sc PSPACE}. Indeed, recall from Proposition~\ref{vv-sizes-prop} that each instance $J$ as above is of size polynomial in the size of the instance $MV$ in ${\cal M}\Sigma$. Further, the size of each candidate valuation from $Q({\bar t})$ to $J$ is linear in the size of $Q$; thus, we satisfy the {\sc PSPACE} requirement as long as we generate these candidate valuations one at a time (``on the fly'' for each fixed $J$), in some clear algorithmic order. 

Further, the entire view-verified data-exchange algorithm (i.e., finding {\em all} the certain-answer tuples for the given input CQ weakly acyclic pair $({\cal M}\Sigma,Q)$) also runs in {\sc PSPACE}, provided that we: 

\begin{itemize} 
	\item[(a)] Output each certain-answer tuple $\bar t$ ``on the fly'' (i.e., as soon as we know that it is a certain answer), and 
	\item[(b)] Use a counter (e.g., a binary-number representation of each $k$-ary ground candidate certain-answer tuple $\bar t$, as generated in the main loop of the algorithm) to keep track of the ``latest'' $\bar t$ that we have looked at and to generate from that ``latest'' $\bar t$ the next candidate certain-answer tuple $\bar t$ that we are to examine for the given input; the size of such a counter would be polynomial in the size of the problem input. 
\end{itemize}

\section{The Number of Leaves in $MV$-Enh- anced Chase Can be Exponential in the Size of the Input} 
\label{expon-vv-number-app}

In this appendix we show by example a family of CQ
materialized-view settings ${\cal M}\Sigma$ with $\Sigma$ $=$ $\emptyset$, such that the number of leaves in a chase tree for each %member 
setting in the family is exponential in the size of the setting. %, even  the set $\Sigma$ of dependencies in each instance is the empty set. 
As usual and similarly to \cite{ZhangM05}, we assume that the {\em size} of a given materialized-view setting ${\cal M}\Sigma$ is the size of its instance $MV$, % and query $Q$, 
with the remaining elements of ${\cal M}\Sigma$ being fixed. (See Section~\ref{vv-correctness-sec} for a detailed discussion.) 

%In the next example we show that the number of view-verified universal solutions {\em can be} exponential in the size of the instance $MV$, even when the set $\Sigma$ of dependencies in ${\cal M}\Sigma$ is the empty set. (As usual, we assume that all the elements of ${\cal M}\Sigma$ except $MV$ are fixed.) 

\begin{example} 
\label{expon-num-vv-sols-ex} 
Consider a schema {\bf P} with two binary relations $P$ and $R$, and with $\Sigma$ $=$ $\emptyset$. Let the set of %two unary 
views $\cal V$ $=$ $\{ V, W \}$ be defined via two CQ queries, as follows: 

\begin{tabbing} 
$V(X) \leftarrow P(X,Y), R(Y,Z) .$ \\ 
$W(Z) \leftarrow R(Y,Z) .$ 
\end{tabbing} 

For each $n$ $\geq$ $1$, consider a set $MV_{(n)}$ of answers for $\cal V$, with $n+2$ tuples, as follows. The relation $MV_{(n)}[V]$ has $n$ tuples $V(1)$, $V(2)$, $\ldots,$ $V(n)$, and $MV_{(n)}[W]$ has tuples $W(0)$ and $W(1)$. 

%Let the secret CQ query $Q$ be 

%\begin{tabbing} 
%$Q(X,Z) \leftarrow P(X,Y), R(Y,Z) .$
%\end{tabbing} 

For each $n$ $\geq$ $1$, let the materialized-view setting ${\cal M}\Sigma^{(n)}$ be the tuple $(${\bf P}, $\Sigma$, $\cal V$, $MV_{(n)})$, with all the components as described above. (As specified above, the set $\Sigma$ is the empty set for each $n$ $\geq$ $1$.) 

The canonical universal solution $J^{{\cal M}\Sigma^{(n)}}_{de}$ for $MV_{(n)}$ has two tuples, $P(i,\perp_{(i,1)})$ and $R(\perp_{(i,1)},\perp_{(i,2)})$, for $V(i)$ in $MV_{(n)}$, for each $i$ $\in$ $[1,$ $n]$. It also has the tuples $R(\perp_{(n+1,1)},0)$ and $R(\perp_{(n+2,1)},1)$ for $MV_{(n)}[W]$. 

The process of creating view-verified universal solutions for ${\cal M}\Sigma^{(n)}$ involves assigning either $0$ or $1$ independently to each of the nulls $\perp_{(i,2)}$, for all $i$ $\in$ $[1,$ $n]$.  It is easy to see that this process creates $2^n$ nonisomorphic instances, one for each assignment of zeroes and ones %independently 
to each element of the vector $[\perp_{(1,2)},$ $\perp_{(2,2)},$ $\ldots$, $\perp_{(n,2)}]$. The expression $2^n$ is exponential in the size of the set $MV_{(n)}$ of view answers in ${\cal M}\Sigma^{(n)}$. 
\end{example} 

\section{Chase Cannot Be Staged for \\ Finding All Certain Answers}  
\label{chase-not-staged-for-view-verif-app} 

%{\bf [[[Done! NB! must finish typing (from handwritten) my examples 1 and 2 of $MV$-$\Sigma$-chase interleaving for the certain-query-answer problem. ]]]}  
	
	In this appendix we provide two examples that show that in the problem of finding all certain answers to a CQ query w.r.t. a CQ weakly acyclic materialized-view setting, one cannot always find all the certain answers correctly if one does the  chase (in view-verified data-exchange, see Appendix~\ref{vv-dexchg-sec}) {\em in stages.} That is, chase only with the input dependencies $\Sigma$, followed by chase only with the ``$MV$-induced dependencies,'' does not always yield a correct solution. (This is the point of Example~\ref{chase-interleave-certain-answ-one-ex}.) The reverse order of the ``stages'' does not always work either. (This is the point of Example~\ref{chase-interleave-certain-answ-two-ex}.) 
	
%	Here is a short summary of this appendix with a bit more detail. The point of Example~\ref{chase-interleave-certain-answ-one-ex} is that given a CQ weakly acyclic materialized-view setting ${\cal M}\Sigma$ and a CQ query $Q$, one cannot always obtain a correct set of certain answers to $Q$ w.r.t. ${\cal M}\Sigma$ via view-verified data exchange  (see Appendix~\ref{vv-dexchg-sec}) if one first does chase with the input dependencies $\Sigma$, followed by an application in chase of the ``$MV$-induced dependencies'' for ${\cal M}\Sigma$, while refraining from following up with more chase with $\Sigma$. (The ``$MV$-induced dependencies'' are constructed as is done in the view-verified data-exchange approach, see Appendix~\ref{vv-dexchg-sec} for the details.) Similarly, one cannot always obtain a correct set of certain answers to $Q$ w.r.t. ${\cal M}\Sigma$  via view-verified data exchange if one first does chase with  the $MV$-induced dependencies, followed by an application in chase of the input dependencies $\Sigma$, while refraining from following up with more chase with  the $MV$-induced dependencies. The latter observation is illustrated by Example~\ref{chase-interleave-certain-answ-two-ex}. 
	
\begin{example} 
\label{chase-interleave-certain-answ-one-ex} 
Consider a schema {\bf P} $=$ $\{ P, S \}$ with binary relation symbols $P$ and $S$. Let $\sigma$ be a dependency defined on the schema {\bf P}, as follows.  (The dependency $\sigma$ is an egd, specifically a functional dependency.)

\begin{tabbing} 
$\sigma:$ $P(X,Y) \wedge P(X,Z) \rightarrow Y = Z.$ 
\end{tabbing} 

Further, let $U$, $V$, and $W$ be three CQ views over {\bf P}, and let $MV$ be the set of answers for these views, as follows. 

\begin{tabbing} 
$U(X) \leftarrow P(X,Y), S(Y,Z).$ \\ 
$V(X) \leftarrow P(X,Y).$ \\ 
$W(X,Z) \leftarrow S(X,Y), P(Y,Z).$ \\ 
$MV$ $=$ $\{ U(c ), V(c ), W(g,h) \}$.  
\end{tabbing}

%\noindent 
We denote the set $\{ U, V, W \}$  by $\cal V$, and the set $\{ \sigma \}$ by $\Sigma$. Then the setting ${\cal M}\Sigma$ $=$ ({\bf P}, $\Sigma,$ $\cal V$, $MV)$ is CQ weakly acyclic. 

Now let $Q$ be a CQ query: 

\begin{tabbing} 
$Q(X) \leftarrow S(X,Y).$ 
\end{tabbing} 

We consider the problem of finding the set of certain answers to the query $Q$ w.r.t. the setting ${\cal M}\Sigma$ using the view-verified data-exchange approach, as described in Appendix~\ref{vv-dexchg-sec}. By this approach, we first construct, from the materialized-view setting ${\cal M}\Sigma$, a {\em data-exchange} setting ${\cal S}^{(de)}({\cal M}\Sigma)$ $=$ $({\cal V}$, {\bf P}, $\Sigma_{st} \cup \Sigma)$.   Here, $\Sigma_{st}$ $=$ $\{ \sigma_U, \sigma_V, \sigma_W \}$ is the set of the following three tgds: 

\begin{tabbing} 
hokoban teteu i\= keke momo\= l \kill 
$\sigma_U: U(X)$ \> $\rightarrow \exists Y,Z$ \> $P(X,Y) \wedge S(Y,Z)$. \\ 
$\sigma_V: V(X)$ \> $\rightarrow \exists Y$ \> $P(X,Y)$. \\ 
$\sigma_W: W(X,Z)$ \> $\rightarrow \exists Y$ \> $S(X,Y) \wedge P(Y,Z)$.  
\end{tabbing} 

We then designate the set of view answers $MV$ in the materialized-view setting ${\cal M}\Sigma$ to be a {\em source instance} for the data-exchange setting ${\cal S}^{(de)}({\cal M}\Sigma)$. 

We now proceed to construct the canonical universal solution, call it $J_0$, for the source instance $MV$ in the data-exchange setting ${\cal S}^{(de)}({\cal M}\Sigma)$:  

\begin{tabbing} 
hokob   \= teteu i\= keke momo\= l \kill 
$J_0$ $=$ $\{ P(c,\perp_1), S(\perp_1,\perp_2), P(c,\perp_3),$ \\  
\> $S(g,\perp_4), P(\perp_4,h) \}$. 
\end{tabbing} 

%\noindent 
In the instance $J_0$, the atoms $P(c,\perp_1)$ and $S(\perp_1,$ $\perp_2)$ are due to the atom $U(c )$ in $MV$ and to the tgd $\sigma_U$, and so on for the rest of $MV$ and of $\Sigma_{st}$. 

As described in Appendix~\ref{vv-dexchg-sec}, toward finding all the certain answers to the query $Q$ w.r.t. the materialized-view setting ${\cal M}\Sigma$,  we now chase the instance $J_0$, using both the set $\Sigma$ $=$ $\{ \sigma \}$ in ${\cal M}\Sigma$, as well as the dependencies $\tau_{U(c )}$, $\tau_{V(c )}$, and $\tau_{W(g,h)}$, as follows. (The three latter dependencies are generated  by the view-verified data-exchange approach from the inputs $\cal V$ and $MV$.) 

\begin{tabbing} 
hokoban teteu i\= keke momo\= l \kill 
$\tau_{U(c )}:$ $P(X,Y) \wedge S(Y,Z)$ $\rightarrow$ $X = c$. \\ 
$\tau_{V(c )}:$ $P(X,Y)$ $\rightarrow$ $X = c$. \\ 
$\tau_{W(g,h)}:$ $S(X,Y) \wedge P(Y,Z)$ $\rightarrow$ $X = g$ $\wedge$ $Z = h$. 
\end{tabbing}

We do three stages of the chase of the instance $J_0$ with $\sigma$, $\tau_{U(c )}$, $\tau_{V(c )}$, and $\tau_{W(g,h)}$. In Stages I and III, we perform the chase steps with the input dependency $\sigma$ on the schema {\bf P}, and in Stage II, we chase the instance with the $MV$-induced dependencies $\tau_{U(c )}$, $\tau_{V(c )}$, and $\tau_{W(g,h)}$. 

{\em Stage I:} A chase step of the instance $J_0$ with the egd  $\sigma$ turns the atom $P(c,\perp_3)$ of $J_0$ into a copy of its atom $P(c,\perp_1)$, resulting in the following instance $J_1$ (in which we drop the duplicate of $P(c,\perp_1)$): 

\begin{tabbing} 
hokob   \= teteu i\= keke momo\= l \kill 
$J_1$ $=$ $\{ P(c,\perp_1), S(\perp_1,\perp_2), S(g,\perp_4), P(\perp_4,h) \}$. 
\end{tabbing} 

\noindent 
The egd $\sigma$ does not apply to the instance $J_1$. 

{\em Stage II:} We now chase the instance $J_1$ with the $MV$-induced dependencies $\tau_{U(c )}$, $\tau_{V(c )}$, and $\tau_{W(g,h)}$. The egd $\tau_{V(c )}$ applies to the atom $P(\perp_4,h)$ in $J_1$, turning it into $P(c,h)$ and, as a side effect, also turning the atom $S(g,\perp_4)$ of $J_1$ into $S(g,c)$. We call the resulting instance $J_2$: 

\begin{tabbing} 
hokob   \= teteu i\= keke momo\= l \kill 
$J_2$ $=$ $\{ P(c,\perp_1), S(\perp_1,\perp_2), S(g,c), P(c,h) \}$. 
\end{tabbing} 

\noindent 
The $MV$-induced dependencies $\tau_{U(c )}$, $\tau_{V(c )}$, and $\tau_{W(g,h)}$ do not apply to the instance $J_2$. 

Note that if we stop after this Stage II, the set of answers without nulls to the query $Q$ on the instance $J_2$ is $Q(J_2)$ $=$ $\{ (g) \}$. However, we observe that the instance $J_2$ does not satisfy the egd $\sigma$. We can then do {\em Stage III} of the chase, by applying $\sigma$ to the instance $J_2$. The application binds the null $\perp_1$ to the constant $h$, in the atoms $P(c,\perp_1)$ and $S(\perp_1,\perp_2)$ of the instance $J_2$. We call the resulting instance $J_3$: 

\begin{tabbing} 
hokob   \= teteu i\= keke momo\= l \kill 
$J_3$ $=$ $\{ P(c,h), S(h,\perp_2), S(g,c) \}$. 
\end{tabbing} 

The set of answers  without nulls  to the query $Q$ on the instance $J_3$ is $Q(J_3)$ $=$ $\{ (g), (h) \}$. By the results reported in Appendix~\ref{vv-dexchg-sec}, this set is a correct set of certain answers to $Q$ w.r.t. the given materialized-view setting ${\cal M}\Sigma$. We can see that by not applying Stage III in the chase, we would have missed the certain answer $(h)$ to the query $Q$ w.r.t. the setting ${\cal M}\Sigma$. 
\end{example} 
	
\begin{example} 
\label{chase-interleave-certain-answ-two-ex} 
Consider a schema {\bf P} $=$ $\{ P, S \}$ with a unary relation symbol $P$ and a binary relation symbol $S$. Let $\sigma$ be a dependency (specifically, a full tgd)  defined on the schema {\bf P}, as follows.  

\begin{tabbing} 
$\sigma:$ $S(X,Y)  \rightarrow P(Y).$ 
\end{tabbing} 

Further, let $U$ and $V$ be two CQ views over {\bf P}, and let $MV$ be the set of answers for these views, as follows. 

\begin{tabbing} 
$U(X) \leftarrow P(X), S(X,Y).$ \\ 
$V(X) \leftarrow P(X).$ \\ 
$MV$ $=$ $\{ U(c ), V(c ) \}$.  
\end{tabbing}

\noindent 
We denote the set $\{ U, V \}$  by $\cal V$, and the set $\{ \sigma \}$ by $\Sigma$. Then the setting ${\cal M}\Sigma$ $=$ ({\bf P}, $\Sigma,$ $\cal V$, $MV)$ is CQ weakly acyclic. 

Now let $Q$ be a CQ query: 

\begin{tabbing} 
$Q(X) \leftarrow S(X,X).$ 
\end{tabbing} 

We consider the problem of finding the set of certain answers to the query $Q$ w.r.t. the setting ${\cal M}\Sigma$, using the view-verified data-exchange approach detailed in Appendix~\ref{vv-dexchg-sec}. By this approach, we first construct, from the materialized-view setting ${\cal M}\Sigma$, a {\em data-exchange} setting ${\cal S}^{(de)}({\cal M}\Sigma)$ $=$ $({\cal V}$, {\bf P}, $\Sigma_{st} \cup \Sigma)$.   Here, $\Sigma_{st}$ $=$ $\{ \sigma_U, \sigma_V \}$ is the set of the following two tgds: 

\begin{tabbing} 
hokoban tet\= keke m \= momo\= l \kill 
$\sigma_U: U(X)$ \> $\rightarrow \exists Y$ \> $P(X) \wedge S(X,Y)$. \\ 
$\sigma_V: V(X)$ \> $\rightarrow P(X)$. 
\end{tabbing} 

We then designate the set of view answers $MV$ in the materialized-view setting ${\cal M}\Sigma$ to be a {\em source instance} for the data-exchange setting ${\cal S}^{(de)}({\cal M}\Sigma)$. 

We proceed to construct the canonical universal solution, call it $J_0$, for the source instance $MV$ in the data-exchange setting ${\cal S}^{(de)}({\cal M}\Sigma)$. 

\begin{tabbing} 
hokob   \= teteu i\= keke momo\= l \kill 
$J_0$ $=$ $\{ P(c), S(c,\perp_1) \}$. 
\end{tabbing} 

%\noindent 
In the instance $J_0$, the atoms $P(c)$ and $S(c,\perp_1)$ are due to the atom $U(c )$ in $MV$ and to the tgd $\sigma_U$. At the same time, the atom $P(c)$ is also due to the atom $V(c )$ in $MV$ and to the tgd $\sigma_V$. 

As described in Appendix~\ref{vv-dexchg-sec}, toward finding all the certain answers to the query $Q$ w.r.t. the materialized-view setting ${\cal M}\Sigma$,  we now chase the instance $J_0$, using both the set $\Sigma$ $=$ $\{ \sigma \}$ in ${\cal M}\Sigma$, as well as the dependencies $\tau_{U(c )}$ and $\tau_{V(c )}$, as follows. (The two latter dependencies are generated  by the view-verified data-exchange approach from the inputs $\cal V$ and $MV$.) 

\begin{tabbing} 
hokoban teteu i\= keke momo\= l \kill 
$\tau_{U(c )}:$ $P(X) \wedge S(X,Y)$ $\rightarrow$ $X = c$. \\ 
$\tau_{V(c )}:$ $P(X)$ $\rightarrow$ $X = c$. 
\end{tabbing}

We do three stages of the chase of the instance $J_0$ with $\sigma$, $\tau_{U(c )}$, and $\tau_{V(c )}$. In Stages I and III, we perform the chase steps with the $MV$-induced dependencies $\tau_{U(c )}$ and $\tau_{V(c )}$,   and in Stage II, we chase the instance with the input dependency $\sigma$ on the schema {\bf P}. 

{\em Stage I:} A chase step of the instance $J_0$ with %the egd  $\sigma$ 
the $MV$-induced dependencies $\tau_{U(c )}$ and $\tau_{V(c )}$ leaves the instance $J_0$ unchanged. To indicate that we have performed this stage of the chase, we rename $J_0$ into $J_1$: %turns the atom $P(c,\perp_3)$ of $J_0$ into a copy of its atom $P(c,\perp_1)$, resulting in the following instance $J_1$ (in which we drop the duplicate of $P(c,\perp_1)$): 

\begin{tabbing} 
hokob   \= teteu i\= keke momo\= l \kill 
$J_1$ $=$ $\{ P(c), S(c,\perp_1) \}$. 
\end{tabbing} 

\noindent 
The $MV$-induced dependencies $\tau_{U(c )}$ and $\tau_{V(c )}$ do not apply to the instance $J_1$. 

{\em Stage II:} We now chase the instance $J_1$ with the tgd  $\sigma$ on the schema {\bf P}. The application adds to the instance $J_1$ the atom  $P(\perp_1)$. We call the resulting instance $J_2$: 

\begin{tabbing} 
hokob   \= teteu i\= keke momo\= l \kill 
$J_2$ $=$ $\{ P(c), S(c,\perp_1), P(\perp_1) \}$. 
\end{tabbing} 

\noindent 
The tgd  $\sigma$ does not apply to the instance $J_2$. 

Note that if we stop after this Stage II, the set of answers without nulls to the query $Q$ on the instance $J_2$ is $Q(J_2)$ $=$ $\emptyset$. However, we observe that the instance $J_2$ does not satisfy the $MV$-induced dependency $\tau_{V(c )}$. We can then do {\em Stage III} of the chase, by applying the dependencies $\tau_{U(c )}$ and $\tau_{V(c )}$ to the instance $J_2$. An application of $\tau_{V(c )}$ in a chase step to $J_2$ binds the null $\perp_1$, in the atoms $P(\perp_1)$ and $S(c,\perp_1)$, to the constant $c$. We call the resulting instance $J_3$: 

\begin{tabbing} 
hokob   \= teteu i\= keke momo\= l \kill 
$J_3$ $=$ $\{ P(c), S(c,c) \}$. 
\end{tabbing} 

The set of answers  without nulls  to the query $Q$ on the instance $J_3$ is $Q(J_3)$ $=$ $\{ (c ) \}$. By the results reported in Appendix~\ref{vv-dexchg-sec}, this set is a correct set of certain answers to $Q$ w.r.t. the given materialized-view setting ${\cal M}\Sigma$. We can see that by not applying Stage III in the chase, we would have missed the certain answer $(c)$ to the query $Q$ w.r.t. ${\cal M}\Sigma$. 
\end{example}

\section{The Certain-Query-Answer Problem Is $\Pi^p_2$ Complete for Conjunctive Weakly Acyclic Input Inst- ances} 
\label{pi-p-2-hardness-proof-sec} 

In this appendix %Finally, %or the runtime complexity of the %view-verified data-exchange 
%approach when given as inputs %CQ weakly acyclic 
%instances ${\cal M}\Sigma$ in this class, 
we prove that the %view-verified data-exchange 
certain-query-answer problem, for a query and ground tuple w.r.t. a materialized-view-setting,   
 is $\Pi^p_2$ 
%approach is $\Pi^p_2$ 
complete for CQ weakly acyclic input instances $({\cal M}\Sigma,Q,{\bar t})$. We say that a triple $({\cal M}\Sigma,Q,{\bar t})$, with ${\cal M}\Sigma$ a materialized-view setting, $Q$ a query,  and $\bar t$ a ground tuple, is a {\em CQ weakly acyclic input instance} if and only if ${\cal M}\Sigma$ is CQ weakly acyclic and $Q$ is a CQ query.  

In the complexity measure used throughout this appendix, we assume, in a natural extension of the complexity setting introduced in \cite{ZhangM05} (see Section~\ref{vv-correctness-sec} for the detailed discussion), that all elements of ${\cal M}\Sigma$ except $MV$ are fixed, and that $Q$ is not fixed. That is, the {\em size} of a given input instance $({\cal M}\Sigma,Q,{\bar t})$ is the size of its set of view answers $MV$ and of its query $Q$, with the remaining elements of $({\cal M}\Sigma,Q,{\bar t})$ being fixed. Note that in all input instances 
$({\cal M}\Sigma,Q,{\bar t})$ in which $\bar t$ could be a certain answer to $Q$ w.r.t. ${\cal M}\Sigma$, 
 the size of the ground input tuple $\bar t$ must be linear in the size of $Q$; more precisely, the size of $\bar t$ must be the arity of the query $Q$. Thus, %in the sequel 
in this appendix we restrict our consideration to the problem-input triples that satisfy this property. That is, in all of the results in this appendix we assume that, in all the given input instances $({\cal M}\Sigma,Q,{\bar t})$, we have that: $Q$ is a $k$-ary CQ query for some $k$ $\geq$ $0$; $\bar t$ is a $k$-ary ground tuple; and the size of the CQ weakly acyclic instance $({\cal M}\Sigma,Q,{\bar t})$ is the size of the set of answers $MV$ in ${\cal M}\Sigma$ and of the query $Q$.

%{\bf Complexity of the certain-query-answer problem for CQ queries and CQ weakly acyclic mate- rialized-view-settings.} %Finally, we use our view-verified data exchange approach to
%We now show that the certain-query-answer problem is ${\Pi}^p_2$ complete for CQ weakly acyclic input instances. Here, we assume, as in \cite{ZhangM05}, that all the elements of ${\cal M}\Sigma$ except $MV$ are fixed, and that $Q$ is not fixed. 

We first observe that the problem is in ${\Pi}^p_2$. 

%\vspace{-0.1cm} 

\begin{proposition} 
\label{vv-sizes-corol} 
The certain-answer problem for a query and a %candidate certain-answer 
tuple w.r.t. a materialized-view setting 
is in ${\Pi}^p_2$ for CQ weakly acyclic input instances. 
%Given a CQ weakly acyclic instance ${\cal M}\Sigma$, the problem of information-leak disclosure for ${\cal M}\Sigma$ %using the view-verified data-exchange approach 
%is in ${\Pi}^p_2$. 
\end{proposition} 

%\vspace{-0.1cm} 

\begin{proof} 
Given a CQ weakly acyclic input instance $({\cal M}\Sigma,Q,{\bar t})$, 
we show how to ascertain that the ground tuple $\bar t$ is {\em not} a certain answer to the query $Q$ w.r.t. the setting ${\cal M}\Sigma$. Observe first that if this is the case, then, by soundness of view-verified data exchange (see Section~\ref{vv-correctness-sec}), there must be a view-verified universal solution, $J$, for ${\cal M}\Sigma$ such that $\bar t$ is not an answer to $Q$ on the instance $J$. We can thus: 
\begin{itemize} 
	\item[(1)] Guess a view-verified universal solution $J$ for ${\cal M}\Sigma$, and then  % (this takes time and space polynomial in the size of the instance $MV$ in ${\cal M}\Sigma$), and 
	
	% such that $\bar t$  
	
%	in polynomial time 

	\item[(2)] Verify that there is no valuation from the query $Q({\bar t})$ to $J$; here, $Q({\bar t})$ is the result of applying to the query $Q$ the homomorphism\footnote{It is easy to verify that if a homomorphism $\mu$ specified by (i)-(ii) does not exist, then $\bar t$ cannot be a certain answer to $Q$ w.r.t. ${\cal M}\Sigma$.} $\mu$, such that: 
	\begin{itemize} 
		\item[(i)] $\mu$ maps the head vector $\bar X$ of $Q$ to $\bar t$, and 
		\item[(ii)] $\mu$ is the identity mapping on each term that occurs in $Q$ but not in its head vector $\bar X$. 
	\end{itemize} 

\end{itemize} 

By Proposition~\ref{vv-sizes-prop}, the step (1) that generates the instance $J$ can be done in polynomial space in the size of the set $MV$ in ${\cal M}\Sigma$. Further, step (2) can be done using an $NP$-oracle for $Q$ and $J$, as the size of each valuation from $Q({\bar t})$ to $J$ must be polynomial in the size of $Q$ and $J$ (it is, in fact, linear in the size of $Q$). 
\end{proof} 

We now provide a ${\Pi}^p_2$ hardness result, even for the special case of CQ weakly acyclic inputs with $\Sigma$ $=$ $\emptyset$.

%In this appendix we prove Theorem~\ref{complexity-hardness-thm}, which states that the problem of information-leak disclosure with  CQ inputs is ${\Pi}^p_2$ hard, even when the set $\Sigma$ of dependencies in the input instance ${\cal M}\Sigma$ is the empty set. % when using the view-verified data-exchange approach of Section~\ref{vv-dexchg-sec}. 

\begin{theorem} 
\label{complexity-hardness-thm} 
The certain-answer problem for a query and a %candidate certain-answer 
tuple w.r.t. a materialized-view setting 
is ${\Pi}^p_2$ hard for CQ input instances $({\cal M}\Sigma,Q,{\bar t})$ in which $\Sigma$ $=$ $\emptyset$ in the input materialized-view setting ${\cal M}\Sigma$. 
%is ${\Pi}^p_2$ hard for CQ inputs with $\Sigma$ $=$ $\emptyset$. 
\end{theorem} 

%\vspace{-0.1cm} 

Before providing a proof of Theorem~\ref{complexity-hardness-thm}, we observe that as 
an immediate corollary of Theorem~\ref{complexity-hardness-thm} and of Proposition~\ref{vv-sizes-corol}
we obtain the main result of this appendix, a ${\Pi}^p_2$-completeness result for the certain-answer problem for a query and a %candidate certain-answer 
tuple w.r.t. a materialized-view setting, for the case of 
CQ weakly acyclic input instances: 

%\vspace{-0.1cm} 

\begin{theorem} 
\label{complexity-hardness-and-membership-thm} 
The certain-answer problem for a query and a %candidate certain-answer 
tuple w.r.t. a materialized-view setting 
is ${\Pi}^p_2$ complete for CQ weakly acyclic input instances. % $({\cal M}\Sigma,Q,{\bar t})$. 
\end{theorem} 

%\vspace{-0.1cm} 

In the remainder of this appendix, we provide a proof of Theorem~\ref{complexity-hardness-thm}. 
As a summary, the result of Theorem~\ref{complexity-hardness-thm} is by reduction from the $\forall$$\exists$-$CNF$ problem, which is known to be ${\Pi}^p_2$ complete \cite{Stockmeyer76}. %\footnote{Note that we cannot infer the ${\Pi}^p_2$ hardness of the problem just from the fact that our rewriting-based approach of Section~\ref{rewriting-sec}, which is sound and complete  for the case of $\Sigma$ $=$ $\emptyset$, uses a ${\Pi}^p_2$-complete containment test of \cite{ZhangM05}.} 
We start off from the reduction that was used by Millstein and colleagues in \cite{MillsteinHF03} for the problem of query containment for data-integration systems. We modify the reduction of \cite{MillsteinHF03} in the spirit that is similar to the modification of that reduction (of \cite{MillsteinHF03}) 
as suggested in \cite{ZhangM05}. (Recall that the full version of \cite{ZhangM05}, including any of its proofs, has never been published.) The goal of our modification is to comply with our assumptions about the input size, specifically with the assumption that the input view definitions are fixed. (In \cite{MillsteinHF03} it is assumed that both the queries and the view definitions can vary.) %The proof of Theorem~\ref{complexity-hardness-thm} can be found in Appendix~\ref{pi-p-2-hardness-proof-sec}.  

%{\bf [[[ Started here T07/23/13 ]]]} 

%{\bf ******************** [[[ Everything below, all until the end of the document, is old --- as of December 2012 ]]]} 

\begin{proof}{(Theorem~\ref{complexity-hardness-thm})} 
In this proof, we build on the constructions from the proof of Theorem 3.3 in \cite{MillsteinHF03}; that result of \cite{MillsteinHF03} states ${\Pi}^p_2$ hardness for a subclass of the problem of query containment for data-integration systems. The reason that we modify the reduction of \cite{MillsteinHF03} is that we need to comply with our assumptions about the size of our input instances ${\cal M}\Sigma$, specifically with the assumption that the input view definitions are fixed. (In \cite{MillsteinHF03} it is assumed that both the queries and the view definitions can vary.) Thus, our variation on the reduction of \cite{MillsteinHF03} is similar in spirit to the modification suggested in \cite{ZhangM05}. 

Similarly to the reduction in \cite{MillsteinHF03}, we reduce the $\forall$$\exists$-$CNF$ problem, known to be ${\Pi}^p_2$ complete \cite{Stockmeyer76}, to our problem. The $\forall$$\exists$-$CNF$ problem is defined as follows: Given a 3-$CNF$ propositional formula $F$ with variables $\bar X$ and $\bar Y$, is it the case that for each truth assignment to $\bar Y$, there exists a truth assignment to $\bar X$ that satisfies $F$? Here, we denote by $\bar X$ the set of $n$ variables $X_1$, $\ldots$, $X_n$, for some $n$ $\geq$ $0$, and we denote by $\bar Y$ the set of $m$ variables $Y_1$, $\ldots$, $Y_m$, for some $m$ $\geq$ $1$. 

The reduction is as follows. Suppose we are given a 3-$CNF$ formula $F$, with variables 

\begin{tabbing} 
$\bar Z$ $=$ $\{ X_1,\ldots,X_n \}$ $\cup$ $\{ Y_1,\ldots,Y_m \}$.  
\end{tabbing} 

\noindent 
The formula $F$ has clauses $\bar C$ $=$ $\{ C_1,\ldots,C_l \}$. Clause $C_i$ contains the three variables (either positive or negated) $Z_{i,1}$, $Z_{i,2}$, and $Z_{i,3}$; each of the three variables is an element of the set $\bar Z$. 
%&&
%For ease of exposition, in the main body of this proof we will assume that in the instance of the $\forall$$\exists$-$CNF$ problem under consideration, we have $m$ $\geq$ $1$. (In the last paragraph of the proof, we outline how to modify the main body of the proof  so that it would also go through for each input instance with $m$ $=$ $0$.) %First, we replace $MV[W]$ by the relation $MV[W]$ $=$ $\{$ $(0)$ $\}$. Second, we replace the conjunction of the $P$-subgoals in the query $Q$ by a single subgoal $P(Y_0,0)$. Third, we replace $MV[U]$ by $MV[U]$ $=$ $\{$ $U(1,0)$, $U(0,0)$ $\}$. Fourth, we use $I^{(1)}$ and $I^{(0)}$ and as the only two ``core'' instances $I$. Here,  $I^{(1)}$ has $P$ $=$ $\{ P(1,0) \}$, and $I^{(0)}$ has $P$ $=$ $\{ P(0,0) \}$. The other relations ($R$ and $S$) in each of $I^{(1)}$ and $I^{(0)}$ mirror $R$ and $S$ in the core relations for $m$ $\geq$ $1$. Finally, for the first argument $Y_0$ of the subgoal $P(Y_0,0)$ of the query $Q$ we have (unlike the case $m$ $\geq$ $1$) that $Y_0$ is {\em not} in the conjunction of the $R$-subgoals of the query $Q$.  %Check that the rest of the proof goes through for this case $m$ $=$ $0$. 

%As just discussed, for the remainder of this proof, we assume that in the input instance of the $\forall$$\exists$-$CNF$ problem under consideration, we have $m$ $\geq$ $1$.  
For the input formula $F$, %with $m$ $\geq$ $1$, 
we begin building the corresponding (CQ weakly acyclic) instance\linebreak $({\cal M}\Sigma,Q,{\bar t})$ of the certain-answer problem for a query and a %candidate certain-answer 
tuple w.r.t. a materialized-view setting. In each such instance $({\cal M}\Sigma,Q,{\bar t})$, the query $Q$ will be Boolean, hence $\bar t$ will be the empty tuple. The setting ${\cal M}\Sigma$ that we will construct is as usual a quadruple of the form  $(${\bf P}, $\Sigma$, $\cal V$, $MV)$, always with $\Sigma$ $=$ $\emptyset$. We show below how to construct each of {\bf P}, $\cal V$, $MV$, and $Q$ for the input formula $F$. 

 The schema {\bf P} $=$ $\{ P, R, S \}$ that we construct for the input formula $F$ uses three relation symbols: $R$ of arity $k_R$ $=$ $4$, and two binary relation symbols $P$ and $S$. Intuitively, for each $i$ $\in$ $[1,$ $l]$ and for the clause $C_i$ in $F$, in each ``relevant'' instance of schema {\bf P} we will have in $R$ a nonempty set of tuples whose fourth argument is the constant $i$. Further, for each $j$ $\in$ $[1,$ $m]$ and for the variable $Y_j$ in $F$, in each ``relevant'' instance of schema {\bf P} we will have in each of $P$ and $S$ a nonempty set of tuples whose second argument is the constant $j$. 

%The set $\Sigma$ in the instance ${\cal D}_s(F)$ is the empty set: $\Sigma$ $=$ $\emptyset$. 

We now define the set of views $\cal V$ $=$ $\{ U, V, W \}$ in the materialized-view setting ${\cal M}\Sigma$ that we construct for the given formula $F$. First, for the clauses $\bar C$ in $F$ we introduce the following view $V$: 

\begin{tabbing} 
$V(Z_1,Z_2,Z_3,i)$ $\leftarrow$ $R(Z_1,Z_2,Z_3,i).$
\end{tabbing} 

\noindent 
(Here, $i$ is a variable rather than a constant; we use the variable name $i$ in the definition of the view $V$ to mnemonically refer to each clause $C_i$ in $\bar C$ as discussed above.) The answer to this view simply mirrors the relation $R$. 
%For each $i$ $\in$ $[1,$ $l]$ and for the clause $c_i$ in $F$, the tuples with the value $i$ of the fourth argument in the relation for the view $V$ record which variables are in the clause $c_i$. 

In the set $MV$ of answers to the views in $\cal V$ that we are constructing for the given formula $F$, the relation $MV[V]$  records, for each $i$ $\in$ $[1,$ $l]$ and for the clause $C_i$ in $F$, the seven (out of the total eight possible) satisfying assignments for the clause. (We follow \cite{MillsteinHF03} in using $1$ for $true$ and $0$ for $false$.)  The fourth argument of each tuple in $MV[V]$ for these seven assignments for $C_i$ is always $i$. 

As a running example, we use the following example from the proof of Theorem 3.3 in \cite{MillsteinHF03}: Consider the formula 

\begin{tabbing} 
$F = (X_1 \vee X_2 \vee Y_1) \wedge (\neg X_1 \vee \neg X_2 \vee \neg Y_2).$
\end{tabbing} 

\noindent 
The seven satisfying assignments to $X_1$, $X_2$, and $Y_1$ in the first clause $C_1 = (X_1 \vee X_2 \vee Y_1)$ of $F$ are $(1,1,1)$, $(1,1,0)$, $(1,0,1)$, $(1,0,0)$, $(0,1,1)$, $(0,1,0)$, and $(0,0,1)$. For the second clause $C_2 = (\neg X_1 \vee \neg X_2 \vee \neg Y_2)$ of $F$, the seven satisfying assignments to $X_1$, $X_2$, and $Y_2$ are $(0,0,0)$, $(0,0,1)$, $(0,1,0)$, $(0,1,1)$, $(1,0,0)$, $(1,0,1)$, and $(1,1,0)$. 

By construction of the view $V$ and by our intuition for the relation $R$,  see above, in this running example we construct the relation $MV[V]$ from these fourteen assignments, as follows. First, the seven assignments as above for $C_1$ are adorned, in the fourth argument of $V$, by the index $1$ of $C_1$, as follows: $V(1,1,1,1)$, $V(1,1,0,1)$, $V(1,0,1,1)$, $V(1,0,0,1)$, $V(0,1,1,1)$,\linebreak $V(0,1,0,1)$, and $V(0,0,1,1)$. Similarly, the seven assignments as above for $C_2$ get  adorned, in the fourth argument of $V$, by the index $2$ of $C_2$, as follows: $V(0,0,0,2)$, $V(0,0,1,2)$, $V(0,1,0,2)$, $V(0,1,1,2)$, $V(1,0,0,2)$,\linebreak $V(1,0,1,2)$, and $V(1,1,0,2)$. These fourteen tuples together constitute the relation $MV[V]$ for the formula $F$ in this running example. 

We now return from our running example, to continue to define the views in the set $\cal V$ $=$ $\{ U, V, W \}$ for the formula $F$. For the $m$ $\geq$ $1$ variables $Y_1,\ldots,Y_m$ in $F$, we introduce a unary view $W$: 

\begin{tabbing} 
Tab heh \= boom \kill
$W(j)$ $\leftarrow$ $P(Y_j,j), S(Y_j,j).$ 
\end{tabbing} 

\noindent 
Here, $j$ is a variable rather than a constant; we use the variable name $j$ in the definition of the view $W$ to mnemonically refer to each variable $Y_j$ in the formula $F$ as discussed in the beginning of this proof. 

%\noindent 
%The relation $MV[W]$ 
The relation $MV[W]$ for the given formula $F$ is $MV[W]$ $=$ $\{ (1), (2),$ $\ldots,$ $(m) \}$. Intuitively, for each $j$ $\in$ $[1$, $m]$, the tuple $(j)$ in $MV[W]$ witnesses, in each ground instance $I$ of schema {\bf P} such that ${\cal V}$ $\Rightarrow_{I,\Sigma}$ $MV$, the presence of a ground $P$-atom $P(y_j,j)$ and of a ground $S$-atom $S(y_j,j)$ with the same arguments. (Here, $y_j$ is some constant value.) In our running example, we have that $MV[W]$ $=$ $\{ (1), (2) \}$, with one tuple for each of the two variables $Y_1$ and $Y_2$ in the given formula $F$. 

Finally, in the general case we define the view $U$ in the set $\cal V$ $=$ $\{ U, V, W \}$ for the given formula $F$ as follows: 

\begin{tabbing} 
$U(Y,j)$ $\leftarrow$ $S(Y,j).$
\end{tabbing} 

\noindent 
(As in the previous view definitions, $j$ is a variable rather than a constant.) The answer to this view simply mirrors the relation $S$. 

Now in the set $MV$ that we are constructing, the relation $MV[U]$ for $U$ provides the two possible truth assignments, $1$ and $0$, to each variable among $Y_1$, $\ldots,$ $Y_m$ in the set of variables $\bar Y$ in the formula $F$. That is, for each $j$ $\in$ $[1,$ $m]$, the relation $MV[U]$ has exactly two tuples, $U(1,j)$ and $U(0,j)$. For instance, the relation $MV[U]$ for our running example would have four tuples: $MV[U]$ $=$ $\{ U(1,1), U(0,1), U(1,2), U(0,2) \}$. Here, the first two  tuples correspond to the two possible truth assignments, $1$ and $0$, to the variable $Y_1$ in the formula $F$ in the example. Similarly, the last two  tuples in $MV[U]$ correspond to the two possible truth  assignments, $1$ and $0$, to the variable $Y_2$ in the formula $F$ in the example. 

For the general case of the formula $F$, % (with $m$ $\geq$ $1$), 
the above construction generates, for a given $F$, both the ground tuple $\bar t$ $=$ $()$ and the elements {\bf P}, $\Sigma$ ($=$ $\emptyset$), $\cal V$, and $MV$ in the materialized-view setting ${\cal M}\Sigma$ that we are producing for $F$. To complete the construction of the input instance $({\cal M}\Sigma,Q,{\bar t})$ for the given $F$, we now specify the CQ query $Q$: 

\begin{tabbing} 
Tab me\= me crazy \kill
$Q() \leftarrow \bigwedge_{j=1}^m P(Y_j,j) \ \bigwedge_{i=1}^l R(Z_{i,1},Z_{i,2},Z_{i,3},i).$
\end{tabbing} 

%\noindent 
%As in the view definitions given earlier in this proof, this definition of the query $Q$ does not involve any constants. 

Intuitively, the Boolean query $Q$ has a separate subgoal, $P(Y_j,j)$, for each $j$ $\in$ $[1,$ $m]$, that is for each $Y_j$ among the variables $Y_1$, $\ldots$, $Y_m$ of the input formula $F$. The query $Q$ also has a separate subgoal, $R(Z_{i,1},Z_{i,2},Z_{i,3},i)$, for each $i$ $\in$ $[1,$ $l]$, that is for each $C_i(Z_{i,1},Z_{i,2},Z_{i,3})$ among the clauses  $C_1$, $\ldots$, $C_l$ of the input formula $F$. By design, $Q$ uses in each of its $R$-subgoals all the variables of the form $Z_{i,k}$ in the same way as they are used in the corresponding clause $C_i$ in the formula $F$. (Recall that for each variable of the form $Z_{i,k}$ in the clauses of $F$, this variable is either among the variables $\bar X$ of $F$ or among the variables $\bar Y$ of $F$.) In addition, also by design of the query $Q$, for each variable $Y_j$ among $Y_1$, $\ldots$, $Y_m$ in the formula $F$, the {\em same} variable name $Y_j$ %(albeit in its uppercase form) 
is used in the $P$-subgoal of $Q$ with $j$ the value of the second argument of the subgoal. It follows that for each variable that is used as the first argument of some $P$-subgoal of the query $Q$, this same variable must occur in the conjunction $\bigwedge_{i=1}^l R(Z_{i,1},Z_{i,2},Z_{i,3},i)$ in the body of $Q$. 

As an illustration, the query $Q$ for the formula $F$ in our running example is as follows: 

\begin{tabbing} 
$Q() \leftarrow P(Y_1,1), P(Y_2,2), R(X_1, X_2, Y_1, 1), R(X_1, X_2, Y_2, 2).$
\end{tabbing} 

We have completed the construction of the  instance $({\cal M}\Sigma,Q,{\bar t})$ for each 3-$CNF$ propositional formula $F$. % in which $m$ $\geq$ $1$. 
By construction, the instance $({\cal M}\Sigma,Q,{\bar t})$ is CQ weakly acyclic and has $\Sigma$ $=$ $\emptyset$ and $\bar t$ $=$ $()$. Further,  each of {\bf P}, $\Sigma$, and $\cal V$, in the materialized-view setting ${\cal M}\Sigma$ constructed for the input formula $F$, does not depend on $F$; thus, both the formulation and the size of each of {\bf P}, $\Sigma$, and $\cal V$ are fixed across all the input formulae $F$. In contrast, the size of each of $Q$ and $MV$ in the instance $({\cal M}\Sigma,Q,{\bar t})$ is linear in the size of the input formula $F$. As a result, the overall size of the instance $({\cal M}\Sigma,Q,{\bar t})$ is polynomial (linear) in the size of the input formula $F$. 

It turns out that for each formula $F$, % (with $m$ $\geq$ $1$), 
the materialized-view setting ${\cal M}\Sigma$ that we construct for $F$
 is a valid setting by definition. Indeed, for each such $F$ and for the corresponding setting ${\cal M}\Sigma$ $=$ $(${\bf P}, $\Sigma$, $\cal V$, $MV)$, there exists a ground instance $I^{(1,1,\ldots,1)}(F)$ of schema {\bf P} such that $\cal V$ $\Rightarrow_{I^{(1,1,\ldots,1)}(F),\Sigma}$ $MV$, as follows. (Intuitively, the $m$-tuple $(1,1,\ldots,1)$ in the name of this instance $I^{(1,1,\ldots,1)}(F)$ refers to the fact that this instance represents an assignment of the truth value $1$ to each of the $m$ variables $Y_1$, $\ldots$, $Y_m$ of the input formula $F$.) $I^{(1,1,\ldots,1)}(F)$ is the union of the instance $I^{(1,1,\ldots,1)}_{(P,S)}(F)$, as specified below, with the instance $I^{(1,1,\ldots,1)}_{(R)}(F)$ that has the same set of tuples for the relation $R$ as the instance $MV$ has in the relation for $V$ (see the specification of $MV[V]$ above). As an illustration, in our running example,  the instance $I^{(1,1)}_{(R)}(F)$ has the same fourteen tuples as we saw above in the relation $MV[V]$, except that these same tuples are now in the relation $R$ in $I^{(1,1)}(F)$: 

\begin{tabbing} 
Tab he he ho \= ho \kill
$I^{(1,1)}_{(R)}(F)$ $=$ $\{ R(1,1,1,1)$, $R(1,1,0,1)$, $R(1,0,1,1)$, \\ 
\> $R(1,0,0,1)$, $R(0,1,1,1)$, $R(0,1,0,1)$, \\ 
\> $R(0,0,1,1)$, $R(0,0,0,2)$, $R(0,0,1,2)$, \\ 
\> $R(0,1,0,2)$, $R(0,1,1,2)$, $R(1,0,0,2)$, \\ 
\> $R(1,0,1,2)$, $R(1,1,0,2) \}$.  
\end{tabbing} 

For the general case, the instance $I^{(1,1,\ldots,1)}_{(P,S)}(F)$ is of the following form: 

\begin{tabbing} 
Tab me \= me m\= me crazy \kill
$I^{(1,1,\ldots,1)}_{(P,S)}(F) = \{ P(1,1), P(1,2), \ldots, P(1,m), S(1,1), $ \\ 
\> $S(0,1), S(1,2), S(0,2), \ldots, S(1,m), S(0,m) \}.$
\end{tabbing} 

%noindent 
As an illustration, for our running example, the instance $I^{(1,1)}_{(P,S)}(F)$ is as follows: 

\begin{tabbing} 
Tab me me m\= me crazy \kill
$I^{(1,1)}_{(P,S)}(F) = \{ P(1,1), P(1,2), S(1,1), S(0,1),$ \\ 
\> $S(1,2), S(0,2) \}.$
\end{tabbing} 

\noindent 
The entire instance $I^{(1,1)}(F)$ for our running example is the union of the instances $I^{(1,1)}_{(P,S)}(F)$ and $I^{(1,1)}_{(R)}(F)$ as given above. 

In the general case, intuitively, the tuples in the relation $S$ in $I^{(1,1,\ldots,1)}_{(P,S)}(F)$ mirror the instance $MV[U]$, by definition of the view $U$ and of $MV[U]$. The tuples $P(1,1)$, $P(1,2)$, $\ldots$, $P(1,m)$ in $I^{(1,1,\ldots,1)}_{(P,S)}(F)$ give us a key part of this proof, by representing a particular assignment of the truth values $1$ and $0$, one value to each $Y_j$ ($j$ $\in$ $[1,$ $m]$) among the variables $Y_1$, $\ldots$, $Y_m$ of the formula $F$. The particular assignment of the  truth values $1$ and $0$ to the variables $\bar Y$ of $F$ that is represented in the instance $I^{(1,1,\ldots,1)}(F)$ is the assignment of the truth value $1$ to each of the $m$ variables. We represent this fact in the name of the instance $I^{(1,1,\ldots,1)}(F)$, by using this assignment as an $m$-tuple $(1,1,\ldots,1)$ in the superscript in the name. 

%{\bf [[[ Stopped here on T07/23/13 ]]]} 

It is straightforward to verify that the ground instance $I^{(1,1,\ldots,1)}(F)$ satisfies $\cal V$ $\Rightarrow_{I^{(1,1,\ldots,1)}(F),\Sigma}$ $MV$. Further, it is straightforward to verify that there exist $2^m$ $-$ $1$ more ground instances of schema {\bf P}, as follows. Each such instance, denote it for now by $J$, differs from the instance $I^{(1,1,\ldots,1)}(F)$ only in whether the first argument of one or more $P$-tuple(s) in $J$ is the value $0$, instead of $1$ as it is in $I^{(1,1,\ldots,1)}(F)$. It is straightforward to verify that for each such instance $J$,  we have that $\cal V$ $\Rightarrow_{J,\Sigma}$ $MV$ in the context of the setting ${\cal M}\Sigma$ that we have constructed for the given formula $F$. Clearly, the total number of such instances $J$, including the instance $I^{(1,1,\ldots,1)}(F)$, is $2^m$, as the value $1$ of the first argument of $P(1,j)$ in $I^{(1,1,\ldots,1)}(F)$ can be flipped to $0$ independently for each $j$ $\in$ $[1,$ $m]$.

For each instance $J$ constructed as above, we name the instance using the same notation as for the instance $I^{(1,1,\ldots,1)}(F)$, by adorning the name $I(F)$ with an $m$-tuple $(a_1,a_2,\ldots,a_m)$ in the superscript, where $a_j$ for each $j$ $\in$ $[1,$ $m]$ is either $1$ or $0$, and ($a_j$) is precisely the first argument of the $P$-atom in the instance such that this $P$-atom has $j$ as its second argument. For example, the instance $I^{(0,0,\ldots,0)}(F)$ is the instance that differs from the instance $I^{(1,1,\ldots,1)}(F)$ only in that the first argument of each $P$-tuple in $J$ is the value $0$, instead of $1$ as it is in $I^{(1,1,\ldots,1)}(F)$. 

To the $2^m$ instances of the form $I^{(a_1,a_2,\ldots,a_m)}(F)$ as defined above, we refer collectively as {\em the core instances (of schema {\bf P}) for the input formula} $F$. In addition to these core instances, there is an infinite number of other ground instances of schema {\bf P}, where for each instance, $I$, we have that $\cal V$ $\Rightarrow_{I,\Sigma}$ $MV$ in the context of the setting ${\cal M}\Sigma$ that we have constructed for the given formula $F$. By definition of ${\cal M}\Sigma$, each such instance $I$ can be obtained by unioning one of the core instances for $F$, call the instance $K$, with a finite set of ground atoms of the form $P(c,d)$, where either (i) $d$ is a constant that is not in the set $[1,$ $m]$, or (ii) $d$ is in the set $[1,$ $m]$, while $c$ is a constant distinct from the value $g$ in the atom $P(g,d)$ in the ``core part'' $K$ of the instance $I$. 

For each ground instance $I$ that satisfies $\cal V$ $\Rightarrow_{I,\Sigma}$ $MV$ and by definition of the query $Q$ that we have constructed for the given formula $F$, we obtain the following useful observations. 

\begin{lemma} 
\label{pi-p-two-same-r-lemma} 
Given a 3-$CNF$ propositional formula $F$ %(with $m$ $\geq$ $1$) 
and the instance $({\cal M}\Sigma,Q,{\bar t})$ for $F$. Then, for all ground instances $I$ of schema {\bf P} such that $\cal V$ $\Rightarrow_{I,\Sigma}$ $MV$ in the context of the instance $({\cal M}\Sigma,Q,{\bar t})$, the relation $R$ is the same in all the instances $I$, and the relation $S$ is also the same in all the instances $I$. % Then there exists a core instance $K$ for $F$ such that there is an identity homomorphism from $K$ to $I$. %for each valuation, $\nu$, from $Q$ to $I$, the image of all the $P$-subgoals of $Q$ under $\nu$ is the relation for $P$ in one of the core instances for $F$. 
\end{lemma} 

%\noindent 
This result is immediate from the definitions of the views $U$ and $V$. The result of Lemma~\ref{pi-p-two-same-r-lemma} implies that in each ground instance $I$ of schema {\bf P} such that $\cal V$ $\Rightarrow_{I,\Sigma}$ $MV$, %in the context of the instance ${\cal D}_s(F)$, 
(i) the relation $R$ is the same in $I$ as the relation $R$ in the fixed instance $I^{(1,1,\ldots,1)}(F)$ defined above, and (ii) the relation $S$ is the same in $I$ as the relation $S$ in $I^{(1,1,\ldots,1)}(F)$. 

\begin{lemma} 
\label{pi-p-at-least-m-p-lemma} 
Given a 3-$CNF$ propositional formula $F$ %(with $m$ $\geq$ $1$) 
and the instance $({\cal M}\Sigma,Q,{\bar t})$ for $F$. Then, for each ground instance $I$ of schema {\bf P} such that $\cal V$ $\Rightarrow_{I,\Sigma}$ $MV$ in the context of the instance $({\cal M}\Sigma,Q,{\bar t})$, the relation $P$ in $I$ has for each $j$ $\in$ $[1,$ $m]$ an atom of the form $P(e_j,j)$, with the constant $e_j$ $\in$ $\{ 1, 0 \}$.  %is the same in all the instances $I$, and the relation $S$ is also the same in all the instances $I$. % Then there exists a core instance $K$ for $F$ such that there is an identity homomorphism from $K$ to $I$. %for each valuation, $\nu$, from $Q$ to $I$, the image of all the $P$-subgoals of $Q$ under $\nu$ is the relation for $P$ in one of the core instances for $F$. 
\end{lemma} 

\begin{proof} 
This result is immediate from the definitions of the views $U$ and $W$ and from the specifications of the relations $MV[U]$ and $MV[W]$ in ${\cal M}\Sigma$ as constructed for the given formula $F$. Indeed, recall that $MV[W]$ $=$ $\{ (1), (2), \ldots, (m) \}$. By definition of the view $W$, for each $j$ $\in$ $[1$, $m]$ we have that the tuple $(j)$ in $MV[W]$ witnesses, in each ground instance $I$ of schema {\bf P} such that ${\cal V}$ $\Rightarrow_{I,\Sigma}$ $MV$, the presence of a ground $P$-atom $P(e_j,j)$ and of a ground $S$-atom $S(e_j,j)$ with the same arguments. Now consider an arbitrary ground atom $S(d,f)$ in any ground instance $I$ of schema {\bf P} such that ${\cal V}$ $\Rightarrow_{I,\Sigma}$ $MV$. By definition of the view $U$ and by the contents of the relation $MV[U]$, the value $d$ in this atom $S(d,f)$ must be one of $1$ and $0$, and the value $f$ in $S(d,f)$ must belong to the set $[1$, $m]$. Thus, via the definition of the view $W$ and the contents of the relation $MV[W]$ as discussed above, we obtain that each ground atom of the form $P(e_j,j)$ as above must have its value $e_j$ restricted to one of $1$ and $0$. The claim of the lemma follows. 
\end{proof}

\begin{lemma} 
\label{pi-p-two-homomorphism-lemma} 
Given a 3-$CNF$ propositional formula $F$ %(with $m$ $\geq$ $1$) 
and the instance $({\cal M}\Sigma,Q,{\bar t})$ for $F$. Let $I$ be a ground instance of schema {\bf P} such that $\cal V$ $\Rightarrow_{I,\Sigma}$ $MV$ in the context of the instance $({\cal M}\Sigma,Q,{\bar t})$. Then there exists a core instance $K$ for $F$ such that there is an identity homomorphism from $K$ to $I$. %for each valuation, $\nu$, from $Q$ to $I$, the image of all the $P$-subgoals of $Q$ under $\nu$ is the relation for $P$ in one of the core instances for $F$. 
\end{lemma} 

\begin{lemma} 
\label{pi-p-two-valuation-lemma} 
Given a 3-$CNF$ propositional formula $F$ %(with $m$ $\geq$ $1$) 
and the instance $({\cal M}\Sigma,Q,{\bar t})$ for $F$. Let $I$ be a ground instance of schema {\bf P} such that $\cal V$ $\Rightarrow_{I,\Sigma}$ $MV$ in the context of the instance $({\cal M}\Sigma,Q,{\bar t})$. Then for each valuation, $\mu$, from $Q$ to $I$, the image of all the $P$-subgoals of $Q$ under $\mu$ is the relation for $P$ in one of the core instances for $F$. 
\end{lemma} 

%\noindent 
(Toward the proof of Lemma~\ref{pi-p-two-valuation-lemma}, recall that  by design of the query $Q$, the first argument of each $P$-subgoal of the query $Q$ must also occur as one of the first three arguments of at least one $R$-subgoal of $Q$.) 

By the above lemmae and by the structure of all the ground instances $I$ of schema {\bf P} such that $\cal V$ $\Rightarrow_{I,\Sigma}$ $MV$, it must be that, in the context of the instance $({\cal M}\Sigma,Q,{\bar t})$: 

\mbox{} 

\noindent 
(*) For each $I$ such that $\cal V$ $\Rightarrow_{I,\Sigma}$ $MV$ and for each valuation, $\mu$, from the query $Q$ to $I$, the image $\mu(body_{(Q)})$ of the body of the query $Q$ under $\mu$ is a subset of one of the core instances for $F$. %, call any one such core instance $K$. 

\mbox{} 

%\noindent 
Specifically, by construction of the query $Q$, for any such core instance $K$ for $F$, all the ground atoms in the relation $P$ in $K$ must be present in the set $\mu(body_{(Q)})$. % are all the ground atoms in the relation $S$ in $K$. %$S(b_1,1)$, $S(b_2,2)$, $\ldots$, $S(b_m,m)$, such that for each $j$ $\in$ $[1,$ $m]$ and for the atom $P(a_j,j)$ in $\mu(body_{(Q)})$, we have that $a_j$ $\neq$ $b_j$. 
(Recall that in each core instance, $K$, for the input formula $F$, the relation $P$ intuitively represents exactly one specific assignment of values $1$ an $0$ to the $m$ $\geq$ $1$ variables $\bar Y$ of the formula $F$. Further, each specific assignment of values $1$ an $0$ to the $m$ $\geq$ $1$ variables $\bar Y$ of the formula $F$ is represented by a separate core instance for $F$.) 

As an illustration, consider the query $Q$ of our running example and the instance $I^{(1,1)}(F)$ of schema {\bf P} for that example. (Both the query $Q$ and the instance $I^{(1,1)}(F)$ for this running example have already been given in this proof.) Consider a mapping $\mu$ $=$ $\{$ $X_1$ $\rightarrow$ $1$, $X_2$ $\rightarrow$ $0$, $Y_1$ $\rightarrow$ $1$, $Y_2$ $\rightarrow$ $1$ $\}$. We can show that $\mu$ is a valuation from the query $Q$ to the instance $I^{(1,1)}(F)$. The image $\mu(body_{(Q)})$ of the body of the query $Q$ under the valuation $\mu$ includes all the ground atoms in the relation $P$ in the instance $I^{(1,1)}(F)$, that is both atoms $P(1,1)$ and $P(1,2)$  in $I^{(1,1)}(F)$. 

We now proceed to show that for the input formula $F$ and for the corresponding instance $({\cal M}\Sigma,Q,{\bar t})$ constructed for $F$ as above, the following two statements are equivalent: 
\begin{itemize} 
	\item[(I)] For each assignment of truth values $1$ and $0$ to the variables $\bar Y$ in the formula $F$, there exists an assignment of truth values $1$ and $0$ to the variables $\bar X$ in $F$ such that $F$ is true under these assignments; and  
	\item[(II)] The tuple $()$ is an answer to the query $Q$ on all the ground instances $I$ of schema {\bf P} such that $\cal V$ $\Rightarrow_{I,\Sigma}$ $MV$ in the context of the instance $({\cal M}\Sigma,Q,{\bar t})$. 
\end{itemize} 

%\noindent 
Note that by the definition in Section~\ref{probl-stmt-defs-sec}, the statement (II) says that the tuple $()$ is a certain-answer tuple for the query $Q$ w.r.t. the setting ${\cal M}\Sigma$. Thus, once we show the equivalence of the statements (I) and (II), our proof of $\Pi^p_2$ hardness of the certain-query-answer problem for CQ weakly acyclic inputs with $\Sigma$ $=$ $\emptyset$ will be complete. 

We begin the proof of the equivalence of the statements (I) and (II) by making the following observation. Denote by ${\cal R}_{(Q)}$ the conjunction of the $R$-subgoals of the query $Q$ in  the instance $({\cal M}\Sigma,Q,{\bar t})$ for the given formula $F$. Further, denote by ${\cal P}_{(Q)}$ the conjunction of the $P$-subgoals of the query $Q$ in  the instance $({\cal M}\Sigma,Q,{\bar t})$. Consider any ground instance $I$ of schema {\bf P} such that $\cal V$ $\Rightarrow_{I,\Sigma}$ $MV$ in the context of the instance $({\cal M}\Sigma,Q,{\bar t})$. Let $\nu$ be any mapping of the set $\bar Z$ of the variables ($\bar X$ and $\bar Y$) of the formula $F$ into the set $\{ 0, 1 \}$. Similarly to the argument in \cite{AhoSU79} (as also used in the proof of Theorem 3.3 in \cite{MillsteinHF03}), we can show that any such mapping $\nu$ is a satisfying assignment for the formula $F$ if and only if $\nu({\cal R}_{(Q)})$ is a subset of the instance $I$. By Lemma~\ref{pi-p-two-same-r-lemma}, we have that any such mapping $\nu$ is a satisfying assignment for the formula $F$ if and only if $\nu({\cal R}_{(Q)})$ is a subset of {\em each} ground instance $J$ of schema {\bf P} such that $\cal V$ $\Rightarrow_{J,\Sigma}$ $MV$ in the context of the instance $({\cal M}\Sigma,Q,{\bar t})$. From the above reasoning and from Lemmae~\ref{pi-p-at-least-m-p-lemma}--\ref{pi-p-two-homomorphism-lemma}, we obtain the following result of Lemma~\ref{five-lemma}. 

We first introduce some notation. In the remainder of the proof of Theorem~\ref{complexity-hardness-thm}, let $\mu^{(a_1,\ldots,a_m)}_{(\bar Y)}$, with $a_j$ $\in$ $\{ 0,1 \}$ for each $j$ $\in$ $[1,$ $m]$, denote the mapping from the set of variables $\bar Y$ of the formula $F$ to the set $\{ 0,1 \}$, such that $\mu^{(a_1,\ldots,a_m)}_{(\bar Y)}(Y_j)$ $=$ $a_j$ for each $j$ $\in$ $[1,$ $m]$. Further, let $\mu_{(\bar X)}$ denote a mapping from the set of variables $\bar X$ of the formula $F$ to the set $\{ 0,1 \}$. 

\begin{lemma} 
\label{five-lemma} 
Given a 3-$CNF$ formula $F$ with $m$ $\geq$ $1$ variables $\bar Y$ and with $n$ $\geq$ $0$ variables $\bar X$, let $(a_1,\ldots,a_m)$ be an arbitrary $m$-tuple such that $a_j$ $\in$ $\{ 0,1 \}$ for each $j$ $\in$ $[1,$ $m]$. Let $\mu_{(\bar X)}$ be an arbitrary mapping from the set of variables $\bar X$ of the formula $F$ to the set $\{ 0,1 \}$. Let $({\cal M}\Sigma,Q,{\bar t})$ be  the instance that we have constructed for the formula $F$ as above. Finally, let $I$ be a ground instance of the schema {\bf P}, such that $\cal V$ $\Rightarrow_{I,\Sigma}$ $MV$ in the context of the instance $({\cal M}\Sigma,Q,{\bar t})$, and such that the set $\{ P(a_1,1),\ldots,P(a_m,m) \}$  is a subset of the instance $I$. 

Then the following two statements are equivalent: 
\begin{itemize} 

	\item The assignment $\mu^{(a_1,\ldots,a_m)}_{(\bar Y)}$ $\cup$ $\mu_{(\bar X)}$ of the variables of the formula $F$ to elements of the set $\{ 0,1 \}$ is a satisfying assignment for the formula $F$; and 

	\item The empty tuple $()$ is in the answer to the query $Q$ on the instance $I$ due to the valuation $\mu^{(a_1,\ldots,a_m)}_{(\bar Y)}$ $\cup$ $\mu_{(\bar X)}$ from $body_{(Q)}$ to $I$. 

\end{itemize} 
\vspace{-0.5cm} 
\end{lemma}  

We are now ready to show the equivalence of the statements (I) and (II) as formulated above. 

{\em (I) $\rightarrow$ (II):} Suppose that for each $m$-tuple  $(a_1,\ldots,a_m)$, such that $a_j$ $\in$ $\{ 0,1 \}$ for each $j$ $\in$ $[1,$ $m]$, we have that there exists a mapping $\mu_{(\bar X)}$ from the set of variables $\bar X$ of the formula $F$ to the set $\{ 0,1 \}$, such that 

\begin{tabbing} 
$\mu^{(a_1,\ldots,a_m)}_{(\bar Y)}$ $\cup$ $\mu_{(\bar X)}$ 
\end{tabbing} 

\noindent 
is a satisfying assignment for the formula $F$. Fix an arbitrary ground instance $I$ of the schema {\bf P} such that $\cal V$ $\Rightarrow_{I,\Sigma}$ $MV$. Then, by Lemmae~\ref{pi-p-at-least-m-p-lemma} and~\ref{five-lemma}, the empty tuple $()$ is in the relation $Q(I)$. 

{\em (II) $\rightarrow$ (I):} Consider the set $\cal K$ of the $2^m$ core instances (of schema {\bf P}) for the formula $F$. By construction of the set $\cal K$,  for each $m$-tuple  $(a_1,\ldots,a_m)$ such that $a_j$ $\in$ $\{ 0,1 \}$ for each $j$ $\in$ $[1,$ $m]$, there exists an instance  $K$ $\in$ $\cal K$ such that the set $\{ P(a_1,1),\ldots,P(a_m,m) \}$  is a subset of the instance $K$, and the relation $K[P]$ has {\em no other tuples.}  

Fix an arbitrary instance $K$ $\in$ $\cal K$; the relation $K[P]$ $=$ $\{ P(a_1,1),$ $\ldots,$ $P(a_m,m) \}$ specifies a particular $m$-tuple  $(a_1,\ldots,a_m)$ such that $a_j$ $\in$ $\{ 0,1 \}$ for each $j$ $\in$ $[1,$ $m]$. By our assumption (II), there exists a mapping $\mu_{(\bar X)}$ from the set of variables $\bar X$ of the formula\footnote{Recall that the set $\bar Z$ $=$ $\bar X$ $\cup$ $\bar Y$ is the set of all variables of the formula $F$, and is also the set of all variables of the query $Q$ in $({\cal M}\Sigma,Q,{\bar t})$.} $F$ to the set $\{ 0,1 \}$, such that 

\begin{tabbing} 
$\mu(K)$ $=$ $\mu^{(a_1,\ldots,a_m)}_{(\bar Y)}$ $\cup$ $\mu_{(\bar X)}$ 
\end{tabbing} 

\noindent 
is a valuation from the query $Q$ to the instance $K$ that produces the empty tuple $()$ in the relation $Q(K)$. Thus, by  Lemma~\ref{five-lemma}, the mapping $\mu(K)$ of the variables of the formula $F$ to the set $\{ 0,1 \}$ is a satisfying assignment for the formula $F$. The claim of (I) follows from the observation (made above) that for each $m$-tuple  $(a_1,\ldots,a_m)$ such that $a_j$ $\in$ $\{ 0,1 \}$ for each $j$ $\in$ $[1,$ $m]$, there exists an instance  $K$ $\in$ $\cal K$ such that the set $\{ P(a_1,1),\ldots,P(a_m,m) \}$  is in the instance $K$. This completes the proof of Theorem~\ref{complexity-hardness-thm}. 
%
%We now outline how to modify the main body of the proof, given above, so that it would also go through for each input instance with $m$ $=$ $0$. First, we replace $MV[W]$ by the relation $MV[W]$ $=$ $\{$ $(0)$ $\}$. Second, we replace the conjunction of the $P$-subgoals in the query $Q$ by a single subgoal $P(Y_0,0)$, for some variable $Y_0$ that is not used in the formula $F$. (Note that for the first argument $Y_0$ of the subgoal $P(Y_0,0)$ of the query $Q$, we have, unlike the case $m$ $\geq$ $1$, that $Y_0$ is {\em not} in the conjunction of the $R$-subgoals of the query $Q$.) Third, we replace $MV[U]$ by $MV[U]$ $=$ $\{$ $U(1,0)$, $U(0,0)$ $\}$. Fourth, we use $I^{(1)}$ and $I^{(0)}$ and as the only two ``core'' instances $I$. Here,  $I^{(1)}$ has $P$ $=$ $\{ P(1,0) \}$, and $I^{(0)}$ has $P$ $=$ $\{ P(0,0) \}$. The other relations ($R$ and $S$) in each of $I^{(1)}$ and $I^{(0)}$ mirror $R$ and $S$ in the core relations for $m$ $\geq$ $1$. With these modifications, the proof above goes through for all the input formulae $F$ with $m$ $=$ $0$. 
\end{proof}

\newpage 

% end \nop ugly-doll 

} % end \mysize blue lagoon 

\end{document}